\documentclass[a4paper,12pt]{amsart}

\usepackage[english]{babel}
\usepackage{amsmath}
\usepackage{amssymb}
\usepackage{url}
\usepackage{amsthm}
\usepackage{amsxtra}
\usepackage{mathrsfs}
\usepackage{fancyhdr}
\usepackage{hyperref}
\usepackage{verbatim}

\newtheorem{theorem}{Theorem}
\newtheorem{lemma}{Lemma}[section]
\newtheorem{proposition}[lemma]{Proposition}

\theoremstyle{definition}

\numberwithin{equation}{section}

\newcommand{\RE}{\mathbb R}
\newcommand{\CO}{\mathbb C}
\newcommand{\NA}{\mathbb N}
\newcommand{\ID}{\mathbb I}
\newcommand{\E}{\mathbb E}

\newcommand{\EE}{\mathcal{E}}

\newcommand{\LL}{\mathcal{L}}
\newcommand{\MM}{\mathcal{M}}
\newcommand{\NN}{\mathcal{N}}
\newcommand{\PP}{\mathcal{P}}

\newcommand{\II}{{\bf 1}}
\newcommand{\wt}{\widetilde}

\newcommand{\supp}{\operatorname{supp}\,}

\newcommand{\Tr}{\operatorname{Tr}}
\newcommand{\ve}{\varepsilon}
\newcommand{\al}{\alpha}

\newcommand{\ga}{\gamma}
\newcommand{\La}{\Lambda}
\newcommand{\la}{\lambda}
\newcommand{\de}{\delta}

\renewcommand{\aa}{{\bf a}}

\newcommand{\kk}{{\bf k}}

\newcommand{\uu}{{\bf u}}

\renewcommand{\Im}{\operatorname{Im}\,}
\renewcommand{\Re}{\operatorname{Re}\,}

\renewcommand{\P}{\mathbb{P}}
\newcommand{\fP}{\mathfrak{P}}
\newcommand{\g}{{\bf b}}
\newcommand{\T}{\mathbb T}
\newcommand{\dd}{\delta}
\newcommand{\od}{\gamma}

\newcommand{\SR}{\mathcal{SR}}
\newcommand{\bsigma}{\rho}

\newcommand{\tildeeta}{\tilde\eta}

\title[]{Bounds for the Stieltjes Transform and the Density of States of Wigner Matrices}

\author[C.~Cacciapuoti]{Claudio Cacciapuoti}
\address{Hausdorff Center for Mathematics\\
Institute for Applied Mathematics, University of Bonn\\
Endenicher Allee 60, 53115 Bonn, Germany}
\curraddr{Dipartimento di Scienza e Alta Tecnologia, Universit\`a dell'Insubria, Via Valleggio 11, 22100 Como, Italy}
\email{claudio.cacciapuoti@uninsubria.it}

\author[A.~Maltsev]{Anna Maltsev}
\address{Department of Mathematics, University of Bristol\\
University Walk, Clifton, Bristol BS8 1TW, U.K.}
\email{annavmaltsev@gmail.com}

\author[B. Schlein]{Benjamin Schlein}
\address{Hausdorff Center for Mathematics\\
Institute for Applied Mathematics, University of Bonn\\
Endenicher Allee 60, 53115 Bonn, Germany}
\email{benjamin.schlein@hcm.uni-bonn.de}


\date{\today}

\begin{document}

\begin{abstract}
We consider ensembles of Wigner matrices, whose entries are (up to the symmetry constraints) independent and identically distributed random variables. We show the convergence of the Stieltjes transform towards the Stieltjes transform of the semicircle law on optimal scales and with the optimal rate. Our bounds improve previous results, in particular from \cite{EYY,ekyy-rxv12}, by removing the logarithmic corrections. As applications, we establish the convergence of the eigenvalue counting functions with the rate $(\log N)/N$ and the rigidity of the eigenvalues of Wigner matrices on the same scale. These bounds improve the results of \cite{EYY,ekyy-rxv12,GT}.
\end{abstract}

\maketitle

\section{Introduction and main results}

Let $H$ be an $N \times N$ Wigner matrix, whose entries are independent (up to the symmetry constraints) and identically distributed random variables with zero mean and a fixed variance.
We are interested in the statistical properties of the eigenvalues of $H$ in the limit of large $N$. As already shown by Wigner in \cite{W}, the density of the eigenvalues of $H$ converges, after appropriate normalization, towards the famous semicircle law
\[ \rho_{sc} (x) = \left\{ \begin{array}{ll} \frac{1}{\pi} \sqrt{1- \frac{x^2}{4}} \quad & \text{if } |x| \leq 2 \\
0 \quad & \text{if } |x|  > 2 \end{array}\right. \]
Wigner's result concerns the density of states of $H$ on intervals with length of order one, containing typically order $N$ eigenvalues (we normalize Wigner matrices so that the typical distance between eigenvalues in the bulk of the spectrum is of the order $1/N$). It is natural to ask what happens on smaller intervals, in which the typical number of eigenvalues is large, but not macroscopically large. This question was first addressed in \cite{ESY1,ESY2,ESY3}, where the density of states of $H$ was proven to converge towards the semicircle law on microscopic intervals, containing a large but fixed number of eigenvalues (independent of $N$). In order to establish the convergence towards the semicircle law, it is useful to study the Stieltjes transform $m(z)$ of the Wigner matrix $H$, defined for $\text{Im z} > 0$ by
\begin{equation}\label{eq:mdef} m (z) = \frac{1}{N} \Tr \; \frac{1}{H-z} = \frac{1}{N} \sum_{\alpha=1}^N \frac{1}{\lambda_\alpha - z} \end{equation}
where $\lambda_\alpha$, $\alpha=1,\dots , N$, denote the eigenvalues of $H$. We observe that
\begin{equation}\label{eq:Imm} \text{Im } m (E+i\eta) = \frac{1}{N} \sum_{\alpha=1}^N \frac{\eta}{(\lambda_\alpha - E)^2 + \eta^2} \, . \end{equation}
{F}rom (\ref{eq:Imm}) it is clear that the density of states in an interval of size $\eta$ around $E$ is closely related to the imaginary part
of the Stieltjes transform at the point $z = E+i\eta$. To get information on the density of states on a small scale of order $\eta$, we need to control the Stieltjes transform $m$ at distances of order $\eta$ from the real axis. In the limit of large $N$, $m(z)$ approaches the Stieltjes transform of the semicircle law, given for all $x \in \CO \backslash [-2;2]$ by
\begin{equation}\label{eq:msc} m_{sc} (z) = \int dx \, \frac{\rho_{sc}(x)}{x-z} \,. \end{equation}
There are many reasons for analyzing the Stieltjes transform instead of looking directly at the density of states. Since every eigenvalue contributes to $m$, it seems that $m$ should be more stable with respect to small fluctuations of the eigenvalues. More importantly, the Stieltjes transform $m_{sc}$ of the semicircle law satisfies the simple fixed point equation
\begin{equation}\label{eq:fix} m_{sc} (z) + \frac{1}{z + m_{sc}(z)} = 0\,. \end{equation}
As a consequence, to prove that $m$ is close to $m_{sc}$, it is enough to show that $m$ is an approximate solution of (\ref{eq:fix}). Using this strategy and assuming the entries of $H$ to have subgaussian tails, it was shown in \cite{ESY3} (improving previous results from \cite{ESY1,ESY2}) that, for every $\kappa > 0$, there exist constants $C,c > 0$ such that
\begin{equation}\label{eq:PESY3} \P \left( |m(E+i\eta) - m_{sc}(E+i\eta)| \geq \delta \right) \leq C e^{-c\delta \sqrt{N\eta}} \end{equation}
for all $|E| \leq 2 - \kappa$ (i.e. for $E$ in the bulk of the spectrum). The estimate (\ref{eq:PESY3}) establishes the convergence of the Stieltjes transform $m(E+i\eta)$ for all $\eta \gg 1/N$ and implies therefore the convergence of the density of states on all intervals of size $\eta \gg 1/N$. It is important to observe that the scale $\eta \simeq 1/N$ is optimal. For $\eta \lesssim 1/N$, the number of eigenvalues contained in an interval of size $\eta$ is small and therefore the fluctuations of the density of states cannot be neglected. It is clear therefore that on scales $\eta \ll 1/N$ it is impossible to have convergence in probability for the Stieltjes transform and for the density of states (one has, however, convergence in expectation of the density of states, if the probability density of the matrix entries is sufficiently regular; see \cite{MS}). While (\ref{eq:PESY3}) is optimal in the sense that it gives convergence for all $\eta \gg1/N$, it turns out that it is not optimal in the size of the fluctuations around $m_{sc}$. {F}rom (\ref{eq:PESY3}), fluctuations are at most of the order $(N\eta)^{-1/2}$. This bound was substantially improved in \cite{EYY}, where the authors showed that, with high probability,
\begin{equation}\label{eq:rate}  |m(E+i\eta) -m_{sc} (E+i\eta)| \lesssim \frac{(\log N)^c}{N\eta}  \end{equation}
uniformly in $E$. In fact, the results of \cite{EYY} also apply to generalized Wigner matrices, whose entries are not required to be identically distributed. Hence, up to logarithmic corrections, the fluctuations around $m_{sc}$ are of the order $(N\eta)^{-1}$. An important observation which was used in \cite{EYY} to show (\ref{eq:rate}) is the fact that not only the Stieltjes transform,
%\[ m(z) = \frac{1}{N} \Tr \, \frac{1}{H-z} = \frac{1}{N} \sum_{j=1} \frac{1}{H-z} (j,j) \]
given by the average of the diagonal entries $G_{jj} (z) = (H-z)^{-1}$ of the resolvent, converges
towards $m_{sc}$, but that also each $G_{jj} (z)$ approaches the same limit. Making use of a vectorial fixed point equations for the diagonal entries of the resolvent, it was proved in \cite{EYY} that
\begin{equation}\label{eq:Gii-con} | G_{jj} (E+i\eta) - m_{sc}(E+i\eta)| \lesssim \frac{(\log N)^c}{\sqrt{N\eta}}\,. \end{equation}
For the diagonal entries $G_{jj}$ of the resolvent, the rate predicted by (\ref{eq:Gii-con}) turns out to be optimal, up to the logarithmic corrections. In order to improve the bound (\ref{eq:Gii-con}) to the stronger rate on the r.h.s. of (\ref{eq:rate}), two important observations are needed (which go back to \cite{EYY}); first of all, the expectation of $G_{jj} - m_{sc}$ turns out to be of order $(N\eta)^{-1}$, hence much smaller than the typical size of $|G_{jj} - m_{sc}|$. On the other hand, the different resolvent entries are only weakly dependent on each other. As a consequence, the average
\[ m (E+i\eta) - m_{sc} (E+i\eta) = \frac{1}{N} \sum_{j=1}^N  \left(G_{jj} (E+i\eta) - m_{sc} (E+i\eta) \right)   \] exhibits substantial cancellations, which, similarly to what happens in the central limit theorem, make its typical size much smaller than that of the single summands (in fact the argument used in \cite{EYY} is slightly different, controlling $m-m_{sc}$ in terms of the average of certain error terms $Z_j$ which are not exactly the same as $G_{jj} - m_{sc}$; while each $Z_j$ has a typical size of order $(N\eta)^{-1/2}$, the average is much smaller because they are only weakly dependent random variables). To establish the validity of this picture, terms arising in high moments of $m-m_{sc}$ are organized in \cite{EYY} in contributions which are either small or independent of an increasing number of variables. More recently, a simpler expansion algorithm leading to bounds of the form (\ref{eq:rate}) has been developed in \cite{ekyy-rxv12}; this paper is the basis for our analysis.

\medskip

Our goal is to obtain estimates of the form (\ref{eq:rate}) but without logarithmic corrections. To simplify the analysis, we restrict our attention to Wigner matrices with identically distributed entries. We believe, however, that similar arguments would also work for generalized Wigner matrices, as considered in \cite{EYY, ekyy-rxv12}. The bounds in Theorem \ref{t:main} below establish the convergence of the imaginary part of the Stieltjes transform of Wigner matrices towards the imaginary part of the Stieltjes transform of the semicircle law on the optimal
scale $\eta \simeq 1/N$ and with the optimal rate $(N\eta)^{-1}$, for arbitrary energies $E \in \RE$. For energies inside the spectrum of $H$, Theorem \ref{t:main} also gives convergence of the real part of the Stieltjes transform on the optimal scale and with the optimal rate; more precisely, it shows that
\[ \left| m (E+i\eta) - m_{sc} (E+i\eta) \right| \lesssim \frac{1}{N\eta} \]
for all $|E| \leq 2 + \eta$. The reason why we cannot prove the convergence of the real part of the Stieltjes transform on the optimal scale and with the optimal rate outside the spectrum of $H$ (i.e. for $|E| > 2 + \eta$) is an instability of the equation for the difference $\Lambda (E+i\eta) = m (E+i\eta) - m_{sc} (E+i\eta)$; we will discuss this point after Proposition \ref{p:lambda}. If we assume $E$ to be well inside the bulk of the spectrum of $H$ (i.e. $|E| \leq 2-\kappa$, for an arbitrary but fixed $\kappa > 0$) then Theorem \ref{t:bulk} strengthens  the result of Theorem \ref{t:main}, showing a faster decay of the probability. To prove Theorem \ref{t:main} and Theorem \ref{t:bulk}, we follow the general strategy of \cite{EYY,ekyy-rxv12}. Since we do not allow logarithmic corrections, however, we need to modify several steps of the proof and, in particular, of the expansion algorithm. Since our small probabilities are independent of $N$, we cannot estimate the probability of unions of many bad events with the trivial union bound. Compared with \cite{EYY,ekyy-rxv12}, we have to work with high moments, instead of taking probabilities; similar ideas were used in \cite{ESY3} to prove the convergence towards the semicircle law on the optimal scale (but, in that case, not with the optimal rate).

\vspace{.3cm}

As an application of the optimal convergence of the Stieltjes transform, we obtain in Theorem~\ref{t:counting} bounds on the rate of convergence of the eigenvalue counting function $n (E) = N^{-1} |\{ \alpha =1 , \dots , N  : \lambda_\alpha \leq E \}|$ and of the density of states. More precisely we show that
\begin{equation}\label{eq:nnsc} |n(E) - n_{sc} (E)| \lesssim \frac{\log N}{N} \end{equation}
for all $E \in \RE$. Notice that, to prove (\ref{eq:nnsc}), convergence of the imaginary part of $m(E+i\eta)$ is not enough, we really need convergence of the full Stieltjes transform, including its real part. Fortunately, Theorem \ref{t:main} does imply convergence of $\text{Re } m (E+i\eta)$, for $|E| \leq 2$. This allows us to prove that, with high probability, there are at most $K \log N$ eigenvalues outside the interval $[-2;2]$; hence, we can establish convergence of the eigenvalue counting function $n(E)$ for all $E \in \RE$.

\vspace{.3cm}

Another application of the convergence of the Stieltjes transform is the rigidity of the eigenvalues of Wigner matrices. In Theorem \ref{t:rigidity} we show that, with high probability, the distance between an eigenvalue $\lambda_\alpha$ of the Wigner matrix $H$ and the location $\gamma_\alpha$ of the same eigenvalue as predicted by the semicircle law is smaller than $(\log N)/N$. The results that we obtain in Theorem \ref{t:counting} and in Theorem \ref{t:rigidity} on the eigenvalue counting function, the density of states and the rigidity of the eigenvalues improve similar results from \cite{EYY,ekyy-rxv12} by removing all logarithmic factors but one. While it is known that these bounds cannot hold true without logarithmic corrections, one expects eigenvalues to fluctuate on the scale $\sqrt{\log N} / N$, rather than on the scale $(\log N) / N$ obtained in Theorems \ref{t:counting} and \ref{t:rigidity} (for GUE and for hermitian ensembles whose entries have four moments matching GUE, this follows from the results of \cite{Gu} and, respectively, \cite{tv09}, where the eigenvalues are shown to be, asymptotically, gaussian). We remark, moreover, that bounds for the rate of convergence of the eigenvalue counting function, as well as rigidity estimates for the eigenvalues of Wigner matrices, have been recently obtained in \cite{GT}. Roughly speaking, the authors of \cite{GT} show with methods which are different from the ones of \cite{EYY,ekyy-rxv12} and of the present paper, that, with high probability, $|n (E) - n_{sc} (E)| \lesssim (\log N)^4 / N$.

\vspace{.3cm}

In the last years, there has been a lot of progress in the mathematical understanding of the statistical properties of the eigenvalues of Wigner matrices. As indicated above, the validity of the semicircle law for the density of states on microscopic scales was established in \cite{ESY1,ESY2,ESY3}. In the same works, the local semicircle law was applied to show the complete delocalization of the eigenvectors of Wigner matrices and to prove the repulsion among the eigenvalues. In \cite{EPRSY}, the validity of the semicircle law on microscopic intervals was used (in combination with \cite{Jo}) to prove the (bulk) universality of the local eigenvalue statistics of hermitian Wigner matrices; independently of the choice of the law of the entries, the eigenvalue correlation functions always converge towards the same Wigner-Dyson distribution observed for the Gaussian Unitary Ensemble.  At the same time, a different proof of universality for hermitian Wigner matrices was obtained in \cite{tv09} (the two results have been combined in \cite{ERSTVY}). Soon after, a new argument, based on the introduction of a relaxation flow approximating Dyson Brownian motion, was introduced in \cite{ESY6} to prove bulk universality of Wigner matrices with arbitrary symmetry. In all these proofs of universality, the local convergence of the density of states was a crucial ingredient. In \cite{EYY1},  local convergence towards the semicircle law and universality were extended to generalized Wigner matrices. Universality at the edge of the spectrum was proven in \cite{S} and more recently in \cite{tv10cmp,EYY}; a necessary and sufficient condition on the decay of the matrix entries to obtain edge universality has been proven in \cite{ly}. For sample covariance matrices, local convergence towards the Marchenko-Pastur law \cite{MP67} and universality of the local eigenvalue statistics were determined in the bulk \cite{ESYY,tv-rxv09} and at the soft edge \cite{WangSoftEdge,py12}. Local convergence towards the Marchenko-Pastur law at the hard edge of the spectrum was proven in \cite{CMS} and, implicitly, in \cite{BYY,TV4}. The local convergence towards the circular law for the density of states of non-symmetric random matrices with independent entries was established in \cite{BYY,TV4,BYY2,Y}. We notice that the local convergence of the density of states and delocalization of eigenvectors have also been obtained, in the last years, for more structured ensembles, such as the adjacency matrices of Erd{\H o}s-R{\'e}nyi graphs \cite{spa1, spa2} and band matrices \cite{ban,ban2,alt1,alt2}. A nice overview on these and other results can be found in \cite{E}.

\vspace{.3cm}

Next, we present our results in more details. Let $H = (h_{\ell j} )$ be an $N \times N$ hermitian Wigner matrix, with entries $h_{jj} = x_{jj} / \sqrt{N}$ on the diagonal and
\[ h_{\ell j} = \frac{1}{\sqrt{N}} (x_{\ell j} + i y_{\ell j}) \] for all $\ell < j$ where $x_{jj}, x_{\ell j}, y_{\ell j}$ are independent random variables. The off-diagonal variables $x_{\ell j}, y_{\ell j}$ are identically distributed with $\E \, x_{\ell j} = 0$ and $\E \, x_{\ell j}^2 = 1/2$. The diagonal entries $x_{ii}$ are also identically distributed with $\E \, x_{ii} =0$ and $\E x_{ii}^2 = 1$ (in fact the variance of the diagonal entries is not important, it should only be finite). We assume that the common distribution $\nu$ of $x_{\ell j}, x_{jj}$ and $y_{\ell j}$ has subgaussian decay, i.e., there exists $\de_0>0$ such that
\begin{equation}
\label{e:gd}
\int_\RE e^{\de_0 x^2} d\nu(x) < \infty\,.
\end{equation}
This implies that  $\E|x_{ij}|^{2q}\leq (Cq)^q$ for any $q\geq1$. We define the Stieltjes transform $m(z)$ of $H$ by (\ref{eq:mdef}) and the Stieltjes transform $m_{sc} (z)$ of the semicircle law by (\ref{eq:msc}).
%
%\medskip
%
%For any $z\in\CO\backslash [-2,2]$ we denote by $m_{sc}$ the Stieltjes transform of the semicircle %law defined by
%\begin{equation}\label{eq:msc}
%m_{sc}(z) = \int \frac{1}{x-z} \rho_{sc}(x) dx
%\end{equation}
%with $\rho_{sc}(x) = \frac1\pi \sqrt{\left[1-\frac{x^2}4\right]_+}$.  We recall that $m_{sc}$ is the%solution of the equation
%\begin{equation}
%m_{sc} (z) = - \frac{1}{z+m_{sc} (z)}
%\end{equation}
%such that $\Im m_{sc}(z) >0$ for $\Im z >0$ and $\Im m_{sc}(E+i\eta) <0$ for $\Im z <0$.
%
%\medskip
%
%For any $z\in\CO\backslash \RE$ we denote by $m$ the Stieltjes transform of the empirical spectral %distribution of the eigenvalues of the matrix $H$,
%\begin{equation}
%m(z) = \frac{1}{N} \sum_{\al} \frac{1}{\la_\al-z} = \frac1N \Tr (H-z)^{-1}\,,
%\end{equation}
%where $\{\la_\al\}_{\al=1,...,N}$ are the eigenvalues of $H$.
%%%%%%%%
%THEOREM
%%%%%%%%
\begin{theorem}\label{t:main}
Assume (\ref{e:gd}) and fix $\tildeeta > 0$.
\begin{itemize}
\item[i)] There exist constants $M_0, N_0, C, c , c_0 > 0$ such that
\begin{equation}
\label{e:mmsc}
\P\left(|m (E+i\eta) -m_{sc} (E+i\eta)|\geq \frac{K}{N\eta}\right)\leq \frac{(Cq)^{cq^2}}{K^q}
\end{equation}
for all $\eta \leq \tildeeta$, $|E| \leq 2 + \eta$,  $K > 0$, $N > N_0$ such that $N\eta \geq M_0$, $q \in \NA$ with $q \leq c_0 (N\eta)^{1/8}$.
\item[ii)] For any $\tilde{E} > 0$ there exist constants $M_0, N_0, C, c , c_0 > 0$ such that
\begin{equation}
\label{e:mmscim}
\P\left(|\Im m (E+i\eta) -\Im m_{sc} (E+i\eta)|\geq \frac{K}{N\eta}\right)\leq \frac{(Cq)^{cq^2}}{K^q}
\end{equation}
for all $\eta \leq \tildeeta$, $|E| \leq \tilde E$, $K > 0$, $N > N_0$ such that $N\eta \geq M_0$, $q \in \NA$ with $q \leq c_0 (N\eta)^{1/8}$.
\end{itemize}
\end{theorem}

%{\it Remark:} the reason why for $|E| > 2 + \eta$, we can only prove convergence of the imaginary part of the Stieltjes %transform is an instability of the equation for the difference $m (E+i\eta) -m_{sc} (E+i\eta)$; we will come back to this %point after Proposition \ref{p:lambda}.

\medskip

In the bulk of the spectrum, away from the edges, we can improve the bound (\ref{e:mmscim}), showing that it holds for all $q \in \NA$ (in contrast to (\ref{e:mmscim}), which we only prove under the additional condition
$q \leq c (N\eta)^{1/8}$).
%%%%%%%%
%THEOREM
%%%%%%%%
\begin{theorem}\label{t:bulk}
Assume (\ref{e:gd}), fix $\tildeeta>0$ and $\kappa > 0$. Then there exist constants
$M_0, N_0, C, c  > 0$ such that
\[
\P\left(|m (E+i\eta) -m_{sc} (E+i\eta)|\geq \frac{K}{N\eta}\right)\leq \frac{(Cq)^{cq^2}}{K^q}
\]
for all $E \in [-2+\kappa ; 2-\kappa]$, $K > 0$, $N > N_0$, $\eta \leq \tildeeta$ such that $N\eta \geq M_0$ and $q \in \NA$.
\end{theorem}

\medskip

As an application of Theorem \ref{t:main}, we prove the convergence of the counting function of the eigenvalues. We define
\[
 n(E) = \frac1N |\{\al:\, \la_\al\leq E\}|
\]
and compare it with the cumulative distribution of the semicircle law, defined by
\[
 n_{sc}(E) = \int_{-\infty}^E \rho_{sc}(x) dx\,.
\]

\begin{theorem}\label{t:counting}
Assume (\ref{e:gd}). Then there exists constants $N_0, C,c > 0$ such that
\begin{equation}\label{eq:count}
%\begin{equation}\label{you2}
\P\left(|n(E) -n_{sc}(E)| \geq \frac{K\log N}{N}\right) \leq \frac{(Cq)^{cq^2}}{K^q}
%\end{equation}
\end{equation}
for all $E\in\RE$, $K > 0$, $N > N_0$, $q \in \NA$.
\end{theorem}

{\it Remark:} Theorem \ref{t:counting} immediately implies the convergence of the density of states. Let $\mathcal{N} [a;b]$ denote the number of eigenvalues in the interval $[a;b]$. {F}rom (\ref{eq:count}), we find
\begin{equation}\label{eq:DOS} \P \left( \left| \frac{\mathcal{N} [E - \frac{\xi}{2N}; E+ \frac{\xi}{2N}]}{\xi} - \rho_{sc} (E) \right| \geq \frac{K \log N}{\xi} \right) \leq \frac{(Cq)^{cq^2}}{K^q} \, \end{equation}
for all $\xi > 0$.

%%%%%%
%REMARK
%%%%%%
%\begin{remark}
%In the bulk, for $E\in[-2+\ve,2-\ve]$ for any $\ve>0$ not dependent on $N$, one can use Th.
%\ref{t:bulk} instead of Th. \ref{t:main}, then the bounds \eqref{you0} and \eqref{you2} hold true for %any $q\geq 1$.
%\end{remark}

%%%%%%
%SECTION
%%%%%%

\medskip

Another application of Theorem \ref{t:main} is the rigidity of the eigenvalues of $H$, as stated in the following theorem. Recall that the classical location of the $\alpha$-th eigenvalue,  here denoted by $\gamma_\al$,  is defined by
\[
\int_{-\infty}^{\gamma_{\al}} \rho_{sc}(x) dx = \frac{\al}{N} \qquad 1\leq\al\leq N
\]
\begin{theorem}\label{t:rigidity}
Assume (\ref{e:gd}). For $\alpha = 1, \dots , N$, let $\hat \al = \min\{\al, N+1-\al\}$. Then there exist constants $C, c, N_0,\varepsilon > 0$ such that
\begin{equation}\label{eq:rigidity} \P\left(|\la_\al - \ga_\al| \geq \frac{K\log N}{N}  \left(\frac{N}{\hat \al}\right)^{\frac13}\right) \leq    \frac{(Cq)^{cq^2}}{K^q} \, \end{equation}
for all $N > N_0$,$K > 0$, $q \in \NA$ with $q \leq N^\varepsilon$.
\end{theorem}

{\it Remark:} In contrast with (\ref{e:mmsc}) and (\ref{e:mmscim}), the bounds (\ref{eq:count}), (\ref{eq:DOS}) and (\ref{eq:rigidity}) are not optimal. By the results of \cite{Gu} for GUE, the eigenvalues are expected to fluctuate on the scale $\sqrt{\log N}/N$.

\medskip

The plan of the paper is as follows. In Section \ref{sec:preli} we collect some preliminary results which will be needed in our analysis. In particular, we prove some bounds on the Stieltjes transform of the semicircle law, and we discuss certain algebraic identities concerning the entries of the resolvent $G(z) = (H-z)^{-1}$. In Section \ref{sec:no-opt}, we derive a first non-optimal estimate on the fluctuations of the Stieltjes transform $m$ around $m_{sc}$, see, in particular, Lemma \ref{l:no}.
In Section \ref{sec:opt}, we prove Theorem \ref{t:main}, assuming a bound for the high moments of certain error terms, as stated in Lemma \ref{l:R3}. Section \ref{sec:pf} is then devoted to the proof of Lemma \ref{l:R3}; it is here that, to obtain optimal bounds on $m-m_{sc}$, we make use of an expansion algorithm similar to the one introduced in \cite{ekyy-rxv12}. In Section \ref{sec:bulk}, we prove Theorem \ref{t:bulk}. Finally, in Section \ref{sec:counting} and Section \ref{sec:rigidity}, we prove Theorem \ref{t:counting} and, respectively, Theorem \ref{t:rigidity}.

\section{Preliminary results}
\label{sec:preli}

\subsection{Properties of the Stieltjes transform of the semicircle law}

In the next proposition, we give simple proofs of several useful bounds on the Stieltjes transform $m_{sc}$ of the semicircle law, as defined in (\ref{eq:msc}).
\begin{proposition}\label{p:msc}
Let $z=E+i\eta$. For arbitrary $E\in \RE$ and $\eta > 0$, we have $(1+|z|)^{-1} < |m_{sc} (z)| < 1$ and the bounds
\begin{equation}\label{eq:msc-bds}
\begin{split}
\frac{\Im m_{sc} (z)}{|m_{sc}^2 (z) - 1 |}  &\leq 1 \qquad \textrm{and} \quad
\frac{\eta}{|m_{sc}^2 (z) - 1|} \leq \sqrt \eta \, (1+\sqrt \eta )\,.
\end{split}
\end{equation}
\end{proposition}
\begin{proof}
We rewrite the equation for $m_{sc} (z)$ as
\begin{equation}\label{eq:fixmsc} z+m_{sc} (z) = -\frac{1}{m_{sc} (z)} \,.\end{equation}
Writing explicitly the real and imaginary part we get
\[ \left(\Re m_{sc} (z) \right) \, \left( 1+\frac{1}{|m_{sc} (z)|^2}\right) + i \left(\Im m_{sc} (z) \right)\, \left(1-\frac{1}{|m_{sc} (z)|^2}\right) = -E-i\eta \,.
\]
Therefore, we get
\begin{equation}\label{msc2}
 (1-|m_{sc} (z)|^2)  \Im m_{sc} (z) =\eta |m_{sc} (z)|^2\,.
\end{equation}
Since $\Im m_{sc} (z) > 0$  for $\eta>0 $, it follows that $|m_{sc} (z)| < 1$. The bound $|m_{sc}(z)| > (1+|z|)^{-1}$ follows then immediately from (\ref{eq:fixmsc}).

\medskip

To prove the first bound in (\ref{eq:msc-bds}), we rewrite
\[
|m_{sc}^2 (z) -1|=|m_{sc} (z) -1||m_{sc} (z) +1| \,.
\]
We distinguish two cases. If $|\Re m_{sc} (z) - 1| \geq 1$, we use $|m_{sc} (z) -1|\geq 1$ and $|m_{sc} (z) +1| \geq \Im m_{sc} (z)$. If  $|\Re m_{sc} (z)-1| \leq 1$ we use $|m_{sc} (z)+1| \geq |\Re m_{sc} (z) +1|\geq 2 - |\Re m_{sc} (z) -1| \geq 1 $ and $|m_{sc} (z)-1| \geq \Im m_{sc} (z)$. In both cases, we find $|m_{sc}^2 (z) -1| \geq \Im m_{sc} (z)$.

\medskip

Also to prove the second bound in (\ref{eq:msc-bds}), we distinguish two cases. If $\Im m_{sc} (z) \geq \sqrt \eta$, we use \[ |m_{sc}^2 (z)-1| \geq \Im m_{sc} (z) \geq \sqrt{\eta} \geq \frac{\sqrt{\eta}}{1+\sqrt{\eta}} \,. \]
If, on the other hand, $\Im m_{sc} (z) \leq \sqrt \eta$, Eq. \eqref{msc2} gives
\[ \eta \, |m_{sc} (z)|^2 = (1-|m_{sc}(z)|^2) \, \text{Im } m_{sc} (z) \leq \sqrt{\eta} \, (1-|m_{sc} (z)|^2) \]
and therefore
\[ |m_{sc} (z)|^2 \leq \frac{1}{1+\sqrt\eta}\,. \]
The desired bound follows then from $|m_{sc}^2 (z)-1|\geq 1- |m_{sc} (z)|^2$.
\end{proof}

\subsection{Identities for resolvent entries and for the Stieltjes transform} We use the following notation
\[
G (z) = (H-z)^{-1}
\]
for the resolvent of the Wigner matrix $H$, so that
\[
m (z)=\frac1N \Tr G (z) \,.
\]
To shorten the notation, we will often omit from the notation the dependence on $z$ from the resolvent $G(z)$ and from the Stieltjes transform $m(z)$ (and also from the Stieltjes transform of the semicircle law $m_{sc} (z)$), writing $G \equiv G(z)$, $m \equiv m(z)$ (and $m_{sc} \equiv m_{sc} (z)$).

\medskip

We denote by $H^{(j)}$ the $j$-th minor of the matrix $H$ and by $G^{(j)}$ the matrix $G^{(j)} = (H^{(j)}-z)^{-1}$. For any $j =1 ,\dots , N$, we can express the $(j,j)-$entry of the resolvent $G$ as
\begin{equation}\label{Gjj}
G_{jj}= \frac{1}{h_{jj} - z - \aa_j^* G^{(j)} \aa_j}
\end{equation}
where $\aa_j \in \CO^{N-1}$ is a vector, whose components are the off-diagonal entries of the $j$-th column of $H$. We will label the components of $\aa_j$ with indices in $\{1,...,N\}\backslash\{j\}$.

\medskip

To compare the diagonal entries of $G$ with the diagonal entries of $G^{(j)}$, we will make use of the identity
\begin{equation}\label{eq:G-Gj}
 G_{kk} = G_{kk}^{(j)} + \frac{G_{jk}G_{kj}}{G_{jj}} \end{equation}
valid for all $k\neq j$. A proof of (\ref{Gjj}) and of (\ref{eq:G-Gj}) can be found, for example in \cite{E}.

\medskip

Let \[\Lambda (z) = m (z) - m_{sc} (z) .\] Also in this case, we will often omit the $z$ dependence from the notation, writing just $\Lambda \equiv \Lambda (z)$. Next, we are going to obtain a self-consistent equation for $\Lambda$; we follow here Section 5.1 in \cite{ekyy-rxv12}. {F}rom (\ref{Gjj}), we have
\[
G_{jj}= - \frac{1}{-h_{jj} + z + (\ID - \E_j) \, \aa_j^* G^{(j)} \aa_j + \E_j \, \aa_j^* G^{(j)} \aa_j}
\]
where $\E_j$ denotes the expectation with respect to $\aa_j$. Since the components of $\aa_j$ are independent of $G^{(j)}$, we find
\[ \E_j \, \aa_j^* G^{(j)} \aa_j = \frac{1}{N} \sum_{k \not = j} G^{(j)}_{kk} = m + \frac{1}{N} \sum_{k \not = j} (G^{(j)}_{kk} - G_{kk}) - \frac{1}{N} G_{jj}\,.
\]
Hence
\begin{equation}\label{3.9}
G_{jj} = -\frac{1}{z + m_{sc} + \Lambda + \Upsilon_j}
\end{equation}
with
\begin{equation}\label{eq:Ups}\begin{split}
 \Upsilon_j &  = -h_{jj}  -\frac{1}{N}G_{jj} - \frac{1}{N}\sum_{k\neq j}(G_{kk} - G^{(j)}_{kk}) + (\ID - \E_j) \aa_j^* G^{(j)} \aa_j \\
            & =  -h_{jj}  - \frac{1}{G_{jj}}\frac{1}{N}\sum_{k}G_{jk}G_{kj} + (\ID - \E_j) \aa_j^* G^{(j)} \aa_j \,.
\end{split}\end{equation}

\medskip

Setting
\[g_j = G_{jj} - m_{sc}\,, \]
we find from (\ref{3.9}) that
\begin{equation}\label{eq:gj}
 g_j = m_{sc}(\La + \Upsilon_j)G_{jj}\,.
\end{equation}
Averaging over $j$, we find
\[ \Lambda = \frac{1}{N} \sum_{j=1}^N g_j = \frac{m_{sc} \Lambda}{N} \sum_{j=1}^N G_{jj} + \frac{m_{sc}}{N} \sum_{j=1}^N \Upsilon_j G_{jj} =   m_{sc}^2 \Lambda + m_{sc} \Lambda^2 + \frac{m_{sc}}{N} \sum_{j=1}^N \Upsilon_j G_{jj} \]
and hence, rearranging the terms and dividing by $m_{sc}$,
\begin{equation}\label{eqLa1}
 \Lambda^2 + \frac{m_{sc}^2-1}{m_{sc}} \Lambda + R = 0
\end{equation}
where we set \begin{equation}\label{R} R =  \frac{1}{N}\sum_{j=1}^N \Upsilon_j \, G_{jj}\, . \end{equation}
%so that Eq. \eqref{eqLa1} reads
%\begin{equation}\label{eqLa}
%m_{sc} \Lambda^2 + (m_{sc}^2-1) \Lambda + m_{sc} R = 0\,.
%\end{equation}
Using Eq. \eqref{eq:fix}, we can rewrite (\ref{eqLa1}) as
\begin{equation}\label{eqLa2}
 \Lambda^2 + (2m_{sc}+z) \Lambda + R = 0
\end{equation}
This quadratic equation for $\Lambda$ has two solutions. However, only one of them satisfies the condition $\text{Im } \Lambda > - \text{Im } m_{sc}$. We find
\begin{equation}\label{lambda} \Lambda  = - m_{sc} - \frac{z}{2} + \sqrt{\left(m_{sc}+\frac{z}{2}\right)^2 - R}  \end{equation}
where, for any $w \in \CO$, we denote by $\sqrt{w}$ the square root of $w$ with $\text{Im } w \geq 0$ (for $w$ real and positive, we choose $\sqrt{w} \geq 0$). Notice that, since
\[ m_{sc} (z) = - \frac{z}{2} + \sqrt{\frac{z^2}{4}-1} \, . \]
for any $z$ with $\text{Im } z > 0$, we can write
\begin{equation}\label{lambda2} \Lambda  = - \sqrt{\frac{z^2}{4} - 1} + \sqrt{\frac{z^2}{4} -1 - R}  \end{equation}
%For any $w\in\CO$ we define the square root of $w$ with $\Im \sqrt{w} \geq 0$ (branch cut on the positive real axes) %and write the two solutions of Eq. \eqref{eqLa2} as $S_1 = - m_{sc} - \frac{z}{2} + \sqrt{(m_{sc}+\frac{z}{2})^2 - R} $ and %$S_2 = - m_{sc} - \frac{z}{2} - \sqrt{(m_{sc}+\frac{z}{2})^2 - R} $. \\
%We claim that,  for $\eta>0$, $\Lambda$ coincides with $S_1$. To rule out the solution $S_2$ we note that  $\Im \La = %\Im m - \Im m_{sc} > - \Im m_{sc}$, because the Stieltjes transform $m(z)$ has positive imaginary part for $\eta>0$. %With our choice of the branch cut for the square root, we have that $\Im S_2 < - \Im m_{sc}$, hence $S_2$ cannot be %equal to  $\Lambda$. We conclude that for $\eta >0$
%\begin{equation}\label{lambda} \Lambda  = - m_{sc} - \frac{z}{2} + \sqrt{\left(m_{sc}+\frac{z}{2}\right)^2 - R}
%\end{equation}

\medskip

In the following proposition we use the algebraic identity (\ref{lambda}) to bound $|\Lambda|$ and $|\text{Im } \Lambda|$ in terms of $|R|$. These bounds are crucial in the proof of Theorem \ref{t:main}.
\begin{proposition}\label{p:lambda}
Let $z=E+i\eta$. As in (\ref{R}) and (\ref{lambda}), (\ref{lambda2}) we use the shorthand notations $\Lambda = \Lambda (z)$ and $R = R(z)$. There exists a constant $C>0$ such that
\begin{equation}\label{e:Lambda2} |\Lambda| \leq C  \min \left\{\frac{|R|}{| m_{sc}^2-1|}, \sqrt{|R|} \right\} \end{equation}
for all $\eta > 0$ and $|E| \leq 2 + \eta$ and such that
\begin{align}\label{e:Lambda1} | \Im \Lambda| &\leq C  \min \left\{\frac{|R|}{| m_{sc}^2-1|}, \sqrt{|R|} \right\}
\\ \label{e:Lambda3}
\min \left( |\Lambda| , |\Lambda+ (2m_{sc}+z)| \right) &\leq
 C \sqrt{|R|} \end{align}
 for all $\eta > 0$ and $E \in \RE$.
\end{proposition}

{\it Remark:} Note that, for $|E| > 2 + \eta$, Proposition \ref{p:lambda} only controls the imaginary part of $\Lambda$ in terms of $R$. This is the origin of the distinction between part i) and part ii) of Theorem~\ref{t:main}. The reason why we cannot control the real part of $\Lambda$ for $|E| > 2 + \eta$ is clear from (\ref{lambda2}). Recall that $\sqrt{w}$ denotes the square root of $w$ with positive imaginary part, which is discontinuous across the positive real axis. For $|E| > 2 + \eta$, $\text{Re } (z^2/4-1) > 0$ and $|\text{Im } (z^2/4-1)| = |E| \eta /2$ can be very small (for small $\eta$). For this reason even for very small $|R|$ it is possible that $z^2/4-1$ and $z^2/4-1-R$ are on opposite sides of the positive real axis. In this case, $\Lambda \simeq -2 \sqrt{z^2/4-1}$ would not be small. Notice, however, that the discontinuity of the square root concerns only its real part; the imaginary part is continuous on the whole complex plane. This is why we still get control of the imaginary part of $\Lambda$ in (\ref{e:Lambda1}) also for $|E| > 2 + \eta$.

\begin{proof}
As before we denote by $\sqrt{w}$ the square root of $w \in \CO$ with $\text{Im } w \geq 0$ (and with $\sqrt{w} >0$ for $w >0$ real). We claim that
\begin{itemize}
\item  for any fixed $c > 0$, there exists a constant $C >0$ such that
\begin{equation}\label{rootbound2} \left| \sqrt{a+b} - \sqrt{a} \right| \leq C\frac{|b|}{\sqrt{|a|+|b|}} \end{equation}
for all $a,b \in \CO$ with $|\text{Im } a| \geq c \, \text{Re } a$.
\item there exists a constant $C > 0$ such that
\begin{equation}\label{rootbound1} \left| \Im \sqrt{a+b} - \Im \sqrt{a} \right| \leq C\frac{|b|}{\sqrt{|a|+|b|}}  \end{equation}
for all $a,b\in\CO$.
\end{itemize}

The bounds \eqref{e:Lambda2} and \eqref{e:Lambda1} follow directly from (\ref{lambda}), (\ref{lambda2}), applying \eqref{rootbound1} and, respectively, \eqref{rootbound2} with $a= (m_{sc} + z/2)^2 = (m_{sc}^2-1)^2/ m_{sc}^2$ and $b=-R$ (since $|m_{sc}| \leq 1$; see Proposition~\ref{p:msc}). Here, to prove (\ref{e:Lambda2}), we use the fact that
$a = z^2/4 -1$, and therefore, with $z=E+i\eta$, that
\[ \text{Re } a = \frac{E^2 - \eta^2}{4} - 1 \quad \text{and } \quad \text{Im } a = \frac{E \eta}{2} \]
This implies that $|\text{Im } a| \geq c \, \text{Re } a$ for all $|E| \leq 2 + \eta$ (for $|E| \leq 2$, $\text{Re } a < 0$ and the bound $| \text{Im } a| \geq c \, \text{Re } a$ is trivial).

\medskip

To prove \eqref{e:Lambda3}, we observe from (\ref{eqLa2}) that $\Lambda (\Lambda + (2m_{sc}+z)) = - R$. This implies that
\[ |\Lambda+ (2m_{sc}+z)| = \frac{|R|}{|\Lambda|} \]
and hence that either $|\Lambda| \leq \sqrt{|R|}$ or $|\Lambda+ (2m_{sc}+z)| \leq \sqrt{|R|}$.

\medskip

We conclude with the proof of the bounds \eqref{rootbound2} and \eqref{rootbound1}. Both estimates clearly hold when $|b| > |a|/2$, since then
\[ \left| \sqrt{a+b} - \sqrt{a} \right| \leq C |b|^{1/2} \leq C \frac{|b|}{\sqrt{|a| + |b|}} \]
Hence, we can assume that $|b| \leq |a|/2$. If $|\text{Im } a|  \geq c  \,\text{Re } a$, we have $\text{Im } \sqrt{a} \geq \wt{c} \, |a|^{1/2}$ for an appropriate constant $\wt{c} >0$. This implies that
\[ \left| \sqrt{a+b} - \sqrt{a} \right| = \left| \frac{b}{\sqrt{a+b} + \sqrt{a}} \right| \leq \frac{|b|}{\text{Im } \sqrt{a}  + \text{Im } \sqrt{a+b}} \leq C \frac{|b|}{|a|^{1/2}} \leq C \frac{|b|}{\sqrt{|a| + |b|}} \]
from the assumption $|b| \leq |a|/2$ and because $\text{Im } \sqrt{a+b} \geq 0$.

\medskip

It remains to prove (\ref{rootbound1}) for $|\text{Im } a|  < c \, \text{Re } a$ and $|b| \leq |a|/2$. To this end, we consider first the case $a > 0$ real. If $\text{Im } b \geq 0$, we find
\[\begin{split}  \left| \text{Im } \sqrt{a+b} - \text{Im } \sqrt{a} \right| &\leq \left| \sqrt{a+b} - \sqrt{a} \right| \leq \frac{|b|}{|\sqrt{a+b} + \sqrt{a}|} \\& \leq  \frac{|b|}{\text{Re } \sqrt{a+b} + \sqrt{a}}  \leq \frac{|b|}{\sqrt{a}} \leq C \frac{|b|}{\sqrt{|a| + |b|}} \end{split} \]
since $\text{Re } \sqrt{a+b} \geq 0$ and $|b| < c \, |a|$. If instead $\text{Im } b < 0$, we have (again, for $a>0$ real)
\[ \begin{split}  \left| \text{Im } \sqrt{a+b} - \text{Im } \sqrt{a} \right| &= \left| \text{Im } \sqrt{a+b} + \text{Im } \sqrt{a} \right| \leq \left| \sqrt{a+b} + \sqrt{a} \right| \\ & \leq \frac{|b|}{|\sqrt{a} - \sqrt{a+b}|} \leq \frac{|b|}{\sqrt{|a|}} \leq \frac{|b|}{\sqrt{|a| + |b|}} \end{split} \]
because, in this case, $\text{Re } \sqrt{a+b} < 0$.

\medskip

Now, we consider general $a,b \in \CO$ with $|\text{Im } a|  < c \, \text{Re } a$ and $|b| \leq c \, |a|$. Again, we distinguish two cases. If $\text{Im } (a+b) > 0$, we have $\text{Re } \sqrt{a+b} > 0$ and therefore
\[ \left| \text{Im } \sqrt{a+b} - \text{Im } \sqrt{a} \right| \leq \frac{|b|}{\text{Re } \sqrt{a}} \leq \frac{|b|}{\sqrt{|a|}} \leq  \frac{|b|}{\sqrt{|a|+|b|}} \]
since the assumption $|\text{Im } a|  < c \,  \text{Re } a$ implies that $\text{Re } \sqrt{a} \geq c \sqrt{|a|}$. Finally, let us assume that $\text{Im } (a+b) < 0$. Then we find $\lambda \in (0,1)$ such that $a+\lambda b = d > 0$ is real. Hence, we can estimate
\[ \begin{split} \left| \text{Im } \sqrt{a+b} - \text{Im } \sqrt{a} \right| &\leq \left| \text{Im } \sqrt{a+b} - \text{Im } \sqrt{d}\right| + \left| \text{Im } \sqrt{d} - \text{Im } \sqrt{a} \right| \\ &= \left|\text{Im } \sqrt{d+ (1-\lambda)b} - \text{Im } \sqrt{d} \right| + \left| \text{Im } \sqrt{d-\lambda b} - \text{Im } \sqrt{d} \right| \\ &\leq C \frac{|b|}{\sqrt{|d|+|b|}} \leq C\frac{|b|}{\sqrt{|a|+|b|}}  \end{split} \]
where we applied the bounds obtained above for $a > 0$ real.

%We start with \eqref{rootbound2}. It is enough to prove that
%\begin{equation}\label{eq:3} \left| \sqrt{a+b} - \sqrt{a} \right| \leq C |b|^{1/2} \end{equation}
%if $|b| \geq |a|$ and that
%\begin{equation}\label{eq:2} \left| \sqrt{a+b} - \sqrt{a} \right| \leq C \frac{|b|}{|a|^{1/2}} \end{equation}
%if $|b| < |a|$ and $\Re  a \leq 0$. The first bound is trivial. As for \eqref{eq:2}, we notice that, for $\Re  a \leq 0$, we %have $\Im  \sqrt{a} = \sqrt{(|a|-\Re a)/2} \geq \sqrt {|a| / 2}$.  Since
%\[ \sqrt{a+b} - \sqrt{a} = \frac{b}{\sqrt{a+b} + \sqrt{a}} \]
%we conclude that
%\[ \left| \sqrt{a+b} - \sqrt{a} \right| = \frac{|b|}{|\sqrt{a+b} + \sqrt{a}|} \leq \frac{|b|}{\Im \sqrt{a+b} + \Im \sqrt{a}} \leq \frac{|b|}%{\Im \sqrt{a}} \leq C \frac{|b|}{|a|^{1/2}} \]
%It remains to prove  the bound \eqref{rootbound1}. As we proved \eqref{eq:3} and \eqref{eq:2}, we only need to prove  %that
%\[  \left| \Im \sqrt{a+b} - \Im \sqrt{a} \right|  \leq C  \frac{|b|}{|a|^{1/2}} \]   if $|b| < |a|$ and $\Re  a > 0$. We distinguish two %cases. If $\sgn \Im a = \sgn\Im(a+ b)$ then $\sgn \Re\sqrt a = \sgn\Re \sqrt{(a+ b)}$, which implies $|\sqrt{a+b} +
%\sqrt{a}|\geq |\sqrt{a}|$, and we have that
%\[  \left|\Im \sqrt{a+b} - \Im \sqrt{a} \right| \leq \left| \sqrt{a+b} - \sqrt{a} \right| = \frac{|b|}{|\sqrt{a+b} + \sqrt{a}|} \leq \frac{|b|}%{|\sqrt{a}|}  \]
% If $\sgn \Im a = - \sgn\Im(a+ b)$ . . .
\end{proof}

The inequalities in Proposition \ref{p:lambda} are the starting point for all our estimates on the random variable $\Lambda$. In the next section, we will prove a non-optimal bound on $\Lambda$, based on the second bound in \eqref{e:Lambda2}, proportional to $\sqrt{|R|}$ (for $z=E+i\eta$ with $|E| > 2 +\eta$, we cannot apply (\ref{e:Lambda2}); instead, we will control the imaginary part of $\Lambda$ using (\ref{e:Lambda1}) and we will get some bound on its real part using (\ref{e:Lambda3})). Afterwards, in Section \ref{sec:opt}, we will prove Theorem \ref{t:main}, which gives an optimal estimate on $|\Lambda|$. To achieve this goal, we will apply the first bound in \eqref{e:Lambda2}, proportional to $|R|$ in the bulk of the spectrum (where the denominator $|m^2_{sc}-1|$ is of order one) and we will use the second bound in \eqref{e:Lambda2}, proportional to $\sqrt{|R|}$, close to the edges of the spectrum (where $|m^2_{sc}-1|$ is small). Also in this case, for $|E| > 2 + \eta$ we cannot apply (\ref{e:Lambda2}); instead, in this region we will apply (\ref{e:Lambda1}) (getting only a bound for the imaginary part of $\Lambda$).

\section{Non-optimal bound on $\Lambda = m - m_{sc}$}
\label{sec:no-opt}

\subsection{Bound on moments of $\Upsilon_j$.}

According to \eqref{e:Lambda1} and \eqref{e:Lambda2}, in order to show smallness of $\Lambda$ we need bounds on the quantity $R = N^{-1} \sum_{j=1}^N \Upsilon_j G_{jj}$, introduced in (\ref{R}). We will prove in Lemma \ref{l:boundG11} that the diagonal entries $G_{jj}$ of the resolvent are bounded with high probability, hence it is possible to bound $R$ by controlling the coefficients $\Upsilon_j$.  This is the goal of the next lemma.
\begin{lemma}
Assume (\ref{e:gd}). Let $\Upsilon_1$ be defined as in (\ref{eq:Ups}). Then there exists a universal constant $C$ such that
\begin{equation}\label{boundUps}
\E|\Upsilon_1 |^{2q} \leq  (Cq)^{2q} \left(\frac{(\Im m_{sc})^{q} + \E|\Im\Lambda|^{q}}{(N\eta)^{q}}+\frac1{(N\eta)^{2q}} + \frac1{N^q} \right)
\end{equation}
for all $E \in \RE$, $N \geq 1$, $\eta \geq 1/N$ and $q \in \NA$.
\end{lemma}

{\it Remarks:} Of course, the same bound is valid for $\Upsilon_j$, for any $j=2, \dots , N$, since all these variables have the same law. {F}rom (\ref{boundUps}), we conclude that the typical size of $|\Upsilon_j|$ is of the order $(N\eta)^{-1/2}$ (assuming that $\text{Im } \Lambda|$ is at least bounded by a constant). Combined with an upper bound on the diagonal entries of the resolvent, this would allow us to show that $R=N^{-1} \sum_{j} \Upsilon_j G_{jj}$ is at most of the order $(N\eta)^{-1/2}$. {F}rom (\ref{e:Lambda2}), this would imply (at least in the bulk of the spectrum) that $|\Lambda| \lesssim (N\eta)^{-1/2}$, with high probability. This result, however, is still far from the optimal bound $|\Lambda| \lesssim (N\eta)^{-1}$ that we are aiming at. To obtain the optimal bound, it is crucial to use the cancellations among the different terms in the average defining $R$. We will do so in Sections \ref{sec:opt} and \ref{sec:pf}.

\begin{proof}
We observe that
\[ \left| \frac{1}{G_{11}} \frac{1}{N} \sum_{k=1}^N G_{1k} G_{k1} \right| = \frac{|(G^2)_{11}|}{N |G_{11}|}\,. \]
Since
\[ |(G^2(z))_{11} | = |\langle e_1, G^2 (z) e_1 \rangle| \leq \| G^* (z) e_1 \| \| G (z) e_1 \| \leq  |G (z)|^2_{11} \]
and
\begin{equation}\label{eq:GG*} |G (z)|^2 = G(z) G(\bar{z}) = \frac{\text{Im } G(z)}{\eta} \end{equation}
we obtain
\[ |(G^2(z))_{11} | \leq \frac{\text{Im } G_{11} (z)}{\eta} \]
and therefore
\[ \left| \frac{1}{G_{11}}\frac1N \sum_{k=1}^N G_{1k} G_{k1} \right| \leq \frac{1}{N\eta} \]
and hence, using also the assumption (\ref{e:gd}), we find
\[
%\begin{equation}\label{eq:ups-1}
\E \, |\Upsilon_j|^{2q} \leq \frac{(Cq)^q}{N^q} + \frac{1}{(N\eta)^{2q}} + C^q \, \E  |Z_j|^{2q}
%\end{equation}
\]
where we defined
\begin{equation}\label{eq:Zj} Z_j = (\ID-\E_j) \, \aa^*_j G^{(j)} \aa_j \, . \end{equation}
The claim now follows from Proposition \ref{p:Zkq} below, where we bound the moments of $Z_j$.
\end{proof}

In the next proposition, we bound the moments of the variable $Z_1$, defined in (\ref{eq:Zj}), in terms of the moments of  $\Im \Lambda$.
\begin{proposition}\label{p:Zkq}
%Let $z=E+i\eta$.
%For any $j = 1,...,N$, set
%\begin{equation}\label{Zk}
% Z_j = (\ID - \E_j) \aa_j^* G^{(j)} \aa_j \,.
%\end{equation}
Assume (\ref{e:gd}) and let $Z_1$ be defined as in (\ref{eq:Zj}). Then there exists a constant  $C > 0$ such that
\begin{equation}\label{boundZk}
\E \, |Z_1|^{2q} \leq (Cq)^{2q} \left(\frac{(\Im m_{sc})^{q} + \E|\Im\Lambda|^{q}}{(N\eta)^{q}}+\frac1{(N\eta)^{2q}} \right)\,,
\end{equation}
for all $E \in \RE$, $N \geq 1$, $\eta \geq 1/N$ and $q \in \NA$.
\end{proposition}

%\end{proposition}
%PROOF

\begin{proof}
We keep the randomness in the entries of the minor $H^{(1)}$ fixed. By the Hanson-Wright large deviation estimate, the fluctuations of the quadratic form around its mean can be controlled by the Hilbert-Schmidt norm of $G^{(1)}$. In fact, it follows from Prop. \ref{p:HW} that
\[ \P_1 \left( \left|(\ID - \E_1) \, \aa_1^* G^{(1)} \aa_1 \right| \geq t \right) \leq C \exp \, \left( - \frac{c t N}{\left(\Tr \, |G^{(1)}|^2\right)^{1/2}} \right) \]
and therefore that
\[
%\begin{equation}
%\label{EZ}
\E \, |Z_1|^{2q} = \E \int_0^{+\infty} \P_1\left( |(\ID - \E_1) \aa_1^* G^{(1)} \aa_1 |^{2q} \geq t \right)  dt
 \leq (Cq)^{2q} \, \E \left(\frac{\Tr|G^{(1)}|^2}{N^2}\right)^{q}\,.
%\end{equation}
\]
To conclude the proof of (\ref{boundZk}), we notice that, similarly to (\ref{eq:GG*}),
\begin{equation}\label{eq:|G|2} |G^{(1)} (z)|^2  = \frac{\Im G^{(1)} (z)}{\eta} \end{equation}
and hence that
\begin{equation}\label{eq:trG2} \begin{split}
 \frac{\Tr|G^{(1)}|^2}{N^2} =& \frac{1}{N} \frac{1}{N\eta} \Im \sum_{k\neq1}G^{(1)}_{kk}\\
 = & -\frac{1}{N}  \frac{1}{N\eta}  \Im \frac{1}{G_{11}}\sum_{k} G_{k1}G_{1k} + \frac{\Im m}{N\eta} \\
\leq & \frac{1}{(N\eta)^2} + \frac{\Im m_{sc} + |\Im\La|}{N\eta}\,.
\end{split}\end{equation}
\end{proof}

\subsection{Bound on moments of $\Lambda$ in terms of moments of $G_{11}$}

In the next lemma, we estimate the moments of $|\text{Im } \Lambda|$ in terms of the moments of the diagonal entries of the resolvent. In the region $|E|\leq 2+\eta$, the lemma also gives a bound on the moments of $|\Lambda|$ (including also the real part of $\Lambda$).
\begin{lemma}\label{l:baby_bound_at_edge}
Assume (\ref{e:gd}), fix $\tildeeta> 0$ and set $z = E+i\eta$.
\begin{itemize}
\item[i)] There exists a constant $C>0$ such that
\begin{equation}\label{e:baby_bound_in_bulk}
\E |\Lambda|^{2q}
 \leq \frac{(Cq)^{4q/3}}{(N\eta)^{q/2}} \left[ (\E |G_{11}|^{2q})^{2/3} + 1 \right] \, ,
\end{equation}
for all $|E| \leq 2 + \eta$, $\eta \leq \tildeeta$, $N \geq 1$ such that $N \eta \geq 1$, and for all $q\in \NA$.
\item[ii)] There exists a constant $C>0$ such that
\begin{equation}\label{e:baby_bound_at_edge}
\E |\Im \Lambda|^{2q}
 \leq \frac{(Cq)^{4q/3}}{(N\eta)^{q/2}} \left[ (\E |G_{11}|^{2q})^{2/3} + 1 \right] \, ,
\end{equation}
for all $E\in\RE$, $\eta \leq \tildeeta$, $N \geq 1$ such that $N \eta \geq 1$, and for all $q\in \NA$.
\end{itemize}
\end{lemma}
 \begin{proof}
We show (\ref{e:baby_bound_at_edge}). To this end, we use the bound (\ref{e:Lambda1})
proportional to $\sqrt{|R|}$. We find
\[
% \begin{equation}\label{eq:1}
\begin{split}
 \E |\text{Im } \Lambda|^{2q} &\leq C^q \E |R|^q
 \leq C^q \, \E \left|  \frac{1}{N}\sum_j \Upsilon_j G_{jj}\right|^q
 \leq C^q \, \E |\Upsilon_1 G_{11}|^q
 \\
 &
 \leq C^q \sqrt{\E |\Upsilon_1|^{2q} \, \E |G_{11}|^{2q}}
 \\
 &
 \leq  (Cq)^q \, \sqrt{\E |G_{11}|^{2q}} \, \sqrt{\frac{(\Im m_{sc})^{q} + \E |\Im\Lambda|^{q}}{(N\eta)^q} + \frac{1}{N^{q}}+ \frac{1}{(N\eta)^{2q}}} \,.
 \end{split}
 %\end{equation}
\]
By Cauchy-Schwarz we find
\[
%\begin{equation}\label{eq:2}
\E |\Im\Lambda|^{2q}
\leq  \frac{(Cq)^q \, \sqrt{ \E |G_{11}|^{2q} }}{(N \eta)^{q/2}}\left(\sqrt[4]{ \, \E |\Im\Lambda|^{2q}} + 1 \right)\,.
%\end{equation}
\]
Since $x^{1/4} \leq \delta x + \delta^{-1/3}$ for all $x,\delta >0$, we find
\[ \begin{split}  \E |\Im\Lambda|^{2q} &\leq (Cq)^{4q/3}  \left[  \left(\frac{\E |G_{11}|^{2q}}{(N\eta)^q} \right)^{2/3} +  \left(\frac{\E |G_{11}|^{2q}}{(N\eta)^q} \right)^{1/2} \right] \leq \frac{(Cq)^{4q/3}}{(N\eta)^{q/2}} \left[ (\E |G_{11}|^{2q})^{2/3} + 1 \right] \, . \end{split} \]
The bound (\ref{e:baby_bound_in_bulk}) follows analogously, using \eqref{e:Lambda2} instead of \eqref{e:Lambda1}.
\end{proof}

\subsection{Bound on moments of $G_{11}$}

To obtain a (non-optimal) bound on the moments of $|\text{Im} \Lambda|$ and (in the region $|E| \leq 2 + \eta$) of $|\Lambda|$, we still need control on the moments of the diagonal entries of the resolvent. This is the content of the following lemma, which makes use of the bound (\ref{e:Lambda3}) (in the region $|E| \leq 2$, we could also apply (\ref{e:Lambda2}) to arrive at the same conclusions).
\begin{lemma}\label{l:boundG11}
Assume (\ref{e:gd}), fix $\tilde E>0$ and $\tildeeta> 0$. We set $z = E+i\eta$ and use the shorthand notation $G_{jj} = G_{jj} (z)$. Then there exist constants $C_0, M >0$ such that
\begin{equation}\label{boundG11}
\E \, |G_{11}|^{q} \leq C_0^{q}\, ,
\end{equation}
for all $|E|<\tilde E$, $\eta \leq \tildeeta$, $N \geq1$ such that $N \eta \geq M$, and for all $q\in \NA$ with $q \leq (N\eta)^{1/4}$.
\end{lemma}
\begin{proof} We show first how to relate $G_{11}(E+i\eta/s)$ to $G_{11}(E+i\eta)$. We consider $G_{11} (E+i\nu)$ as a function of $\nu$, with fixed $E$. We find
\[ \frac{d}{d\nu} \, \log \, G_{11}(E+i\nu) = \frac{i}{G_{11} (E+i\nu)}\,  \left\langle e_1, \frac{1}{(H-E-i\nu)^2} e_1 \right\rangle\,. \]
Taking absolute value, and arguing as in (\ref{eq:|G|2}), we find
\[ \left|  \frac{d}{d\nu} \, \log \, G_{11}(E+i\nu) \right|  \leq \frac{1}{|G_{11} (E+i\nu)|}  \left\| G(E+i\nu) e_1 \right\| \, \left\| G(E-i\nu) e_1 \right\| \leq \frac{1}{\nu} \,.\]
Hence
\[ \left| \log \, G_{11} (E+i\eta) - \log G_{11} (E+i\eta /s) \right| = \left| \int_{\eta/s}^\eta d\nu \, \frac{d}{d\nu} \, \log \, G_{11}(E+i\nu) \right| \leq \int_{\eta/s}^\eta d\nu \, \frac{1}{\nu} = \log \, s\,. \]
This gives
\[ \left| \text{Re } \log \, G_{11} (E+i\eta) - \text{Re } \log \, G_{11} (E+i\eta /s) \right| \leq \log \, s \]
and therefore
\[ \log \, |G_{11} (E+i\eta/s)| \leq \log \, |G_{11} (E+i\eta)| + \log \, s \,. \]
Exponentiating, we find
\begin{equation}\label{eq:G11s} |G_{11} (E+i\eta/s)| \leq s \, |G_{11} (E+i\eta)|\,. \end{equation}
The proof of (\ref{boundG11}) now proceeds by induction in the distance $\eta$ from the real axis.
We fix $|E| \leq \tilde{E}$ and we use the notation $G_{11}(E+i\eta)\equiv G_{11}(\eta)$, and similarly for the other quantities depending on $z$. To start the induction, we recall that for $\eta \geq 0.1$, $|G_{11} (\eta)| \leq 10$, so $\E |G_{11}(\eta)|^{q} \leq  10^{q}$.  Next, we assume that, for some $\eta_0 > 0$,
\begin{equation}\label{eq:eta0} \E |G_{11}(\eta_0)|^{q} \leq  C_0^{q} \quad \text{ for all $1 \leq q \leq (N\eta_0)^{1/4}$} \,. \end{equation} We show that the same bounds hold true, if we replace $\eta_0$ by $\eta_1 = \eta_0/16$.

\medskip

We set $s = 16$ in (\ref{eq:G11s}); from the assumption (\ref{eq:eta0}) we get that
\begin{equation}\label{e:3}
\E |G_{11}(\eta_1)|^q \leq (50 C_0)^{q}
\end{equation}
for any $q \leq (N \eta_0)^{1/4}$. To iterate further this bound, we need to improve it and get back to the constant $C_0^q$. To this end, we recall the identity  $G_{11}=  m_{sc} + m_{sc}(\Lambda + \Upsilon_1) G_{11}$ (see \eqref{eq:gj}), valid for all $z \in \CO$. On the one hand, this implies that
\begin{equation}\label{eq:G111} |G_{11} | \leq 1 + |\Lambda | |G_{11}| + |\Upsilon_1| |G_{11} | \, . \end{equation}
for all $z \in \CO$. On the other hand, it also implies that
\[ G_{11} = m_{sc} + m_{sc} (\Lambda + 2m_{sc} + z + \Upsilon_1) G_{11} - m_{sc} (2m_{sc}+z) G_{11} \]
Since $1+m_{sc} (2m_{sc} + z) = m_{sc}^2$, we conclude that
\[ |G_{11}| \leq \frac{1}{|m_{sc}|} + \frac{1}{|m_{sc}|} |\Lambda+2m_{sc}+z| |G_{11}| + \frac{1}{|m_{sc}|} |\Upsilon_1| |G_{11}| \, . \]
Using the bound $|m_{sc}|^{-1} \leq 1 + |z|$, and combining with (\ref{eq:G111}), we obtain that there exists a constant $C$ depending on $\tilde{E}$ and $\tildeeta$ such that
\[ \begin{split} |G_{11}| &\leq C \left( 1+ \min \left(|\Lambda|, |\Lambda + 2m_{sc}+z| \right) |G_{11}| + |\Upsilon_1| |G_{11}| \right) \\ &\leq C \left( 1+ \sqrt{|R|} \, |G_{11}| + |\Upsilon_1| |G_{11}| \right) \end{split}
\]
for all $z = E+i\eta$ with $|E| \leq \tilde{E}$ and $0< \eta < \tildeeta$. In the second inequality, we used (\ref{e:Lambda3}).  Taking the $q$-th moment, using the definition (\ref{R}) and Cauchy-Schwarz, we find
\begin{equation} \label{eq:EG11q} \E |G_{11}|^q \leq (3C)^q \left( 1 + (\E |G_{11}|^{2q})^{3/4} (\E |\Upsilon_1|^{2q})^{1/4} + (\E |G_{11}|^{2q})^{1/2} (\E |\Upsilon_1|^{2q})^{1/2} \right) \end{equation}
for all $z = E+i\eta$ with $|E| \leq \tilde{E}$ and $0< \eta < \tildeeta$. Plugging (\ref{e:baby_bound_at_edge}) into (\ref{boundUps}), we find
\[\E|\Upsilon_1 |^{2q} \leq \frac{(\wt{C}q)^{8q/3}}{(N\eta)^q} \left(1 + (\E |G_{11}|^{2q})^{1/3} \right) \]
for all $z = E+i\eta$ with $|E| \leq \tilde{E}$ and $0< \eta < \tildeeta$. Next, we insert the last bound into (\ref{eq:EG11q}). We specialize to $z = E + i\eta_1$ and to $q \leq (N\eta_1)^{1/4}$ (which implies that $2q \leq (N\eta_0)^{1/4}$ and therefore that we can use (\ref{e:3}) to bound $\E |G_{11}|^{2q}$). We find
\[ \begin{split} \E |G_{11} (\eta_1)|^q  \leq \; &(3C)^q \left( 1+ \frac{(\wt{C}q)^{2q/3}}{(N\eta_1)^{q/4}} (1+ (\E |G_{11} (\eta_1)|^{2q})^{1/12}) (\E |G_{11} (\eta_1)|^{2q})^{3/4} \right. \\ & \left. \hspace{4cm} +  \frac{(\wt{C}q)^{4q/3}}{(N\eta_1)^{q/2}} (1+ (\E |G_{11} (\eta_1)|^{2q})^{1/6}) (\E |G_{11} (\eta_1)|^{2q})^{1/2} \right) \\ \leq \; & (3C)^q \left(1 + \frac{(\wt{C}q)^{2q/3}}{(N\eta_1)^{q/4}} (\E |G_{11} (\eta_1)|^{2q})^{5/6} +
 \frac{(\wt{C}q)^{4q/3}}{(N\eta_1)^{q/2}} (\E |G_{11} (\eta_1)|^{2q})^{2/3} \right)
 \\ \leq \; &(3C)^q  \left(1 + \frac{(\wt{C}q)^{2q/3}}{(N\eta_1)^{q/4}} (50 C_0)^{5q/3} +
 \frac{(\wt{C}q)^{4q/3}}{(N\eta_1)^{q/2}} (50 C_0)^{4q/3} \right)
\\ \leq \;& (3C)^q  \left(1 + \frac{K^q}{(N\eta_1)^{q/12}} \right) \, .
\end{split} \]
for a constant $K > 0$ depending only on $C_0$ (and on the universal constant $\wt{C}$). Choosing first $C_0 > 6C$ and then $M > K^{12}$, it follows that, for $N\eta_1 > M$, $K^q/(N\eta_1)^{q/12} < 1$ and therefore that
\[  \E |G_{11} (\eta_1)|^q  \leq C_0^q \, .\]
This concludes the proof of Lemma \ref{l:boundG11}.
\end{proof}

\subsection{Non-optimal bound on moments of $\Lambda$}

Combining Lemma \ref{l:baby_bound_at_edge} with Lemma \ref{l:boundG11}, we obtain a non-optimal upper bound on the moments of $\Lambda$.
\begin{lemma}\label{l:no}  Assume (\ref{e:gd}) and fix $\tildeeta> 0$. As usual, we set $z = E+i\eta$ and denote $\Lambda = \Lambda (z)$.
\begin{itemize}
\item[i)] There exist constants $C,M > 0$ such that
\[
 \E|\La|^q \leq  \frac{  (Cq)^{\frac{2q}3}}{(N\eta)^{\frac{q}4}}\,
 \]
for all $0 < \eta \leq \tildeeta$, $|E| \leq 2 + \eta$, $N \geq 1$ such that $N \eta \geq M$, and for all $q\in \NA$ with $q \leq (N\eta)^{1/4}$.
\item[ii)] Fix $\tilde{E} > 0$. Then there exist constants $C,M > 0$ such that
\[
 \E|\text{Im } \La|^q \leq  \frac{  (Cq)^{\frac{2q}3}}{(N\eta)^{\frac{q}4}}\,
 \]
for all $0 < \eta \leq \tildeeta$, $|E| \leq \tilde{E}$, $N \geq 1$ such that $N \eta \geq M$, and for all $q\in \NA$ with $q \leq (N\eta)^{1/4}$.
\end{itemize}
\end{lemma}

\section{Optimal bound on $\Lambda$; proof of Theorem \ref{t:main}}
\label{sec:opt}

To prove Theorem \ref{t:main}, we will control the moments of the error term (\ref{R}) through the control parameters
\[
\EE_q = \max \left\{\frac{1}{(N\eta)^{2q}},\frac{(\Im m_{sc})^{q} +\E|\text{Im } \La|^{q}}{(N\eta)^{q}}\right\}   + \frac{1}{N^q} \,.
\]
We notice that, by definition, $\EE_q \geq (N\eta)^{-2q}$. Moreover, assuming that $\E \, |\Im \Lambda|^{2q} \leq 1$ and $\E \, |\Im \Lambda|^p \leq 1$ (by Lemma \ref{l:no}, we know that these assumptions hold true, for $\max \, (p, 2q) \leq (N\eta)^{1/4}$), there exist universal constants $C,c > 0$ such that
\begin{equation}\label{eq:EEpq}  c \EE_p \leq \EE_q \leq C \EE_p^{q/p} \end{equation}
for all $1 \leq q \leq p$. In fact, the second inequality is a consequence of $\E \, |\Im \Lambda|^q \leq (\E \, |\Im\Lambda|^p)^{q/p}$. The first inequality in (\ref{eq:EEpq}), on the other hand, can be proven as follows. If $(\text{Im } m_{sc})^p + \E |\Im \Lambda|^p \leq (N\eta)^{-p}$, then
\[ \EE_p = \frac{1}{(N\eta)^{2p}} + \frac{1}{N^p} \leq \frac{1}{(N\eta)^{2q}} + \frac{1}{N^q} \leq \EE_q \,. \]
If instead $(\text{Im } m_{sc})^p + \E |\Im\Lambda|^p > (N\eta)^{-p}$, then
\[ \EE_p = \frac{(\text{Im } m_{sc})^p + \E |\Im\Lambda|^p}{(N\eta)^p} + \frac{1}{N^p}\,. \]
For $p \geq 2q$, we find (under the assumption that $\E \, |\Im\Lambda|^p \leq 1$)
\[ \EE_p \leq \frac{2}{(N\eta)^{2q}}+ \frac{1}{N^{2q}} \leq 2 \EE_q \,.\]
For $q \leq p < 2q$, we write $p = 2 \alpha q + (1-\alpha) q$ (with $\alpha = (p-q)/q$) and we observe that, by H\"older's inequality (under the assumption $\E |\Im\Lambda|^{2q} \leq 1$)
\[ \E |\Im\Lambda|^p \leq (\E |\Im\Lambda|^q)^{1-\alpha} \,  (\E |\Im\Lambda|^{2q})^\alpha \leq (\E |\Im\Lambda|^q)^{1-\alpha} \,. \]
This gives
\[ \begin{split} \frac{(\text{Im } m_{sc})^p + \E |\Im\Lambda|^p}{(N\eta)^p} &\leq C \left( \frac{(\text{Im } m_{sc})^q + \E |\Im\Lambda|^q}{(N\eta)^q} \right)^{1-\alpha} \, \left(\frac{1}{(N\eta)^{2q}} \right)^\alpha \\ &\leq C \max \left\{\frac{1}{(N\eta)^{2q}} ,  \frac{(\text{Im } m_{sc})^q + \E |\Im\Lambda|^q}{(N\eta)^q} \right\}  \end{split} \]
and proves that $\EE_p \leq C \EE_q$.

\medskip

The proof of Theorem  \ref{t:main}  is  based on the following lemma.
%%%%%%%
%LEMMA
%%%%%%%
\begin{lemma}\label{l:R3}
Assume (\ref{e:gd}), fix $\tilde E>0$ and $\tildeeta> 0$. Set $z = E+i\eta$. There exist constants $C,M,c_0 > 0$ such that
\[
%\begin{equation}\label{e:R3}
  \E\left|\frac{1}{N} \sum_k Z_k G_{kk} \right|^{2q} \leq(Cq)^{cq^2} \EE_{4q}^{\frac12}  \,,
 %\end{equation}
 \]
 for all $|E| \leq \tilde E$, $0 < \eta \leq \tildeeta$, $N \geq1$ such that $N \eta \geq M$, and for all $q\in \NA$ with
 $q \leq c_0  (N\eta)^{1/8}$.
\end{lemma}

Technically, the proof of this lemma is the main part of the paper. We defer it to the next section.
Assuming the result of Lemma \ref{l:R3}, we can now proceed with the proof of Theorem \ref{t:main}.

\medskip

\begin{proof}[Proof of Theorem \ref{t:main}]
We prove first the bound (\ref{e:mmscim}) for the imaginary part of $\Lambda$. According to \eqref{e:Lambda1}, we need to control the moments of $R= N^{-1} \sum_{j=1}^N \Upsilon_j G_{jj}$. Recalling the definition (\ref{eq:Ups}) of the coefficients $\Upsilon_j$, we obtain
\begin{equation}\label{strike}
\E|R|^{2q} \leq C^q \left(\E \left| \frac{1}{N}\sum_j h_{kk} \, G_{kk}
 \right|^{2q} +
\E\left| \frac{1}{N^2}\sum_{k,j} G_{kj}G_{jk} \right|^{2q} +
\E\left|  \frac{1}{N}\sum_{k} Z_k G_{kk} \right|^{2q}
    \right)
\end{equation}
where $Z_k = (\ID - \E_k) \, \aa_k^* G^{(k)} \aa_k$. {F}rom the assumption (\ref{e:gd}) and Lemma \ref{l:boundG11}, the first term in the parenthesis can be bounded by
\begin{equation}\label{post1}
\E \left| \frac{1}{N}\sum_j h_{kk} \, G_{kk}
 \right|^{2q}   \leq \E|h_{11}G_{11}|^{2q} \leq \frac{(Cq)^{q}}{N^q}  ( \E|G_{11}|^{4q} )^{1/2} \leq \frac{(Cq)^{q}}{N^q}  \,
\end{equation}
for all $1 \leq q \leq (N\eta)^{1/4}$. As for the second term on the r.h.s. of Eq. \eqref{strike}, we get
\begin{equation}\label{post2}
\E\left| \frac{1}{N^2}\sum_{k,j} G_{kj}G_{jk} \right|^{2q}\leq\E \left(\frac{1}{N^2} \sum_{k,j} |G_{jk}|^2 \right)^{2q} \leq \E
\left(\frac{\Im m}{N\eta }\right)^{2q}
\leq   C^q  \frac{(\Im m_{sc})^{2q} + \E|\Im \La|^{2q}}{(N\eta)^{2q}}\,.
\end{equation}
Combining \eqref{post1} and \eqref{post2} with Lemma \ref{l:R3}, we find
  \begin{equation}\label{e:R}
   \E|R|^{2q} \leq (Cq)^{cq^2} \EE_{4q}^{\frac12}
  \end{equation}
for all $1 \leq q \leq c_0 (N\eta)^{1/8}$. Next, we fix $N,\eta$ and $q$ with $1 \leq q \leq c_0  (N\eta)^{1/8}$, and we insert the last bound on the r.h.s. of (\ref{e:Lambda1}). We distinguish several cases. We can assume that \[ (\Im m_{sc})^{4q}+\E|\Im\La|^{4q} \geq \frac{1}{(N\eta)^{4q}} \] since otherwise there is nothing to prove. In this case,
\[ \EE_{2q} = \frac{(\Im m_{sc})^{2q}+\E|\Im\La|^{2q}}{(N\eta)^{2q}} + \frac{1}{N^{2q}} = \frac{(\Im m_{sc})^{2q}+\eta^{2q} + \E|\Im \La|^{2q}}{(N\eta)^{2q}} \,.\]

\medskip

If $\E |\Im \Lambda|^{2q} \leq (\text{Im } m_{sc})^{2q} + \eta^{2q}$, we use the bound proportional to $|R|$ in (\ref{e:Lambda1}). We find
\[
 \E |\Im\La|^{q} \leq \frac{C^q \, \E|R|^q}{|m_{sc}^2-1|^q} \leq \frac{(Cq)^{cq^2}}{|m_{sc}^2-1|^q}  \left(\frac{(\Im m_{sc})^{2q} + \eta^{2q}}{(N\eta)^{2q}} \right)^{\frac12} \leq \frac{(Cq)^{cq^2}}{(N\eta)^q} \left[ \frac{\text{Im } m_{sc}}{|m_{sc}^2-1|} + \frac{\eta}{|m_{sc}^2-1|} \right]^q \,.
\]
{F}rom Proposition \ref{p:msc}, we conclude
\begin{equation}\label{eq:momLa} \E \, |\Im\Lambda|^q \leq \frac{(Cq)^{cq^2}}{(N\eta)^q}\,. \end{equation}

\medskip

If, on the other hand, $(\Im m_{sc})^{2q} + \eta^{2q} \leq \E|\Im \La|^{2q}$,  we use the bound proportional to $|R|^{1/2}$ on the r.h.s. of (\ref{e:Lambda1}). We find
\[
 \E|\Im\La|^{2q} \leq C^q \, \E|R|^{q} \leq (Cq)^{cq^2}\left(\frac{\E|\Im \La|^{2q}}{(N\eta)^{2q}}\right)^{\frac12}\,.
\]
This implies that $\E|\Im \La|^{2q} \leq (Cq)^{2cq^2}/(N\eta)^{2q}$ and therefore that
\[ \E |\Im\La|^q \leq \sqrt{\E |\Im\La|^{2q}} \leq \frac{(Cq)^{cq^2}}{(N\eta)^q}\,. \]

\medskip

Combined with (\ref{eq:momLa}), this implies that
\[ \P \left( |\Im m (z) - \Im m_{sc} (z)| \geq \frac{K}{N\eta} \right) \leq \frac{(N\eta)^q}{K^q} \E |\Im \Lambda|^q \leq \frac{(Cq)^{cq^2}}{K^q} \,.\]
for all $1 \leq q \leq c_0  (N\eta)^{1/8}$ and concludes the proof of \eqref{e:mmscim}.

To prove (\ref{e:mmsc}), we proceed similarly, using however (\ref{e:Lambda1}) instead of (\ref{e:Lambda2}).
\end{proof}

\section{Proof of Lemma \ref{l:R3}}
\label{sec:pf}

We rewrite the quantity we are interested in in a more convenient form. Set
\[
%\begin{equation}\label{Wk}
 W_k = Z_k G_{kk}\,,
%\end{equation}
\]
then
\[
\frac{1}{N} \sum_k Z_k G_{kk}  = \frac1N \sum_{k} W_k  = \frac1N \sum_{k} (\ID - \E_k)W_k + \frac1N \sum_{k} \E_kW_k \,.
\]
By H\"older  inequality we get
\begin{equation}\label{separator}
 \E\left| \frac{1}{N} \sum_k Z_k G_{kk}\right|^{2q} \leq C^q \E\left| \frac1N \sum_{k} (\ID - \E_k)W_k\right|^{2q} + C^q \E\left| \E_1W_1\right|^{2q}\,.
\end{equation}
Next  we claim that under the assumptions of Theorem \ref{t:main}
\begin{equation}\label{mean}
\E\left|\E_1 W_1\right|^{2q}
%\\
%\leq
%(Cq)^{4q}  \left(  \frac{(\Im m_{sc})^{4q} + \E|\Lambda|^{4q}}{(N\eta)^{4q}}+\frac1{(N\eta)^{8q}}\right)^\frac12
\leq (Cq)^{4q} \EE_{4q}^{\frac12}
\end{equation}
and
%Next  we claim that the first term at the r.h.s. of \eqref{separator} is bounded by
\begin{equation}\label{555}
 \E\left| \frac1N \sum_k(\ID - \E_k)W_k\right|^{2q} \leq (Cq)^{cq^2}\EE_{4q}^{\frac12}\,.
\end{equation}
Lemma \ref{l:R3} then follows by inserting the bounds \eqref{mean} and \eqref{555} in \eqref{separator}. The rest of this section is devoted to the proof of \eqref{mean} and \eqref{555}.

%QUI
\subsection{Proof of Eq. \eqref{mean}}

Recalling that
\[ Z_k=-(\ID-\E_k) \frac{1}{G_{kk}},\]
we find
\begin{equation}\label{glo}
 \E_kW_k = \E_k\frac{ G_{kk}}{(\E_k \frac{1}{G_{kk}})} \left((\ID-\E_k)\frac{1}{G_{kk}}\right)^2 = \frac{ \E_k G_{kk} Z_k^2}{(\E_k \frac{1}{G_{kk}})}\,.
\end{equation}
for any $k=1,...,N$. To bound the denominator on the r.h.s. of the last equation, we make use of the following lemma.
\begin{lemma}\label{lwidetildeG}
Assume (\ref{e:gd}) and fix $\tilde E>0$,  $\tildeeta> 0$. Set $z = E+i\eta$. There exist constants $c, C_1 ,M > 0$ such that
\begin{equation}\label{boom}
\E\left|\frac{1}{\E_1 \frac{1}{G_{11}}}\right|^q \leq C_1^q\,,
\end{equation}
for all $|E|\leq \tilde E$,  $\eta \leq \tildeeta$, $N \geq1$ such that $N \eta \geq M$, and for all $q\in \NA$ with
$q \leq c (N\eta)^{1/4}$.
\end{lemma}
\begin{proof}
Define
 \[\widetilde{G_{11}} \equiv  \frac{1}{\E_1 \frac{1}{G_{11}}} = - \frac{1}{\frac{N-1}N m^{(1)} + z - h_{11}}\]
with $m^{(1)}=\frac1{N-1}\Tr G^{(1)}$.  We notice that
\[
\left|\frac{d}{d \nu} \log \widetilde{G_{11}}(E+i\nu)\right| = \left|\frac{i + \frac{N-1}N \frac{d}{d\nu} m^{(1)}}{\frac{N-1}{N} m^{(1)} + z - h_{11}}\right| \leq \frac{\nu+\frac{N-1}{N}\Im m^{(1)}}{\nu |\frac{N-1}{N} m^{(1)} + z - h_{11}|}\leq \frac{1}{\nu}\,.
\]
As argued in the proof of Lemma \ref{l:boundG11}, this implies that
\[ |\widetilde{G_{11}} (E+i\eta/s)| \leq s |\widetilde{G_{11}} (E+i\eta)| \,.\]

\medskip

Next we observe that
\[
\widetilde{G_{11}} = G_{11} + G_{11}\widetilde{G_{11}} Z_1 \,
\]
where we recall the definition $Z_1= (\ID-\E_1) \aa_1^* G^{(1)} \aa_1 = (\ID-\E_1) G_{11}^{-1}$. From Prop. \ref{p:Zkq} and Lemma \ref{l:boundG11} we get
\[
\E |\widetilde{G_{11}}|^q \leq  C^{q}+ \frac{(Cq)^q}{(N\eta)^{q/2}} (\E|\widetilde{G_{11}}|^{3q})^{1/3}\,,
\]
provided $3q \leq (N\eta)^{1/4}$. Here we used the fact that $|m_{sc}|< 1$ and $\E|\Im \La|^{3q}\leq 1$ by Lemma \ref{l:no}. At this point, we can use a bootstrap argument similar to the one used in Lemma \ref{l:boundG11} to conclude the proof.
\end{proof}

\medskip

Applying (\ref{boom}) to (\ref{glo}) and using also Lemma \ref{l:boundG11} and Prop. \ref{p:Zkq} we conclude that
\[\begin{aligned}
 \E\left|\E_1 W_1\right|^{2q}& \leq
(\E|G_{11}|^{8q})^{\frac14}
\left( \E\left|\frac{1}{\E_1 \frac{1}{G_{11}}}  \right|^{8q}\right)^{\frac14}
(\E| Z_1|^{8q})^{\frac12} \\
& \leq
(Cq)^{4q}  \left(  \frac{(\Im m_{sc})^{4q} + \E|\Im \Lambda|^{4q}}{(N\eta)^{4q}}+\frac1{(N\eta)^{8q}}\right)^\frac12\,,
\end{aligned}\]
This immediately implies (\ref{mean}).

%QUI
\subsection{Preliminaries to the proof of Eq. \eqref{555}: the expansion algorithm}

In order to prove the bound (\ref{555}), we will follow Theorems 4.6 and 4.7 in \cite{ekyy-rxv12} to
expand the expectation on the l.h.s. in a sum over indices $k_1, \dots , k_{2q}$ of products of terms of the form
\begin{equation}\label{eq:example} (\ID-\E_{k_j}) W_{k_j} =  (\ID - \E_{k_j}) \left[ (\ID - \E_{k_j}) \frac{1}{G_{k_j k_j}} \right] G_{k_j k_j} \,. \end{equation}
In each one of these factors, we will expand further the resolvent entries (both in the numerator and in the denominator) using the relation (\ref{eq:G-Gj}). The gain here is that we either gain independence (the first term on the r.h.s. of (\ref{eq:G-Gj}) does not depend on the randomness in the $j$-th row and column of the original matrix) or, alternatively, we gain smallness (the second term on the r.h.s. of (\ref{eq:G-Gj}) has two more off-diagonal entries, which are typically of size $(N\eta)^{-1/2} \ll 1$). Recursively, we continue to expand the resulting terms using either (\ref{eq:G-Gj}) or a similar formula for the off-diagonal entries (see (\ref{dec1}) below), either until there are no more variables over which we can expand (in which case we reached maximal independence) or until there are sufficiently many off-diagonal terms (to show that (\ref{eq:G-Gj}) is of the order $(N\eta)^{-2q}$). This type of expansions for resolvent entries have first been applied in the analysis of Wigner matrices in \cite{EYY0,EYY1,EYY}; we follow here the more recent work \cite{ekyy-rxv12}.

\medskip

Next, we give a precise and detailed definition of the expansion algorithm. Afterwards, we apply it to
the resolvent entries $G_{k_j k_j}$ and $1/G_{k_j k_j}$ appearing in (\ref{eq:example}).

\subsubsection{A general description of the expansion algorithm}

Let $\T\subset\{1,...,N\}$ be a set of indices. We denote by $H^{(\T)}$ the $(N-|\T|) \times (N-|\T|)$ minor of the matrix $H$ obtained by deleting the rows and columns of $H$ corresponding to the indices in $\T$; the entries of the matrix $H^{(\T)}$ are $h_{ij}$ with $i,j\in\{1,...,N\}\backslash\T$. We denote by $G^{(\T)}$ the matrix
\[ G^{(\T)} = \left(H^{(\T)} -z \right)^{-1} \,.  \]
Rows and columns of $G^{(\T)}$ are also labelled with indices $i,j\in\{1,...,N\}\backslash\T$.

\medskip

Similarly to (\ref{eq:G-Gj}), we have (see, e.g.,  Eq. (4.6) in \cite{ekyy-rxv12})
\begin{equation}
\label{dec1}
  G_{ij}^{(\T)} = G_{ij}^{(\T k)}+ \frac{G_{ik}^{(\T)}G_{kj}^{(\T)}}{G_{kk}^{(\T)}}\qquad \forall i,j,k\notin \T \textrm{ and } i,j\neq k.
\end{equation}
Setting $i = j$, we find
\begin{equation}
\label{dec1.5}
  G_{ii}^{(\T)} = G_{ii}^{(\T k)}+ \frac{G_{ik}^{(\T)}G_{ki}^{(\T)}}{G_{kk}^{(\T)}}\qquad \forall i,k\notin \T \textrm{ and } i\neq k
\end{equation}
 and thus
\begin{equation}
\label{dec2}
 \frac{1}{G_{ii}^{(\T)}} = \frac{1}{G_{ii}^{(\T k)}} - \frac{ G_{ik}^{(\T)}G_{ki}^{(\T)}}{ G_{ii}^{(\T)}G_{ii}^{(\T k)}G_{kk}^{(\T)} } \qquad \forall i,k\notin \T \textrm{ and } i\neq k.
\end{equation}
For any $\kk = (k_1,...,k_{2q})$, with $k_s\in\{1,...,N\}$ for $s = 1, \dots , 2q$, let $\PP(\kk)$ be the partition of $\{1,...,2q\}$ induced by the coincidences in $\kk$. In other words, $\PP (\kk)$ is the partition induced by the equivalence relation on $\{1,...,2q\}$ defined by $r\sim s$ if and
only if $k_r= k_s$. We denote, moreover, by $\fP_{2q}$ the set of all partitions of $\{1,...,2q\}$.

\medskip

Following \cite{ekyy-rxv12}, we can write the expectation on the l.h.s. of Eq. \eqref{555} as
\begin{equation}\label{panama0}
 \E\left| \frac1N \sum_k(\ID - \E_k)W_k\right|^{2q}  = \frac{1}{N^{2q}} \sum_{P\in\fP_{2q}}\sum_{\kk}\II(\PP(\kk)=P) \, V(\kk)
\end{equation}
where each coordinate of $\kk = (k_1, \dots , k_{2q})$ is summed over the set $\{ 1, \dots , N \}$, and
\begin{equation}\label{panama}
V (\kk) = \E (\ID - \E_{k_1})W_{k_1}  \cdots (\ID - \E_{k_q})W_{k_q} \overline{(\ID - \E_{k_{q+1}})W_{k_{q+1}}} \cdots \overline{(\ID - \E_{k_{2q}})W_{k_{2q}}} \, .
 \end{equation}

\medskip

For any fixed  $\kk$ we say that $s \in \{1, \dots , 2q \}$ is a {\it lone} label if $k_s\neq k_r$ for any $r\in\{1,...,2q\}\backslash \{s\}$, i.e., if $s$ is the only element in its equivalence class in the partition $\PP(\kk)$. We denote by $L(\kk)$ the set of lone indices and by $\kk_L \subset \kk$ the set of coordinates of $\kk$ associated with lone labels.

\medskip

For a fixed $\kk$, we say that a resolvent  entry $G^{(\T)}_{ij}$ with $i,j \not \in \T$ is \emph{maximally expanded} if $\kk_L \subseteq \T\cup\{i,j\}$. Similarly, we say that a factor $G_{ii}^{(\T)}$ or $1/G_{ii}^{(\T)}$ with $i \not \in \T$ is maximally expanded, if $\kk_L \subset \T \cup \{ i \}$.

\medskip

Let $Q$ be a product of diagonal and/or off-diagonal resolvent entries of the form
\begin{equation}\label{B}
 Q =G_{ii}^{(\T)} \quad\textrm{or} \quad
   Q=\frac{1}{G^{(\T_1)}_{i_1i_1}} \cdots \frac{1}{G^{(\T_l)}_{i_li_l}} G^{(\T_{l+1})}_{i_{l+1}j_{l+1}} \cdots G^{(\T_{l+m})}_{i_{l+m}j_{l+m}}
\end{equation}
for arbitrary non-negative integers $l,m$ and for $i_r \neq j_r$ and $i_r, j_r \notin \T_r$, and $i\notin \T$. We assume the resolvent entries on the r.h.s. of Eq. \eqref{B} to be ordered in some way.

\medskip

For fixed $\kk$, we define, similarly to \cite{ekyy-rxv12}, the operation $w$ on monomials of the form of $Q$.
\begin{itemize}
 \item[-] The operation $w(Q)$ can be performed only if at least one resolvent entry in $Q$ (either in the numerator or in the denominator) is not maximally expanded.
 \item[-] If at least one of the resolvent entries of $Q$ is not maximally expanded, $w$ acts only on the first (according to the previously chosen order) resolvent entry of $Q$ which is not maximally expanded. We set then $w(Q) = w_0(Q)+w_1(Q)$ where $w_0$ and $w_1$ are defined by the following rules:
\begin{itemize}
 \item[-]  If the first not maximally expanded  resolvent entry is $G^{(\T)}_{ij}$, we use (\ref{dec1}) to define $w_0$ and $w_1$ in the following way:
\begin{equation}\label{wod}
G^{(\T)}_{ij} \xrightarrow{w_0}  G^{(\T u)}_{ij}
\qquad
\textrm{and}
\qquad
G^{(\T)}_{ij} \xrightarrow{w_1}  \frac{1}{G^{(\T)}_{uu}} G^{(\T )}_{iu} G^{(\T)}_{uj}
\end{equation}
 where $u$ is the smallest index such that $u\in\kk_L \backslash (\T\cup\{i,j\})$.
\item[-] If the first not maximally expanded  resolvent entry is $G^{(\T)}_{ii}$, we use (\ref{dec1.5}) to define $w_0$ and $w_1$ in the following way:
\begin{equation}\label{wd1}
 G^{(\T)}_{ii} \xrightarrow{w_0}  G^{(\T u)}_{ii}
\qquad
\textrm{and}
\qquad
 G^{(\T)}_{ii} \xrightarrow{w_1}  \frac{G_{iu}^{(\T)}G_{ui}^{(\T)}}{G_{uu}^{(\T)}}
\end{equation}
   where $u$ is the smallest index such that $u\in\kk_L\backslash (\T\cup \{i\})$.
\item[-] If the first not maximally expanded  resolvent entry is $1/G^{(\T)}_{ii}$, we use (\ref{dec2}) to define $w_0$ and $w_1$ in the following way:
\begin{equation}\label{wd2}
\frac1{ G^{(\T)}_{ii}} \xrightarrow{w_0}  \frac{1}{G^{(\T u)}_{ii}}
\qquad
\textrm{and}
\qquad
\frac1{ G^{(\T)}_{ii}} \xrightarrow{w_1}  -
\frac{1}{G^{(\T )}_{ii}}\frac{1}{G^{(\T)}_{uu}}\frac{1}{G^{(\T u)}_{ii}}
G^{(\T)}_{iu}G^{(\T )}_{ui}
\end{equation}
   where $u$ is the smallest index such that $u\in\kk_L\backslash (\T\cup \{i\})$.
% where $u$ is the smallest index such that $u\in \{1, \dots, 2q\} \backslash (\T\cup \{i\})$.
\item[-] With this definition, we rewrite $Q$ as $w(Q) = w_0(Q) + w_1(Q)$ where $w_0(Q)$ and $w_1(Q)$ are two new monomials in diagonal and off-diagonal resolvent entries of the form
(\ref{B}).
\end{itemize}
\end{itemize}
For any starting monomial $Q$, we use the operation $w$ repeatedly; we decompose $Q$ in a sum of terms having the form (\ref{B}). To keep track of all the resulting terms, let $\sigma$ be a (ordered, finite) string of $0$ and $1$. For any fixed $\sigma$ we construct, analogously to \cite{ekyy-rxv12}, the monomial $Q_\sigma$ through the following recursion
\[
Q_0 =  w_0(Q)  \qquad \textrm{and} \qquad
 Q_1 =  w_1(Q)
\]
and
\[
Q_{{\sigma'}0} =  w_0(Q_{\sigma'})  \qquad \textrm{and} \qquad
 Q_{{\sigma'}1} =  w_1(Q_{\sigma'}).
\]
To write $Q$ as an appropriate sum of factors $Q_\sigma$ we apply recursively the operation $w$ as many times as allowed by the following stopping rule:
\begin{itemize}
 \item[$(\SR)$] Continue to apply $w$ to $Q_\sigma$ until either all the resolvent entries of $Q_\sigma$ are maximally expanded or the number of off-diagonal resolvent entries in $Q_\sigma$ is greater than  $2q$.
\end{itemize}
In this way we obtain
\[
%\begin{equation}\label{expanded0}
 Q = \sum_{\sigma \in \LL_Q} Q_\sigma
%\end{equation}
\]
where $\LL_Q$ is the set of all strings $\sigma$ such that $Q_\sigma$ satisfies the stopping rule.

\subsubsection{Application to terms of the form (\ref{eq:example})}\label{sssectapply}
We use now the expansion algorithm to expand terms of the form (\ref{eq:example}), which arise from (\ref{555}). We consider the initial monomials $A^r = 1/ G_{k_r k_r}$ and $B^r = G_{k_r k_r}$ (in \cite{ekyy-rxv12} there were only terms of the type $A^r$; this is the reason why the operation (\ref{wd1}) was not used there). Applying the expansion algorithm presented in the previous section, we find
\begin{equation}\label{expanded}
 A^r := \frac{1}{G_{k_r k_r}} :=  \sum_{\sigma \in \LL_r} A^r_\sigma
\end{equation}
and
\begin{equation}\label{expandedt}
 B^r := G_{k_r k_r} :=  \sum_{\bsigma \in \MM_r} B^r_\bsigma
\end{equation}
where the sets $\LL_r$ and $\MM_r$ are chosen such that $A^r_\sigma$ and $B^r_\bsigma$ satisfy the stopping rule $(\SR)$.

\medskip

We denote by $\od(A^r_\sigma)$ (respectively $\od(B^r_\bsigma)$) the number of off-diagonal resolvent entries in $A^r_\sigma$ (respectively $B^r_\bsigma$). {F}rom the stopping rule $(\SR)$, it follows that for all $\sigma \in \LL_r$, we can have only two possibilities: either all resolvent entries in $A^r_\sigma$ are maximally expanded and $\od(A^r_{\sigma}) \leq 2q$, or alternatively, $2q+1\leq \od(A^r_\sigma)\leq 2q+2$. This follows because, by \eqref{wod} - \eqref{wd2}, the operation $w$ increases the number of off-diagonal resolvent entries by at most two. Analogously, we find for $B^r_\bsigma$ that, for all $\bsigma \in \MM_r$, either all resolvent entries are maximally expanded and $\od(B^r_{\bsigma})\leq 2q$, or, alternatively, $2q+1\leq \od(B^r_\bsigma)\leq 2q+2$.

\medskip

We denote by $\dd(A^r_\sigma)$ (respectively $\dd(B^r_\bsigma)$)  the number of  diagonal resolvent entries in $A^r_\sigma$ (respectively $B^r_\bsigma$), appearing either in the numerator or in the denominator. We note that the difference $\dd(A^r_\sigma) - \od(A^r_\sigma)$ is invariant with respect to the operation $w$ (because, when dealing with terms arising from the initial $A^r$, we only apply (\ref{wod}) or (\ref{wd2}), never (\ref{wd1})). This implies that
\begin{equation}\label{ddod1}
\dd(A^r_\sigma) = \od(A^r_\sigma) + 1 \,.
\end{equation}
On the other hand, the difference $\dd(B^r_\bsigma) - \od(B^r_\bsigma)$ is not always invariant. In fact, the operation $w_1$ applied to a resolvent entry in the numerator, according to (\ref{wd1}), decreases it by two. We remark, however, that this operation can take place at most once (because  $w_1$ in (\ref{wd1}) removes the diagonal entry from the numerator). As a consequence, we have either $\dd(B^r_\bsigma)=1$ and $\od(B^r_\bsigma) = 0$ (if the operation $w_1$ in (\ref{wd1}) never occurs) or  $\dd(B^r_\bsigma) - \od(B^r_\bsigma) = -1$ (if the operation $w_1$ in (\ref{wd1}) takes place once). We conclude that
\begin{equation}\label{ddod2}
 \dd(B^r_\bsigma) = \max\{ \od(B^r_\bsigma) - 1,1\}\,.
\end{equation}
It follows from (\ref{ddod1}) and (\ref{ddod2}) and from the previous bounds on the number of off-diagonal entries, that the number of diagonal entries in $A_\sigma^r$ and $B^r_\bsigma$ is always bounded by $(2q+3)$ (even by $(2q+1)$ in the $B^r_\bsigma$ terms). Hence the total number of resolvent entries in $A^r_\sigma$ and in $B^r_\bsigma$ (diagonal and off-diagonal, in the numerator and in the denominator) is bounded by $4q+5$.

\medskip

Following \cite{ekyy-rxv12}, we denote by $b(\sigma)$ and $b(\bsigma)$ the number of ones in the strings $\sigma$ and $\bsigma$, respectively. For any $r \in \{1, \dots , 2q \}$ and for any $\sigma\in\LL_r$ and $\bsigma\in\MM_r$, we have:
 \begin{enumerate}
\item \label{f:1} If $b(\sigma)\geq1$ then  $\od(A_\sigma^r)\geq b(\sigma)+1$. Analogously, if $b(\rho)\geq 1$ then $\od(B_\bsigma^r)\geq b(\bsigma)+1$. This follows from the observation that the first application of $w_1$ generates (according to (\ref{wd1}) and (\ref{wd2})) two off-diagonal entries while all further applications create (according to (\ref{wod}), (\ref{wd1}), (\ref{wd2})) at least one additional off-diagonal entry.
\item \label{f:2} $b(\sigma)\leq 2q$ and $b(\bsigma)\leq 2q$. The first application of $w_1$ generates two off-diagonal terms. Each subsequent application of $w_1$ generates at least one more off-diagonal term. Hence $2q$ applications of $w_1$ create at least $(2q+1)$ off diagonal terms, which are enough to satisfy the stopping rule $(\SR)$.
% \item By \eqref{ddod1} and \eqref{ddod2}, the total number of resolvent entries in $A^r_\sigma$ %and $B^r_\bsigma$ is bounded by $4q+5$.
\item The number of zeros in the strings $\sigma$ and in $\bsigma$ is bounded by $2q(4q+5)$.
This follows because every application of $w_0$ produces one additional ``top'' index (indicating that the generated entry is independent of on additional row and column). The total number of ``top'' indices is bounded, however, by the number of resolvent entries (which, as discussed above, is at most $(4q+5)$) times $2q$ (the number of coordinates of $\kk$).
\item The length of $\sigma$ and $\bsigma$ is bounded by $4q (2q +3)$ (this follows combining the bounds for the number of zeros and the number of ones in the strings).
 \item Considering that the length of $\sigma$ and $\bsigma$ is at most $4q (2q+3)$ and that the number of ones is at most $2q$, we can bound the cardinality of $\LL_r$ and $\MM_r$ by
 \begin{equation}\label{eq:car-LM}  |\LL_r|, |\MM_r| \leq \sum_{k=0}^{2q} \binom{4q(2q+3)}{k} \leq (Cq)^{2q} \,. \end{equation}
\end{enumerate}

%QUI

\subsection{Preliminaries to the proof of Eq. \eqref{555}: bounds on resolvent entries}
\label{sec:bd-lm}

We are going to use (\ref{expanded}) and (\ref{expandedt}) to expand the initial resolvent entries $A^r = 1/ G_{k_r k_r}$ and $B^r = G_{k_r k_r}$. To prove a bound of the form (\ref{555}), we need to estimate the resolvent entries appearing in the expanded terms $A_\sigma^r$ and $B_\bsigma^r$. More precisely, after applying H\"older's inequality to separate the many factors in the products $A_\sigma^r$, $B_\bsigma^r$, we will need control on high moments of quantities of the form
\begin{equation}\label{eq:4terms}  |G_{kk}^{(\T)}|, \quad \frac{1}{|G_{kk}^{(\T)}|}, \quad \left| (\ID-\E_k) \frac{1}{G_{kk}^{(\T)}} \right|, \quad |G_{kl}^{(\T)}| \,. \end{equation}
While the first two terms in (\ref{eq:4terms}) are typically of order one, the last two are expected to be small. To show (\ref{555}), it is important to extract the correct small factor in the bounds for these quantities.

\medskip

Bounds for moments of $|G_{kk}|$ have already been obtained in Lemma \ref{boundG11}. Similarly one can also estimate moments of $|G_{kk}^{(\T)}|$, for $\T \subset \{ 1,\dots , N \}$ with $|\T|  \leq 2q$ and $k \not \in \T$.

\begin{lemma}\label{l:Gkkinv}
 Assume (\ref{e:gd}), fix $\tilde E>0$,  $\tildeeta> 0$. Set $z = E+i\eta$. There exist constants $c, C,M_1, M_2 > 0$ such that
\[
% \begin{equation}\label{Gkkinv}
\E \frac{1}{|G_{11}^{(\T)}|^{2q}} \leq C^q\,
%\end{equation}
\]
for all $|E|\leq \tilde E$, $\eta \leq \tildeeta$, $N > M_1$ such that $N\eta \geq M_2$, $q \in \NA$ with $q \leq c (N\eta)^{1/4}$ and $\T \subset \{ 1, \dots , N \}$ with $|\T|\leq 2q$.
\end{lemma}

{\it Remark.}  For $\T \subset \{ 1, \dots , N \}$ with $|\T| \leq 2q$, let $\Lambda^{(\T)}  = m^{(\T)} - m_{sc}$, where $m^{(\T)}$ is the Stieltjes transform of $H^{(\T)}$. By the interlacing properties of the eigenvalues of $H$ and $H^{(\T)}$ it is easy to check that  $|m - m^{(\T)}| \leq C |\T| / (N\eta)$, which implies
\begin{equation}\label{T1}
 \E |\La^{(\T)}|^{q} \leq C^q\E |\La|^{q} + \frac{(Cq)^q}{(N\eta)^{q}}   \qquad\textrm{and} \qquad  \E |\Im \La^{(\T)}|^{q} \leq C^q\E |\Im \La|^{q} + \frac{(Cq)^q}{(N\eta)^{q}}  \, .
\end{equation}
We will use this observation to reduce our analysis to the case $\T = \emptyset$.

%PROOF
\begin{proof}
By the above remark, we can take $\T=\emptyset$.  We have
\[
%\begin{equation}\label{22}
\E \frac{1}{|G_{11}|^{2q}} = \E | h_{11} - z - \aa_1^*  G^{(1)}\aa_1 |^{2q}
\leq \frac{(Cq)^{q}}{N^{q}}+ C^q + C^q \, \E \, |\aa_1^* G^{(1)}\aa_1|^{2q}
\leq  C^q + C^q \, \E \, |\aa_1^* G^{(1)}\aa_1|^{2q}\,.
%\end{equation}
\]

\medskip

We write
\begin{align}
\E |\aa_1^* G^{(1)}\aa_1|^{2q} &= \E |\aa_1^* G^{(1)}\aa_1- \E_1\aa_1^* G^{(1)}\aa_1 + \E_1\aa_1^* G^{(1)}\aa_1|^{2q}
\nonumber
\\
& \leq C^q( \E |\aa_1^*  G^{(1)}\aa_1- \E_1\aa_1^* G^{(1)}\aa_1|^{2q} + \E| \E_1\aa_1^*  G^{(1)}\aa_1|^{2q}) \,.
\label{e:1}
\end{align}

\medskip

For the first term we use the Hanson-Wright large deviation estimate (see Prop. \ref{p:Zkq}). We find
\[
 \E |\aa_1^*  G^{(1)}\aa_1- \E_1\aa_1^* G^{(1)}\aa_1|^{2q} \leq
 (Cq)^{2q} \left(\frac{(\Im m_{sc})^{q} + \E|\Im\Lambda|^{q}}{(N\eta)^{q}}+\frac1{(N\eta)^{2q}} \right)
 \leq 1
\]
since, by Lemma \ref{l:no}, $ \E|\Im\Lambda|^{q} <1$ for all $q \leq (N\eta)^{1/4}$, and $N$ large enough.

\medskip

As for the second term on the r.h.s. of \eqref{e:1}, we find
\[
%\begin{equation}\label{33}
\E |\E_1\aa_1^* G^{(1)}\aa_1|^{2q}
= \E \left|\frac{1}{N} \Tr G^{(1)}\right|^{2q}
\leq  C^q\left( \E \left|\frac{1}{N} \Tr G\right|^{2q} +\frac{1}{(N\eta)^{2q}}\right)\leq C^q
%\end{equation}
\]
from $ \E \left|\frac{1}{N} \Tr G\right|^{2q}  \leq  \E |G_{11}|^{2q}\leq C^q$ (by Lemma \ref{l:boundG11}).
\end{proof}

\medskip

To estimate the third term in (\ref{eq:4terms}), we recall that (similarly to (\ref{eq:Zj}))
\[ (\ID - \E_k)\frac{1}{G_{kk}^{(\T)}}= - Z_k^{(\T)} , \]
where we set $Z_k^{(\T)} = (\ID -\E_k){\aa_k^{(\T)*}} G^{(\T)} \aa_k^{(\T)}$, and where $\aa^{(\T )}_l$ is the $l$-th  column of the matrix $H$ without the elements corresponding to the labels in $\T$. Applying Proposition \ref{p:Zkq} to $Z_k^{(\T)}$ and using the remark after  Lemma \ref{l:Gkkinv}, we find that (for $q \leq N/2$)
\[
%\begin{equation}\label{boundZk-2}
\begin{split}
\E \, \left| (\ID - \E_k)\frac{1}{G_{kk}^{(\T)}} \right|^{2q} &\leq (Cq)^{2q} \left(\frac{(\Im m_{sc})^{q} + \E|\Im\Lambda^{(\T)}|^{q}}{((N-|\T|)\eta)^q}  +\frac{1}{((N-|\T|)\eta)^{2q}}\right)
%\\
%&\leq (Cq)^{2q} \left(\frac{(\Im m_{sc})^{q} + \E|\Lambda|^{q}}{(N\eta)^q}  +\frac{1}{(N\eta)^{2q}}\right) \\ &
\leq  (Cq)^{3q} \EE_q \,.
\end{split}
%\end{equation}
\]

\medskip

Finally, we show how to estimate the last term in (\ref{eq:4terms}). To this end, we use the formula (see, e.g., Eq. (2.8) in \cite{E})
\begin{equation}\label{522} G^{(\T)}_{kl} = G^{(\T)}_{ll} G^{(\T l)}_{kk} (h_{kl} - \aa^{(\T l)*}_k G^{(\T kl)} \aa^{(\T k)}_l ) = G^{(\T)}_{ll} G^{(\T l)}_{kk} K_{kl}^{(\T)}  \end{equation}
valid for any $k\neq l$, $k,l\in \{1,...,N\}\backslash \T$, with
\begin{equation}\label{eq:Kkl} K_{kl}^{(\T)} = (h_{kl} - \aa^{(\T l)*}_k G^{(\T kl)} \aa^{(\T k)}_l ) \, . \end{equation}
High moments of the diagonal entries $G^{(\T)}_{ll}$ and $G^{(\T l)}_{kk}$ can be bounded with Lemma \ref{boundG11}. High moments of $K_{kl}^{(\T)}$, on the other hand, are controlled (and shown to be small) in the next lemma.
\begin{lemma}\label{l:G12}
Assume (\ref{e:gd}), set $z = E+i\eta$ and let $K_{kl}^{(\T)}$ be defined as in (\ref{eq:Kkl}) (with $\T = \emptyset$). Then there exist constants $c,c_0, C, M_1, M_2> 0$ such that
\begin{equation}\label{bK}
\E |K_{kl}^{(\T)} |^{2q}
\leq (Cq)^{cq} \EE_q \, ,
\end{equation}
for all $E\in\RE$, $N > M_1$, $\eta >0$ with $N \eta > M_2$, $k \not = l \in\{1,....,N\}$, $q \in \NA$ with $q \leq c_0 N$.
\end{lemma}

% \begin{remark} Statement remains true if we remove some rows and column. If $\T \subset \{1,...,N\}$ such that $|\T|\leq 2q$,  $k\in\{1,....,N\}\backslash \T$, we have
% \[ \E |K^{(\T)}_{kl} |^{2q} \leq (Cq)^{2q} \EE_q  \]
% where $K^{(\T)}_{kl}$ is defined as in (\ref{eq:Kkl}).
% \end{remark}

%PROOF
\begin{proof} By the remark following Lemma \ref{l:Gkkinv} we can restrict our attention to the case $\T=\emptyset$.

By the definition of $K_{kl}$ we get
\[
\E |K_{kl}|^{2q} \leq  C^{q}\left(\E |h_{kl}|^{2q} + \E| \aa_k^{(l)*} G^{(kl)} \aa_l^{(k)}|^{2q}\right)
\leq C^{q}  \left(\frac{(Cq)^{q}}{N^{q}} + \E| \aa_k^{(l)*} G^{(kl)} \aa_l^{(k)}|^{2q}\right)\,.
\label{a1}
\]
Let $\g^{(k)} = G^{(kl)}\aa_l^{(k)}$. Then
\[
 \E|\aa_k^{(l)*} G^{(kl)} \aa_l^{(k)}|^{2q} = \E| (\aa_k^{(l)},\g^{(k)})|^{2q} \leq   C^{q}\left( \E\left|| (\aa_k^{(l)},\g^{(k)})|^{2} - \frac{\|\g^{(k)}\|^2}{N} \right|^q
+ \E\frac{\|\g^{(k)}\|^{2q}}{N^q}\right).
\]
Noticing that $ \E_k| (\aa_k^{(l)},\g^{(k)})|^{2}= N^{-1} \, \|\g^{(k)}\|^2$, the Hanson-Wright large deviation estimate (see Prop.~\ref{p:HW}) implies that
\begin{align}
  \E\left|| (\aa_k^{(l)},\g^{(k)})|^{2} - \frac{\|\g^{(k)}\|^2}{N} \right|^q
\leq
(Cq)^q \, \E \, \frac{\|\g^{(k)}\|^{2q}}{N^q} \,.
\end{align}
Hence
\begin{align}
\label{a2}
  \E \, |\aa_k^{(l)*} G^{(kl)} \aa_l^{(k)}|^{2q} \leq & (Cq)^q \, \E\frac{\| G^{(kl)}\aa_l^{(k)}\|^{2q}}{N^q} .
\end{align}
Noticing that $\E_l \| G^{(kl)}\aa_l^{(k)}\|^{2} = N^{-1} \, \Tr|G^{(kl)}|^2$, applying again Prop. \ref{p:HW}, we conclude that
\begin{align}
 \E \, \frac{\| G^{(kl)}\aa_l^{(k)}\|^{2q}}{N^q} \leq & C^q \left( (Cq)^q \, \E \left(\frac{\Tr|G^{(kl)}|^4}{N^4}\right)^{q/2}+ \E\left(\frac{\Tr|G^{(kl)}|^2}{N^2}\right)^{q}\right)
 \nonumber
 \\
 \leq &  (Cq)^q \left(\frac{(\Im m_{sc})^q +\E|\Im \La^{(kl)}|^q}{(N\eta)^q} + \frac{(\Im m_{sc})^\frac{q}2 +\E|\Im\La^{(kl)}|^\frac{q}2}{(N\eta)^{\frac32 q}}  \right)\,,
\label{a3}
\end{align}
where we used the bound $\Tr|G^{(kl)}|^4 \leq \eta^{-2} \Tr|G^{(kl)}|^2$ and the estimate
\[ \frac{1}{N^2} \Tr \, |G^{(kl)}|^2 =\frac{1}{N\eta} \Im m^{(kl)} \leq \frac{1}{(N\eta)^2} + \frac{\text{Im } m_{sc} + |\Im\Lambda^{(k\ell)}|}{N\eta} \]
proven as in (\ref{eq:trG2}).
%Here we used the notation $\La^{(\T)} = m^{(\T)} - m_{sc}$ with $m^{(\T)}$ being the Stieltjes transform of the minor of the original Wigner matrix $H$ obtained after removing rows and columns with indices in $\T$.
Inserting \eqref{a3} into \eqref{a2} and then into \eqref{a1}, we find
\begin{align}\label{eq:Kkl-2}
\E |K_{kl}|^{2q} \leq (Cq)^{2q} \left( \frac{(\Im m_{sc})^q +\E|\Im\La^{(kl)}|^q}{(N\eta)^q} + \frac{(\Im m_{sc})^\frac{q}2 +\E|\Im\La^{(kl)}|^\frac{q}2}{(N\eta)^{\frac32 q}}+ \frac{1}{N^{q}} \right)\,.
\end{align}
{F}rom the interlacing properties of the eigenvalues of $H$ and of its minors, it is easy to check that $|m-m^{(kl)}| \leq C/ (N\eta)$. This implies that
\[
 \E |\Im\La^{(kl)}|^{q} \leq  C^q\left( \E |\Im\La|^{q} + \E |\Im m-\Im m^{(kl)}|^{q} \right) \leq
C^q\left( \E |\Im\La|^{q} + \frac{1}{(N\eta)^{q}}\right) \,,
\]
and hence, from (\ref{eq:Kkl-2}), that
\[
%\begin{equation}\label{5.50}
\E |K_{kl}|^{2q} \leq (Cq)^{2q} \left(\frac{1}{(N\eta)^{2q}} + \frac{(\Im m_{sc})^q +\E|\Im\La|^q}{(N\eta)^q} + \frac{(\Im m_{sc})^\frac{q}2 +\E|\Im \La|^\frac{q}2}{(N\eta)^{\frac32 q}}+ \frac{1}{N^{q}} \right)\,.
\]
%\end{equation}
To conclude the proof, we observe that
\[
%\begin{equation}\label{5.50a}
\frac{(\Im m_{sc})^\frac{q}2 +\E|\Im\La|^\frac{q}2}{(N\eta)^{\frac32 q}} \leq \frac{1}{(N\eta)^q} \sqrt{ \frac{(\text{Im } m_{sc})^q + \E \, |\Im\La|^q}{(N\eta)^q}} \leq \frac{1}{(N\eta)^{2q}} + \frac{(\Im m_{sc})^q +\E|\Im\La|^q}{(N\eta)^q}\,. 
%\end{equation}
\]
\end{proof}

%QUI
\subsection{Proof of Eq. \eqref{555}}

For a fixed $\kk$ we use \eqref{expanded} and \eqref{expandedt}
to expand $V(\kk)$ in Eq. \eqref{panama} as
\[
V(\kk)= \sum_{\substack{\sigma_1 \in \LL_1\\ \bsigma_1 \in \MM_1}} \dots \sum_{\substack{\sigma_{2q} \in \LL_{2q} \\ \bsigma_{2q} \in \MM_{2q}}}
\E ( (\ID-\E_{k_1}) ( (\ID-\E_{k_1}) A^1_{\sigma_1})B^1_{\bsigma_1}) \cdots \overline{((\ID-\E_{k_{2q}})( (\ID-\E_{k_{2q}}) A^{2q}_{\sigma_{2q}}) B^{2q}_{\bsigma_{2q}})}\,.
\]
We claim that
\begin{equation}\label{claim1}
\left|\E ((\ID-\E_{k_1}) ( (\ID-\E_{k_1})A^1_{\sigma_1}) B^1_{\bsigma_1}) \cdots \overline{( (\ID-\E_{k_{2q}}) ((\ID-\E_{k_{2q}}) A^{2q}_{\sigma_{2q}}) B^{2q}_{\bsigma_{2q}}})\right| \leq
 (Cq)^{cq^2} \EE_{2q + |L(\kk)|}^\frac12
\end{equation}
for any $\kk$,  $\sigma_1 \in \LL_1,\bsigma_1 \in \MM_1,  \dots ,\sigma_{2q} \in \LL_{2q},\bsigma_{2q} \in \MM_{2q}$ and for all $q \leq c_0  (N\eta)^{1/8}$ and $(N\eta)$ and $N$ large enough. Here $|L(\kk)|$ is the number of lone labels associated with the vector $\kk$. Using the bounds (\ref{eq:car-LM}) on the cardinality of the sets $\LL_j, \MM_j$, (\ref{claim1}) implies that
\[%\begin{equation}\label{final}
|V(\kk)|\leq  (Cq)^{cq^2} \EE_{2q + |L(\kk)|}^\frac12 \,.
%\end{equation}
\]
{F}rom Eq. \eqref{panama0}, we get
\begin{equation}
\label{cat}\begin{split}
 \E\left| \frac1N \sum_k(\ID - \E_k)W_k\right|^{2q}  \leq &  \frac{(Cq)^{cq^2}}{N^{2q}} \sum_{P\in\fP_{2q}}\sum_{\kk}\II(\PP(\kk)=P) \, \EE_{2q + |L(\kk)|}^\frac12  \\
\leq & (Cq)^{cq^2} \sum_{P\in\fP_{2q}}  \frac{1}{N^{2q-|P|}} \, \EE_{2q + |L(P)|}^\frac12 \,.
%\\
%\leq  & (Cq)^{cq^2} \EE_{4q}^{\frac12} \,,
\end{split} \end{equation}
Here we wrote $|L(P)|$ to indicate the number of lone labels associated with any vector $\kk$ such that $\PP (\kk) = P$ (it is clear that the number of lone labels only depends on the partition $P$), and we used the fact that
\[ \left| \{ \kk : \PP (\kk) = P \} \right| \leq N^{|P|} \]
where $|P|$ denotes the size of the partition $P$ (the number of sets in which $\{ 1, \dots , 2q \}$ is divided). Since, by definition, $\EE_q \geq N^{-q}$, we have
\begin{equation}\label{eq:EEbd} \frac{1}{N^{2q-|P|}} \, \EE_{2q + |L|}^\frac12 \leq \EE_{2q-|P|} \, \EE_{2q + |L|}^\frac12 \leq C \EE_{4q}^{\frac{2q-|P|}{4q}} \, \EE_{4q}^{\frac12\left(\frac{2q+|L|}{4q}\right)} =C \EE_{4q}^{\frac12+ \frac{q+\frac{|L|}{2}-|P|}{4q}} \leq C \EE_{4q}^\frac12 \end{equation}
where we used (\ref{eq:EEpq}) and the bound
\[ |P|\leq |L| +\frac{2q-|L|}{2} \quad \Rightarrow \quad  q+\frac{|L|}{2}-|P| \geq 0\,.\]
Inserting (\ref{eq:EEbd}) into (\ref{cat}), and using the fact that the number of partitions of $2q$ elements is bounded by $(Cq)^q$, we obtain that
\[ \E\left| \frac1N \sum_k(\ID - \E_k)W_k\right|^{2q}  \leq (Cq)^{cq^2} \EE_{4q}^{\frac12} \]
which concludes the proof of (\ref{555}).

\medskip

Next, we prove Eq. \eqref{claim1}, for fixed $\kk$, $\sigma_1, \bsigma_1,  \dots , \sigma_{2q},\bsigma_{2q}$. We distinguish two cases.

\medskip

{\bf Case 1}: For any  $r \in\{1,...,2q\}$ all resolvent entries in  $A^r_{\sigma_r}$ and $B^r_{\rho_r}$ are maximally expanded.

\medskip

The l.h.s. of \eqref{claim1} vanishes if there exists a lone label $r$ such that, for all $s\in\{1,...,2q\}\backslash\{r\}$, the monomials $A^s_{\sigma_s}$ and $B^s_{\rho_s}$ do not depend on
$r$, i.e. if there exists a lone label $r$ such that $k_r$ appears as an upper index in all resolvent entries in $A^s_{\sigma_s}$, $B^s_{\rho_s}$, for all $s\in\{1,...,2q\}\backslash\{r\}$.

\medskip

For $r\in L(\kk)$, we say, following \cite{ekyy-rxv12}, that $s\in\{1,...,2q\}\backslash\{r\}$ is a \emph{partner label} of $r$ if $k_r$ appears as lower index in one of the resolvent entries of $A^s_{\sigma_s}$ or of $B^s_{\rho_s}$. Since, by the definition of the expansion algorithm, the same index cannot appear both as a lower and as an upper index in the same resolvent entry, we conclude that non-zero terms on the l.h.s. of  Eq. \eqref{claim1} are characterized by the fact that for every lone label $r\in L(\kk)$ there exists at least one partner label $s$.

\medskip

We denote by $\ell(s)$ the number of lone labels having $s$ as a partner. In order for the l.h.s. of \eqref{claim1} not to vanish, we must have (as remarked in \cite{ekyy-rxv12})
\[ \sum_{s=1}^{2q}\ell(s)\geq |L(\kk)|. \]
This follows because all resolvent entries are maximally expanded and therefore, for every lone label $r$, and every $s \in \{1, 2 , \dots , 2q \} \backslash \{ r \}$, $k_r$ either appears as lower label in at least one resolvent entry in $A^s_{\sigma_s}$ or $B^s_{\bsigma_s}$ (in which case, $s$ has $r$ as a partner), or every resolvent entry in $A^s_{\sigma_s}$ or $B^s_{\bsigma_s}$ has $r$ as an upper index. This is the point where the assumption that all entries are maximally expanded is used.

\medskip

The combined number of lower indices different from $k_s$ in $A^s_{\sigma_s}$ and $B^s_{\rho_s}$ is at least $\ell (s)$. Since the operation $w_1$ produces only one additional lower index, while
$w_0$ does not produce new lower labels, we find that
\[%\begin{equation}\label{change}
 b(\sigma_s)+b(\rho_s) \geq \ell(s)\,.
%\end{equation}
\](Recall that $b(\sigma)$ denotes the number of ones occurring in the string $\sigma$). We notice, moreover, that for any $\T\subseteq\{1,...,2q\}$ the operation $w_1$, applied on a diagonal entry like $G_{kk}^{(\T)}$ or $1/G_{kk}^{(\T)}$, adds a new lower index, associated with a lone label from $L(\kk)$. Since any subsequent operation will not remove the new lower index, $\ell (s) = 0$ if and only if the strings $\sigma_s$ and $\bsigma_s$ contain only zeros. Since we assumed all entries to be maximally expanded, we have that
\begin{equation}\label{543}
 \ell(s)=0 \quad \textrm{if and only if} \quad A^s_{\sigma_s} = \frac1{G_{ss}^{(L(\kk)\backslash\{s\})}}\;\textrm{ and }\;B^s_{\rho_s} = G_{ss}^{(L(\kk)\backslash\{s\})}\,.
\end{equation}

\medskip

Let us assume that in the term on the l.h.s. of (\ref{claim1}) there are $t$ labels $s_1, \dots , s_t$ with $\ell (s_j) = 0$. Then there will be $u=2q-t$ labels $r_1, \dots , r_u$ with $\ell (r_j) \geq1$. Without loss of generality we can assume the first $t$ labels in (\ref{claim1}) to be $s_1, \dots , s_t$ and the last $u$ labels to be $r_1 ,\dots , r_u$ (the fact that some terms are complex conjugated is irrelevant for our argument). For $s = 1, \dots , t$ (so that $\ell (s) = 0$), we observe that
\begin{equation}\label{boundm}
\begin{split}
| (\ID-\E_{s}) ( (\ID-\E_{s}) &A^s_{\sigma_s})B^s_{\bsigma_s}|  \\ &\leq |((\ID-\E_s) A^s_{\sigma_s}) B^s_{\bsigma_s}| + \E_s |(\ID-\E_s) A^s_{\sigma_s}) B^s_{\bsigma_s}| \\ & \leq  |(\ID-\E_s) A^s_{\sigma_s}| |B^s_{\bsigma_s}| + ( \E_s |(\ID - \E_s) A^s_{\sigma_s}|^2 )^{1/2} (\E_s |B^s_{\bsigma_s}|^2)^{1/2} \\ &\leq C m_s  \widetilde m_s
\end{split}
\end{equation}
where we defined the random variables
\begin{equation} \label{m}
\begin{split}
 m_s& =|(\ID-\E_{s})A^s_{\sigma_s}|+(\E_{s} |(\ID-\E_{s})A^s_{\sigma_s}|^2)^\frac12 \, , \\
 \widetilde m_s &=|B^s_{\bsigma_s}|+(\E_{s} |B^s_{\bsigma_s}|^2)^\frac12 \,
\end{split}
\end{equation}
for all $s=1, \dots , t$. We will estimate high moments of $m_s, \widetilde{m}_s$ using the bounds established in Section~\ref{sec:bd-lm}. Notice that, while $\widetilde{m}_s$ is only bounded, the quantity $m_s$ is expected to be small. It is important that the estimates that we will use reflect this smallness.

\medskip

To bound the terms on the l.h.s. of (\ref{claim1}) with $s=t+1, \dots , 2q$ we proceed as follows. Since here $\ell (s) \geq 1$, either $A^s_{\sigma_s}$ or $B^s_{\rho_s}$ must contain at least one off-diagonal entry. Using (\ref{522}), we can express the off-diagonal entries in terms of the variable $K_{kl}^{(\T)}$, defined in Eq. \eqref{eq:Kkl}. We collect all the factors $K_{kl}^{(\T)}$ appearing in $A^s_{\sigma_s}$ (one for every off-diagonal resolvent entry) in a single monomial $O_s$ and all the diagonal entries (appearing in $A^s_{\sigma_s}$, either in the numerator or in the denominator) in a monomial $P_s$; accordingly $A^s_{\sigma_s} = O_s P_s$. We proceed similarly for $B^s_{\rho_s}$ and we obtain $B^s_{\rho_s}=\widetilde O_s \widetilde P_s$. We estimate
\begin{equation}\label{boundM}
\begin{split} | (\ID-\E_{s})(&(\ID-\E_{s}) A^s_{\sigma_s})B^s_{\bsigma_s}| \\ = \; &\left|(\ID-\E_{s}) ((\ID-\E_{s}) O_s P_s) \widetilde{O}_s \widetilde{P}_s \right|  \\
\leq \; &|O_s| |\widetilde{O}_s | |P_s | |\widetilde{P}_s | + |\widetilde{O}_s \widetilde{P}_s  \, \E_s (O_s P_s ) | + | \E_s (O_s P_s \widetilde{O}_s \widetilde{P}_s)| + |\E_s (O_s P_s) \, \E_s (\widetilde{O}_s \widetilde{P}_s)| \\
 \leq \; &|O_s| |\widetilde{O}_s| |P_s| |\widetilde{P}_s| + |\widetilde{O}_s|  |\widetilde{P}_s| (\E_s |O_s|^2)^{1/2}  (\E_s  |P_s|^2)^{1/2}  \\ &+
(\E_s |O_s \widetilde{O}_s|^2)^{1/2} (\E_s |P_s \widetilde{P}_s|^2)^{1/2} + (\E_s |O_s|^2)^{1/2} (\E_s |P_s|^2)^{1/2} (\E_s |\widetilde{O}_s|^2)^{1/2} (\E_s |\widetilde{P}_s|^2)^{1/2} \\
\leq \; & C M_s \widetilde M_s
\end{split}
\end{equation}
for all $s=t+1, \dots , 2q$. Here we defined the random variables
\begin{align*}
 M_s &= |O_s| |\widetilde O_s|+ (\E_{s}|O_s|^2)^\frac12 |\widetilde O_s| +(\E_{s}|O_s\widetilde O_s|^2)^\frac12+(\E_{s}|O_s|^2)^\frac12(\E_{s}|\widetilde O_s|^2)^\frac12 \\
 \widetilde M_s &= |P_s| |\widetilde P_s|+ (\E_{s}|P_s|^2)^\frac12 |\widetilde P_s|+(\E_{s}|P_s\widetilde P_s|^2)^\frac12+(\E_{s}|P_s|^2)^\frac12(\E_{s}|\widetilde P_s|^2)^\frac12 \,.
\end{align*}
Combining (\ref{boundm}) with (\ref{boundM}), and applying Cauchy-Schwarz's inequality, we find
\begin{equation}\label{eq:mMmM} \begin{split}
\Big| \E ((\ID-\E_{k_1}) ((\ID-\E_{k_1}) &A^1_{\sigma_1})B^1_{\bsigma_1}) \cdots \overline{((\ID-\E_{k_{2q}}) ((\ID-\E_{k_{2q}}) A^{2q}_{\sigma_{2q}})B^{2q}_{\bsigma_{2q}})}\Big| \\
\leq &C \, (\E  \, m^2_1\cdots m^2_t M^2_{t+1} \cdots M^2_{2q})^\frac12 \, (\E \, \widetilde m^2_1\cdots\widetilde m^2_t \widetilde M^2_{t+1} \cdots \widetilde M^2_{2q} )^\frac12\,.
\end{split}\end{equation}
We bound the first factor on the r.h.s. of the last equation. Recall that $O_s$ and $\widetilde O_s$ are the product of respectively $\gamma_s := \gamma(A^{s}_{\sigma_s})$ and $\tilde\gamma_s := \gamma(B^{s}_{\sigma_s})$ entries of the form $K_{kl}^{(\T)}$. With H\"older's inequality we obtain
\begin{equation}\label{eq:holder} \E \, m_1^2 \dots m_t^2 M_{t+1}^2 \dots M_{2q}^2 \leq \prod_{s=1}^t  \left[ \E \, m_s^{2(t+\sum_s (\gamma_s+ \widetilde{\gamma}_s))} \right]^{\frac{1}{t+\sum_s (\gamma_s+ \widetilde{\gamma}_s)}} \, \prod_{s=t+1}^{2q} \left[ \E \, M_s^{2 \frac{ t+\sum_s (\gamma_s+ \widetilde{\gamma}_s)}{\gamma_s + \widetilde{\gamma}_s}} \right]^{\frac{\gamma_s + \widetilde{\gamma}_s}{t+\sum_s (\gamma_s+ \widetilde{\gamma}_s)}}\,. \end{equation}
According to (\ref{543}) we have, for $s=1, \dots , t$,
\begin{equation}\label{eq:ms2} m_s^2 \leq C \left| (\ID-\E_s) \frac{1}{G_{k_s k_s}^{(\T_s)}} \right|^2 + C \, \E_s \left| (\ID-\E_s) \frac{1}{G_{k_s k_s}^{(\T_s)}} \right|^2
\end{equation}
where we set $\T_s = L(\kk) \backslash \{ s \}$. This implies that
\begin{equation}\label{eq:m2s} \E \, m_s^{2(t+\sum_s (\gamma_s + \wt{\gamma}_s))} \leq C^{t+\sum_s (\gamma_s + \wt{\gamma}_s)} \,  \E \,   \left| (\ID-\E_s) \frac{1}{G_{k_s k_s}^{(\T_s)}} \right|^{2(t+\sum_s (\gamma_s + \wt{\gamma}_s))} \,. \end{equation}

\medskip

To bound the $M_s$ variables, we observe that, for any $s = t+1,\dots ,2q$,
\begin{equation}\label{sat}
M_s^2  \leq  C (|O_s |^2|\widetilde O_s|^2+ (\E_{s}|O_s|^2) |\widetilde O_s|^2 +(\E_{s}|O_s\widetilde O_s|^2)+(\E_{s}|O_s|^2)(\E_{s}|\widetilde O_s|^2) )\,.
\end{equation}
Inserting the contribution of
\[ |O_s|^2 |\widetilde{O}_s|^2 = \prod_{j=1}^{\gamma_s + \wt{\gamma}_s} |K^{(\T_j)}_{i_j l_j}|^2 \]
(for appropriate indices $i_j,l_j$ and sets $\T_j$) in the expectation on the r.h.s. of (\ref{eq:holder}), we find, again by H\"older's inequality,
\[ \left[ \E \, (|O_s| |\widetilde{O}_s|)^{2 \frac{ t+\sum_s (\gamma_s+ \widetilde{\gamma}_s)}{\gamma_s + \widetilde{\gamma}_s}} \right]^{\frac{\gamma_s + \widetilde{\gamma}_s}{t+\sum_s (\gamma_s+ \widetilde{\gamma}_s)}} \leq \prod_{j=1}^{\gamma_s + \wt{\gamma}_s}  \left[ \E \, |K^{(\T_j)}_{i_j l_j}|^{2(t+ \sum_s (\gamma_s + \wt{\gamma}_s))} \right]^{\frac{1}{t+\sum_s (\gamma_s + \wt{\gamma}_s)}}\,. \]
The contribution of the second term on the r.h.s. of (\ref{sat}), having the form
\[ (\E_s |O_s|^2) |\wt{O}_s|^2 = \left(\E_s \prod_{j=1}^{\gamma_s} |K_{i_j l_j}^{(\T_j)}|^2 \right) \prod_{j=\gamma_s+1}^{\gamma_s + \wt{\gamma}_s} |K_{i_j l_j}^{(\T_j)}|^2
\leq \prod_{j=1}^{\gamma_s}  \left(\E_s |K_{i_j l_j}^{(\T_j)}|^{2\gamma_s} \right)^{\frac{1}{\gamma_s}}
 \prod_{j=\gamma_s+1}^{\gamma_s + \wt{\gamma}_s} |K_{i_j l_j}^{(\T_j)}|^2 \, ,
\]
to (\ref{eq:holder}) can be estimated by
\[ \begin{split}
\E \Big[ (\E_s |O_s|^2)^{\frac{t+\sum_s (\gamma_s + \wt{\gamma}_s)}{\gamma_s + \wt{\gamma}_s}} &|\wt{O}_s|^{2 \frac{t+\sum_s (\gamma_s + \wt{\gamma}_s)}{\gamma_s + \wt{\gamma}_s}}  \Big]
\\ &\leq \E \,  \prod_{j=1}^{\gamma_s}  \left(\E_s |K_{i_j l_j}^{(\T_j)}|^{2\gamma_s} \right)^{\frac{t+ \sum_s (\gamma_s + \wt{\gamma}_s)}{\gamma_s (\gamma_s + \wt{\gamma}_s)}} \prod_{j=\gamma_s+1}^{\gamma_s + \wt{\gamma}_s} |K_{i_j l_j}^{(\T_j)}|^{2 \frac{t+ \sum_s (\gamma_s + \wt{\gamma}_s)}{\gamma_s + \wt{\gamma}_s}} \\ &\leq \prod_{j=1}^{\gamma_s} \left[ \E (\E_s |K_{i_j l_j}^{(\T_j)}|^{2\gamma_s} )^{\frac{t+\sum_s (\gamma_s + \wt{\gamma}_s)}{\gamma_s}} \right]^{\frac{1}{\gamma_s + \wt{\gamma}_s}} \prod_{j=\gamma_s+1}^{\gamma_s+\wt{\gamma}_s} \left[ \E \, |K_{i_j l_j}^{(\T_j)}|^{2(t+\sum_s (\gamma_s + \wt{\gamma}_s))} \right]^{\frac{1}{\gamma_s + \wt{\gamma}_s}}
\\ &\leq \prod_{j=1}^{\gamma_s + \wt{\gamma}_s}  \left[ \E \, |K_{i_j l_j}^{(\T_j)}|^{2(t+\sum_s (\gamma_s + \wt{\gamma}_s))} \right]^{\frac{1}{\gamma_s + \wt{\gamma}_s}}\,.
\end{split} \]
The contributions from the last two terms on the r.h.s. (\ref{sat}) can also be bounded similarly. We obtain
\[ \left[ \E \, M_s^{2 \frac{t+\sum_s (\gamma_s + \wt{\gamma}_s)}{\gamma_s + \wt{\gamma}_s}} \right]^{\frac{\gamma_s + \wt{\gamma}_s}{t+\sum_s (\gamma_s + \wt{\gamma}_s)}} \leq C \prod_{j=1}^{\gamma_s + \wt{\gamma}_s} \left[ \E \, |K_{i_j l_j}^{(\T_j)}|^{2(t+\sum_s (\gamma_s + \wt{\gamma}_s))} \right]^{\frac{1}{t+\sum_s (\gamma_s + \wt{\gamma}_s)}}\,. \]
Inserting (\ref{eq:m2s}) and the last equation into (\ref{eq:holder}), we conclude that
\begin{equation}\label{eq:EmMs} \begin{split} \E \, m_1^2 \dots m_t^2 M_1^2 \dots M_u^2 &\leq  \prod_{s=1}^t \left[ \E \left| (\ID - \E_s) \frac{1}{G_{k_s k_s}^{(\T_s)}} \right|^{2 (t+\sum_s (\gamma_s + \wt{\gamma}_s))} \right]^{\frac{1}{t+\sum_s (\gamma_s + \wt{\gamma}_s)}} \\ &\hspace{3cm} \times \prod_{s=t+1}^{2q} \, \prod_{j=1}^{\gamma_s + \wt{\gamma}_s}  \left[ \E   \, |K_{i_{j,s} l_{j,s}}^{(\T_{j,s})}|^{2(t+\sum_s (\gamma_s + \wt{\gamma}_s))} \right]^{\frac{1}{t+\sum_s (\gamma_s + \wt{\gamma}_s)}} \\ &\leq (Cq)^{cq^2} \EE_{t + \sum_{s=t+1}^{2q} (\gamma_s + \wt{\gamma}_s)}\\ & \leq (Cq)^{cq^2} \EE_{2q + |L(\kk)|}\ \end{split} \end{equation}
where we used Lemma \ref{l:G12}, Eq. (\ref{boundZk}), the bound $t +  \sum_{s=t+1}^{2q} (\gamma_s+\tilde \gamma_s) \leq cq^2$ (which follows from $t, \tilde\gamma_s,\gamma_s\leq 2q$), and the estimate
 \[t+ \sum_{s=t+1}^{2q} (\gamma_s+\tilde \gamma_s) \geq 2q +\sum_{s=1}^u \ell(s) \geq 2q + |L(\kk)|,\]
 which follows from Remark (\ref{f:1}) in Section \ref{sssectapply}.

%  because for  $b(\sigma)\geq 1$ or $b(\rho)\geq 1$ one has $\gamma(A_{\sigma}) + \gamma(B_
%\sigma) \geq b(\sigma) + b(\rho) + 1 \geq \ell+1$ by fact (\ref{f:1}) in Section \ref{sssectapply},  and  %$\sum_{s=1}^{u}\ell(s)\geq |L(\kk)|$ (see also Prop. \ref{p:app}, here we used the fact that $\E|\La|%^{cq^2}\leq 1$ by Lem. \ref{l:no}).

\medskip

To bound the second factor on the r.h.s. of (\ref{eq:mMmM}), we observe that all resolvent entries in $\widetilde m_s$ and $\widetilde M_s$ are diagonal and that their  total number is at most $cq^2$. Applying Lemma \ref{boundG11} and Lemma \ref{l:Gkkinv} (using the assumption $q^2 \leq c (N\eta)^{1/4}$), we find
\begin{equation}\label{eq:tilde} \E \, \widetilde{m}^2_1 \dots \wt{m}^2_t \wt{M}^2_{t+1} \dots \wt{M}_{2q}^2  \leq (Cq)^{cq^2}\,. \end{equation}
Inserting in (\ref{eq:mMmM}), we obtain
\[ \Big| \E ((\ID-\E_{k_1}) ((\ID-\E_{k_1}) A^1_{\sigma_1})B^1_{\bsigma_1}) \cdots \overline{((\ID-\E_{k_{2q}}) ((\ID-\E_{k_{2q}}) A^{2q}_{\sigma_{2q}})B^{2q}_{\bsigma_{2q}})}\Big| \leq (Cq)^{cq^2} \EE^{1/2}_{2q+|L(\kk)|} \]
which concludes the proof of \eqref{claim1} in Case 1.

\medskip

{\bf Case 2}. For some $r\in\{1,...,2q\}$ there is at least one resolvent entry, either in $A^r_{\sigma_r}$ or in $B^r_{\bsigma_r}$, which is not maximally expanded. Among the other $2q-1$ indices in $\{ 1, \dots , 2q \} \backslash \{ r \}$ let us say that for $t$ of them (denoted by $s_1, \dots , s_t$), the
strings $\sigma_{s_j}$ and $\bsigma_{s_j}$ only contain zeros, while for the other $u = 2q-1-t$ labels, $s_{t+1}, \dots , s_{2q-1}$, either the string $\sigma_{s_j}$ or the string $\bsigma_{s_j}$ contain at least  a one. Without loss of generality, we can assume that $s_1 = 1, \dots , s_t = t$, $s_{t+1} = t+1, \dots , s_{2q-1} = 2q-1$ and $r = 2q$.

\medskip

For $s=1, \dots , t$, we use the estimate
\[ | (\ID - \E_s) ((\ID - \E_s) A_{\sigma_s}^s) B_{\bsigma_s}^s| \leq C m_s \widetilde{m}_s \]
established in (\ref{boundm}) with $m_s$ and $\widetilde{m}_s$ defined as in (\ref{m}). Since, for $s=1, \dots , t$, the strings $\sigma_s$ and $\bsigma_s$ only contain zeros, we have
\[ A_{\sigma_s}^s = \frac{1}{G_{ss}^{(\T_s)}}, \quad \text{ and } \quad B_{\bsigma_s}^s = G_{ss}^{(\T'_s)} \]
for appropriate sets $\T_s$ and $\T'_s$. This implies that (like in (\ref{eq:ms2}))
\begin{equation}\label{eq:mc2} m_s^2 \leq C \left| (\ID-\E_s) \frac{1}{G_{k_s k_s}^{(\T_s)}} \right|^2 + C \, \E_s \left| (\ID-\E_s) \frac{1}{G_{k_s k_s}^{(\T_s)}} \right|^2  \, .
\end{equation}
%and
%\[ \wt{m}_s^2 \leq C | G_{ss}^{(\T'_s)}|^2 + C \E_s | G_{ss}^{(\T'_s)}|^2 \]

\medskip

For $s=t+1, \dots , 2q-1$, either $A^s_{\sigma_s}$ or $B^s_{\bsigma_s}$ contains at least two off-diagonal resolvent entries. Let us assume, for example, that $A^s_{\sigma_s}$ contains two off-diagonal entries (it is easy to invert the roles of $A^s_{\sigma_s}$ and $B^s_{\bsigma_s}$, if it is $B_{\bsigma_s}^s$ which contains the off-diagonal entries). In this case, we write $A^s_{\sigma_s}=O_s P_s$, where the monomial $O_s$ is given by the product of $\gamma_s = \gamma (A^s_{\sigma_s}) \geq 2$ terms of the form $K_{i_j l_j}^{(\T_j)}$ while $P_s$ is a product of diagonal entries (appearing either in the numerator or in the denominator). We bound (for any $s=t+1, \dots , 2q-1$)
\[ \left|(\ID-\E_s) ((\ID-\E_s) A^s_{\sigma_s}) B^s_{\bsigma_s} \right| \leq  C L_s \widetilde L_s
\]
where we defined
\begin{align*}
 L_s &= |O_s| + (\E_s |O_s|^2)^\frac12 \\
 \widetilde L_s &= (|P_s|+(\E_s |P_s|^4)^\frac14) (|B^s_{\bsigma_s}|+(\E_s | B^s_{\bsigma_s}|^4)^\frac14) \,.
\end{align*}
We observe that
\begin{equation}\label{eq:bd-Ls}  L_s^2 \leq |O_s|^2 + \E_s |O_s|^2 = \prod_{j=1}^{\gamma_s} |K_{i_{j,s} l_{j,s}}^{(\T_{j,s})}|^2 + \prod_{j=1}^{\gamma_s} \left[ \E_s  |K_{i_{j,s} l_{j,s}}^{(\T_{j,s})}|^{2\gamma_s}\right]^{\frac{1}{\gamma_s}} \end{equation}
for appropriate indices $i_{j,s}, l_{j,s}$ and appropriate sets $\T_{j,s}$, and for all $s=t+1, \dots , 2q-1$.

\medskip

Finally, we consider $s=2q$, for which one of the resolvent entries, either in $A^{2q}_{\sigma_{2q}}$ or in $B^{2q}_{\bsigma_{2q}}$ is not maximally expanded. Let us assume that the entry which is not maximally expanded is in $A^{2q}_{\sigma_{2q}}$. In this case, we write
\[ A^{2q}_{\sigma_{2q}} = O P , \]
where $O$ is a monomial given by the product of $\gamma \equiv \od(A^{2q}_{\sigma_{2q}})$ terms of the form $K_{kl}^{(\T)}$ and $P$ contains only diagonal entries. Note that, since we assumed $A^{2q}_{\sigma_{2q}}$ not to be maximally expanded, we must have $2q+1\leq \od \leq 2q+2$.
We estimate
\[
\left|(\ID-\E_1) ((\ID-\E_1) OP) B^{2q}_{\bsigma_{2q}} \right| \leq  C L \widetilde L
\]
where we set
\begin{align*}
 L &= |O| + (\E_{2q} |O|^2)^\frac12 \\
 \widetilde L &= (|P|+(\E_{2q} |P|^4)^\frac14) (|B^1_{\rho_1}|+(\E_{2q} |B^1_{\rho_1}|^4)^\frac14) \,.
\end{align*}
Similarly to (\ref{eq:bd-Ls}), we find
\begin{equation}\label{eq:bd-L} L^2 \leq  \prod_{j=1}^{\gamma} |K_{i_j l_j}^{(\T_j)}|^2 + \prod_{j=1}^{\gamma} \left[ \E_s  |K_{i_j l_j}^{(\T_j)}|^{2\gamma_s}\right]^{\frac{1}{\gamma_s}}\,. \end{equation}

\medskip

By Cauchy-Schwarz, we can bound the l.h.s. of Eq. \eqref{claim1} by
\begin{equation}\begin{split}
\label{eq:CS-2}
\Big| \E ((\ID-\E_{k_1}) ((\ID-\E_{k_1}) &A^1_{\sigma_1})B^1_{\bsigma_1}) \cdots \overline{( (\ID-\E_{k_{2q}}) ((\ID-\E_{k_{2q}}) A^{2q}_{\sigma_{2q}})B^{2q}_{\bsigma_{2q}})}\Big| \\
\leq & C (\E \, m_1^2 \dots m_t^2 L_{t+1}^2 \dots L_{2q-1}^2 L^2 )^\frac12 \, (\E \, \widetilde m^2_1\dots \widetilde m^2_{t} \widetilde L_{t+1}^2 \dots \widetilde L_{2q-1}^2 \wt{L}^2)^\frac12  \,.
\end{split}\end{equation}
Proceeding as in (\ref{eq:EmMs}), and using the bounds (\ref{eq:mc2}), (\ref{eq:bd-Ls}), (\ref{eq:bd-L}), we find
\begin{equation}\begin{split}
\label{eq:mLL2}
\E \, m_1^2 &\dots m_t^2 L_{t+1}^2 \dots L_{2q-1}^2 L^2
\\
 \leq & \; (Cq)^{2q} \prod_{s=1}^t \left[ \E \, \left|(\ID-\E_{1}) \frac{1}{G^{(\T_s)}_{k_s k_s}} \right|^{2(t+\gamma+\sum_s \gamma_s))} \right]^{\frac{1}{2(t+\gamma+\sum_s \gamma_s))}} \\ &\hspace{1cm} \times
 \prod_{s=t+1}^{2q-1} \prod_{j=1}^{\gamma_s} \left[ \E \, |K^{(\T_{j,s})}_{i_{j,s},l_{j,s}}|^{2(t+\gamma+\sum_s \gamma_s)} \right]^\frac{1}{2(t+\gamma+\sum_s \gamma_s)}  \prod_{j=1}^\gamma \left[ \E \, |K^{(\T)}_{i_j l_j}|^{2(t+\gamma+\sum_s \gamma_s)} \right]^\frac{1}{2(t+\gamma+\sum_s \gamma_s)}
 \\ \leq & \; (Cq)^{cq^2} \EE_{t+\gamma+\sum_s \gamma_s} \\
\leq & (Cq)^{cq^2} \, \EE_{4q} \,.
 \end{split} \end{equation}
Here we used the fact that $t+\gamma+\sum_{s} \gamma_s \leq cq^2$ and the bound
\[ t+\gamma+\sum_s \gamma_s \geq t +2q+1+ u \geq 4q,\]
because $\gamma \geq 2q+1$ and $\gamma_s \geq 1$ for all $s=t+1, \dots , 2q-1$. On the other hand we have that, similarly to (\ref{eq:tilde}),
\[
\E \, \widetilde m^2_1\dots \widetilde m^2_{t} \, \widetilde L_{t+1}^2 \dots \widetilde L^2_{2q-1} \wt{L}^2  \leq (Cq)^{cq^2} \, .
\]
Inserting (\ref{eq:mLL2}) and the last inequality into (\ref{eq:CS-2}), we obtain
\[ \Big| \E ((\ID-\E_{k_1}) ((\ID-\E_{k_1}) A^1_{\sigma_1})B^1_{\bsigma_1}) \cdots \overline{( (\ID-\E_{k_{2q}}) ((\ID-\E_{k_{2q}}) A^{2q}_{\sigma_{2q}})B^{2q}_{\bsigma_{2q}})}\Big| \leq (Cq)^{c q^2} \EE_{4q}^{1/2}\]
Eq. \eqref{claim1} follows because $2q+|L(\kk)| \leq 4q$ (and because of (\ref{eq:EEpq})).

%%%%%%%%%%%
%SUBSECTION
%%%%%%%%%%%

\section{Improved optimal bound in the bulk}
\label{sec:bulk}

In this section, we prove Theorem \ref{t:bulk}, which gives a stronger bound on the fluctuations of the Stieltjes transform in the bulk of the spectrum. Compared with Theorem \ref{t:bulk}, here there is no upper bound on the size of the power $q$. Looking at the proof of Theorem \ref{t:main}, we observe that there are two results, in which the restriction $q \leq c (N\eta)^{1/8}$ played an important role, namely the proof of Lemma~\ref{l:boundG11}, containing a bound for $\E |G_{11}|^q$, and of Lemma \ref{lwidetildeG}, giving an estimate for $\E | (\E_1  G^{-1}_{11})^{-1} |^q$. In Lemma  \ref{l:bulk1} and Lemma \ref{l:bulk2} below, we show that, in the bulk of the spectrum, at distances of order one from the edges, these moments can be bounded for arbitrary $q \in \NA$. Using these estimates, Theorem \ref{t:bulk} can then be shown exactly as we proved Theorem \ref{t:main}.

\medskip

In order to bound the moments of $|G_{11}|$ and $|(\E_1  G^{-1}_{11})^{-1} |$ in the bulk of the spectrum, we will need, first of all, an upper bound on the density of states, as stated in the next proposition.
\begin{proposition}
\label{p:wwb}
Assume (\ref{e:gd}), fix $\tildeeta> 0$ and $0 < \kappa < 2$ and set $z = E+i\eta$. Then there exist constants $c,C,M,K_0 > 0$ such that
\[
%\begin{equation}
%\label{wwb2}
\P\left(\frac{\NN [E- \frac{\eta}{2}, E+ \frac{\eta}{2}]}{N\eta} \geq K \right)\leq C e^{-c\sqrt{KN\eta}}
%\end{equation}
\]for all $E \in (-2+\kappa, 2-\kappa)$, $N > M$, $K > K_0$, $\eta > 1/N$. Here $\NN [a;b]$ denotes the number of eigenvalues in the interval $[a;b]$.
\end{proposition}
%PROOF
\begin{proof}
We set $I =  [E- \frac{\eta}{2}, E+ \frac{\eta}{2}]$ and denote by $\NN_I$ the number of eigenvalues in the interval $I$. For $N\eta \geq (\log N)^4$, the claim follows from Theorem 4.6 in \cite{ESY3}. So, we can assume that $1 \leq N\eta \leq (\log N)^4$. We distinguish two cases.
\begin{itemize}
\item[$\bullet$] If $K \leq CN^\beta / (N\eta)$, for a $0< \beta < 1/16$, we use Theorem 5.1 of \cite{ESY3}, with $\nu = 1/4$, to estimate
\[
\P\left(\frac{\NN_I}{N\eta}\geq K\right)\leq \left(\frac{C}{K}\right)^{KN\eta/4} \leq e^{-c KN\eta}
\]
which holds true for any interval $I=[E-\eta/2,E+\eta/2]$ with $1\leq N\eta\leq C  N^\beta$ (and in particular holds for all $1\leq N\eta \leq (\log N)^4$).
\item[$\bullet$] If $K\geq C N^\beta / (N\eta)$ we use the bound in Theorem 4.6 of \cite{ESY3}. We consider a new interval $\tilde I = [E- \wt{\eta}/2, E+ \widetilde{\eta}/2]$ with $ (\log N)^4 \leq N \wt{\eta} \leq   2(\log N)^4$. Since obviously $\NN_{\wt{I}} \geq \NN_I$, we find
\[ \P \left( \frac{\NN_I}{N\eta} \geq K \right) = \P \left( \NN_I \geq K N \eta \right) \leq \P \left( \NN_{\wt{I}} \geq K N \eta \right) = \P \left( \frac{\NN_{\wt{I}}}{N\wt{\eta}} \geq  \frac{K\eta}{\wt{\eta}} \right) \leq  C e^{-c \sqrt{K N \eta}} \]
where we used Theorem 4.6 of \cite{ESY3}, and the fact that $K \eta/ \wt{\eta} \geq C N^\beta / (\log N)^4$ is large.
\end{itemize}
\end{proof}

\medskip

Secondly, we will need an upper bound on the size of the gap between eigenvalues; more precisely, we need to know that, if we consider an interval larger than $Kp/N$ the probability of finding fewer than $p$ eigenvalues tends to zero, as $K \to \infty$.
%%%%%%%%%%%%
%PROPOSITION
%%%%%%%%%%%%
\begin{proposition}
\label{p:gap}
Assume (\ref{e:gd}), fix $\tildeeta> 0$ and $0 < \kappa < 2$ and set $z = E+i\eta$. Then there exist constants $c,C,M,K_0 > 0$ such that
\[ %\begin{equation}\label{gapp}
  \P\left(\left|\left\{\al:\, |\la_\al - E| \leq \frac{Kp}{N}\right\}\right| \leq p \right) \leq C e^{-c K^{1/4}}
 %\end{equation}
 \] for any $E\in[-2+\kappa, 2-\kappa]$, $p\geq 1$, $N > M$, $K > K_0$.
\end{proposition}
%PROOF
\begin{proof}
 Set $\eta = p\, \sqrt{K}/N$ and $z=E+i\eta$. We consider the intervals
\[
\begin{split}
 I_0 &= [E-\sqrt K\eta,E+\sqrt K \eta ]\, \quad \text{and} \\
 I_j  &= [E-2^j\sqrt K\eta,E-2^{j-1}\sqrt K \eta  ] \cup [E+2^{j-1}\sqrt K\eta,E+2^{j}\sqrt K \eta] \qquad \text{for all $j \geq 1$.}
\end{split}
\]
We denote by $\NN_j$ the number of eigenvalues in $I_j$.  For a large constant $M$, we define the event
\[
 \Omega := \left\{\max_j \frac{\NN_{j}}{N|I_j|} \geq  M \right\}.
\]
{F}rom the upper bound in Proposition \ref{p:wwb}, we find
\begin{equation}\label{eq:omega}
 \P(\Omega) \leq   \sum_{j=0}^\infty C e^{-c\sqrt{M N |I_j|}} \leq \sum_{j=0}^\infty
 C e^{-c\sqrt{M K p \, 2^j}} \leq C e^{-c K^{1/2}}.
\end{equation}
On $\Omega^c$, and assuming that $\left|\left\{\al:\, |\la_\al - E| \leq Kp/N \right\}\right| \leq p$, we find
\begin{equation}\begin{split}\label{eq:mz}
 \Im m (E+i\eta) = & \frac{1}{N}\sum_{\al=1}^N \frac{\eta}{(\la_\al- E)^2+\eta^2} \\
= & \frac1{N} \sum_{j=0}^\infty \sum_{\al:\,\la_\al\in I_j} \frac{\eta}{(\la_\al-E)^2+\eta^2}\\
\leq & \frac{p}{N\eta} + \frac{C}{\sqrt K} \sum_{j=1}^\infty\frac{1}{2^j} \leq  \frac{C}{\sqrt K}.
\end{split}\end{equation}
On the other hand, since we are away from the edges, we have $c_0 = \text{Im } m_{sc} (z) > 0$. Hence, for $K$ large enough, (\ref{eq:mz}) implies that the fluctuations of $m(z)$ from the Stieltjes transform of the semicircle law are larger than, say, $c_0/2$. Using Theorem 3.1 of \cite{ESY3}, we
know that the probability for such large fluctuations is small. More precisely, we find
\[
\begin{split}
\P\left(\Omega^c \; \textrm{and} \; \left| \left\{ \al:\, |\la_\al - E| \leq  \frac{Kp}{N} \right\} \right| \leq p\right) \leq & \; \P\left( \Im m (E+i\eta) \leq \frac{C}{\sqrt K} \right) \\
\leq &  \; \P\left(\left| m (z) -m_{sc}(z)\right| \geq \frac{c_0}{2} \right)\\
\leq & \; Ce^{-c\sqrt{N\eta}}  \leq \; Ce^{-c K^{1/4}}\,.
\end{split}
\]
Together with (\ref{eq:omega}), this concludes the proof of the proposition.
\end{proof}

\medskip

We are now ready to prove upper bounds for arbitrary moments of diagonal resolvent entries. The next lemma gives an improvement, in the bulk, of Lemma~\ref{l:boundG11}.
%%%%%%
%LEMMA
%%%%%%
\begin{lemma}\label{l:bulk2}
Assume (\ref{e:gd}), fix $\tildeeta> 0$ and $0 < \kappa < 2$ and set $z = E+i\eta$. Then there exist constants $c,C,M> 0$ such that
\[
\E |G_{11}|^{q} \leq (Cq)^{cq}
\]
for all $E\in[-2+\kappa, 2-\kappa]$, $\eta > 1/N$, $N>M$, $q \in \NA$.
\end{lemma}
%PROOF
\begin{proof}
We fix $K\geq 1$. {F}rom (\ref{Gjj}), estimating $|1/w| \leq 1/ |\Im w|$, we find
\[
|G_{11}| \leq \left[\frac{\eta}{N} \sum_{\al=1}^{N-1} \frac{\xi_\al}{(\la_\al^{(1)}-E)^2+\eta^2}\right]^{-1}
\leq  \left[\frac{1}{NK^2\eta} \sum_{\al:\,|\la_\al^{(1)}-E| \leq K \eta} \xi_\al\right]^{-1}
\]
where $\xi_\alpha = N |\aa_1 \cdot \uu_\al^{(1)}|^2$ and $\la_\al^{(1)}, \uu_\al^{(1)}$ are the eigenvalues and the eigenvectors of the minor obtained by removing the first row and column from the original Wigner matrix. We find
\begin{align}\label{eq:bd0}
 \P\left(|G_{11}|\geq  K^3 \right)
\leq & \; \P\left(|\{\al:\, |\la_\al^{(1)}-E|\leq K\eta \}| \leq \sqrt{K}N\eta\right)  +  \P\left(\frac{1}{\sqrt{K}N \eta} \sum_{j = 1 }^{\sqrt{K}N \eta} \xi_{\al_j} \leq \frac{1}{K^{3/2}} \right).
\end{align}
{F}rom Proposition \ref{p:gap}, with $p=\sqrt K N\eta$, we obtain
\begin{equation}\label{eq:bd1}
\P\left(|\{\al:\, |\la^{(1)}_\al-E|\leq K\eta \}| \leq \sqrt{K}N\eta\right)  \leq C e^{-c K^{1/8}}
\end{equation}
for all $K > K_0$ large enough. On the other hand, the large deviation estimates in Proposition \ref{p:HW} imply (see for example Lemma 4.7 in \cite{ESY3}) that
\begin{equation}\label{eq:bd2}
\P\left(\frac{1}{\sqrt{K}N \eta} \sum_{j = 1 }^{\sqrt{K}N \eta} \xi_{\al_j} \leq \frac{1}{K^{3/2}} \right) \leq  e^{-c \sqrt{\sqrt{K}N\eta}} \leq e^{-c K^{1/4}}.
\end{equation}
Inserting (\ref{eq:bd2}) and \eqref{eq:bd1} into (\ref{eq:bd0}), we obtain that there exists a constant $t_0$ such that
\[
 \P\left(|G_{11}| \geq t \right) \leq C e^{-c t^{1/24}}
\]
for all $t \geq t_0$. This implies that
\[
\E|G_{11}|^q =   \int_0^\infty \P (|G_{11}|\geq t^{1/q}) dt  \leq  t_0 +C    \int_{t_0}^\infty e^{-ct^{1/(24q)}} dt  \leq (Cq)^{cq}.
\]
\end{proof}

Finally, we establish also bounds for arbitrary moments of $|(\E_1 G_{11}^{-1})^{-1}|$ which improve, in the bulk of the spectrum, the results of Lemma \ref{lwidetildeG}.
%%%%%%
%LEMMA
%%%%%%
\begin{lemma}\label{l:bulk1}
Assume (\ref{e:gd}), fix $\tildeeta> 0$ and $0 < \kappa < 2$ and set $z = E+i\eta$. Then  there exist constants $c,C,M> 0$ such that
\begin{equation}\label{try}
\E \left|\frac{1}{\E_1 \frac{1}{G_{11}}}\right|^{q} \leq (Cq)^{cq}
\end{equation}
for all $E\in[-2+\kappa, 2-\kappa]$, $\eta > 1/N$, $N > M$, $q\in \NA$.
\end{lemma}
%PROOF
\begin{proof}
Let $K\geq 1$. Using (\ref{Gjj}) and estimating $|1/w| \leq 1/|\Im  w|$, we find
 \[
\left |\frac{1}{\E_1 \frac{1}{G_{11}}}\right| \leq  \left[\frac{\eta}{N} \sum_{\al=1}^{N-1} \frac{1}{(\la_\al^{(1)}-E)^2+\eta^2}\right]^{-1} \leq  \left[\frac{|\{\al:\, |\la_\al^{(1)}-E|\leq K\eta \}| }{NK^2\eta} \right]^{-1}
 \]
 where $\la_\al^{(1)}$ are the eigenvalues of the minor obtained by removing the first row and column from the original Wigner matrix. Using Proposition \ref{p:gap} with $p= N\eta$, we conclude that, for all $K > K_0$ large enough,
 \[
  \P\left(\left |\frac{1}{\E_1 \frac{1}{G_{11}}}\right|\geq  K^2 \right)
 \leq \P\left(|\{\al:\, |\la_\al^{(1)}-E|\leq K\eta \}| \leq N\eta\right)
 \leq C e^{-c K^{\frac14}}\,,
 \]
which clearly implies the claim (\ref{try}).
\end{proof}

\section{Convergence of the eigenvalue counting function}
\label{sec:counting}

In this section we apply Theorem \ref{t:main} to prove the convergence of the counting function of the eigenvalues, as stated in Theorem \ref{t:counting}. We adapt here the approach developed in  \cite{EYY, EYY0,ekyy-rxv12}.

\begin{proof}[Proof of Theorem \ref{t:counting}]
We are going to show that there exist constants $C,c > 0$ such that
\begin{equation}\label{eq:Ebul}
\E \, \left| (n(E_2) - n(E_1)) - (n_{sc} (E_2) - n_{sc} (E_1)) \right|^q \leq (Cq)^{cq^2} \frac{(\log N)^q}{N^q}
\end{equation}
for all $-2 \leq E_1 < E_2 \leq 2$, $N > 1$ and for all $q \in \NA$. This gives
\begin{equation}\label{eq:Pbul} \P \left( \left|  (n(E_2) - n(E_1)) - (n_{sc} (E_2) - n_{sc} (E_1)) \right| \geq \frac{K \log N}{N} \right) \leq \frac{(Cq)^{cq^2}}{K^q} \end{equation}
for all $-2 \leq E_1 < E_2 \leq 2$, $K > 0$ and $q \in \NA$. In particular,
\[ \P \left( \left|  (n(2) - n(-2)) - 1 \right| \geq \frac{K \log N}{N} \right) \leq \frac{(Cq)^{cq^2}}{K^q} \]
for all $K > 0$ and $q \in \NA$. This implies that, with high probability, there are at most $K \log N$ eigenvalues outside the interval $[-2;2]$; this observation allows us to conclude the proof of Theorem~\ref{t:counting}. Indeed, for $E \leq 2$ we have $n_{sc} (E) = 0$ and therefore
\begin{equation}\label{eq:Ped1}
\begin{split} \P \left( \left| n (E) - n_{sc} (E) \right| \geq \frac{K \log N}{N} \right) &= \P \left( n (E) \geq \frac{K \log N}{N} \right) \\ &\leq \P \left( |n(2) - n (-2) - 1| \geq \frac{K \log N}{N} \right) \leq  \frac{(Cq)^{cq^2}}{K^q} \end{split} \end{equation}
because, for $E \leq -2$, $n(E) + n(2) - n(-2) \leq 1$. For $E \geq 2$, we have $n_{sc} (E) = 1$ and hence
\[ \begin{split} \P \left( \left| n (E) - n_{sc} (E) \right| \geq \frac{K \log N}{N} \right) &= \P \left( 1- n (E) \geq \frac{K \log N}{N} \right) \\ &\leq \P \left( |n(2) - n (-2) - 1| \geq \frac{K \log N}{N} \right) \leq  \frac{(Cq)^{cq^2}}{K^q} \end{split} \]
because, for $E \geq 2$, $1- n(E) + n(2) - n(-2) \leq 1$. Finally, for $-2 <  E < 2$, we have $|n(E)-n_{sc} (E)| \leq |n(E) - n (-2) -n_{sc} (E)| + n(-2)$. Since $n_{sc} (-2) = 0$, we find
\[ \begin{split}  \P  \left( \left| n (E) - n_{sc} (E) \right| \geq \frac{K \log N}{N} \right) \leq \; & \P \left( \left| (n (E) - n(-2)) - (n_{sc} (E) - n_{sc} (-2))\right| \geq  \frac{K \log N}{2N} \right) \\ &+ \P \left( n(-2) \geq \frac{K \log N}{2N} \right) \\ \leq  \; &\frac{(Cq)^{cq^2}}{K^q} \end{split} \]
from (\ref{eq:Pbul}) and (\ref{eq:Ped1}).

\medskip

To conclude the proof of Theorem \ref{t:counting}, we show (\ref{eq:Ebul}), following \cite{ekyy-rxv12}. We define the empirical eigenvalue distribution
\[
 \rho(\la) = \frac{1}{N} \sum_{\al=1}^N \delta(\la-\la_\al)\,
\]
and we denote by $\II_{[E_1,E_2]}$ the characteristic function of the interval $[E_1,E_2]$. With this notation we have
\[%\begin{equation}\label{lee1}
(n(E_2) - n(E_1)) - ( n_{sc}(E_2) - n_{sc}(E_1)) = \int_{\RE} \II_{[E_1,E_2]}(\la) \, (\rho(\la) -\rho_{sc} (\la)) \,  d\la\,.
%\end{equation}
\]Next, we approximate $\II_{[E_1,E_2]}$ with a smooth function $f$. For a given $q \in \NA$, we choose $M = C q^8$ for a sufficiently large constant $C > 0$ and we set $\tilde \eta =M/N$ (we choose the constant $C > 0$ large enough and consider $N \geq q^8$ large enough, so that we can apply Theorem \ref{t:main} to bound $\Lambda (E+i\tildeeta)$; if $N < q^8$, (\ref{eq:Ebul}) is trivial since, by definition, $0 \leq n (E) \leq 1$ for all $E \in \RE$). We choose $f\in C_0^\infty(\RE)$  such that $\supp f \subset (E_1, E_2)$ and $f=1$ in $[E_1+\tilde \eta, E_2-\tilde\eta ]$. Notice that, since we are considering $-2 \leq E_1 < E_2 \leq 2$, we can apply part i) of Theorem \ref{t:main} to bound $|\Lambda (E+i\tildeeta)|$ for all $E \in \supp f$. We also assume that $|f'|\leq C/\tilde\eta$ and $|f''| \leq C/ \tilde\eta^2$. {F}rom the boundedness of $\rho_{sc}$, we find
\[
\left| \int_{\RE} (\II_{[E_1,E_2]}(\la) - f(\la)) \rho_{sc}(\la) d\la \right|^q \leq C^q\tilde \eta^q = \frac{C^q M^q}{N^q}  \,.
\]
Moreover, recalling the  bound
\[
| n(E+\eta) - n(E-\eta)| \leq C\eta\Im m(E+i\eta) \leq C\eta \, (1+|\Im \La(E+i\eta)|)
\]
we find \[
\left| \int_{\RE} (\II_{[E_1,E_2]}(\la) - f(\la)) \rho(\la) d\la \right| \leq C\sum_{j=1,2} \left|n(E_j+\tilde\eta) - n(E_j-\tilde \eta)\right| \leq C\tilde\eta(1+|\Im \La(E+i\tilde\eta)|)\,
\]
where, as usual, $\La (z) = m (z) -m_{sc} (z)$. Taking the $q$-th moment, we obtain
\[
\E \left| \int_{\RE} (\II_{[E_1,E_2]}(\la) - f(\la)) \rho(\la) d\la \right|^q  \leq C^q \tilde\eta^q (1+\E |\Im \La(E+i\tilde\eta)|^q) \leq \frac{C^q M^q}{N^q}
\]
where we applied the bound $\E |\Im \La(E+i\tilde\eta)|^q \leq 1$ from Lemma \ref{l:no} (using that $q \leq c (N\tilde \eta)^{1/8}$). We conclude that
\begin{equation}\label{7.13}\begin{aligned}
\E|(n(E_2) - n(E_1))-(n_{sc}(E_2)-n_{sc}(E_1))|^q
 \leq &  \frac{C^q}{N^q}+ C^q  \E\left|\int_\RE f(\la) (\rho(\la) - \rho_{sc}(\la)) d\la\right|^q\,.
\end{aligned}\end{equation}

\medskip

Next, we use Helffer-Sj\"ostrand functional calculus. We introduce  the notation $I=[E_1,E_2]$  and   $\tilde I = [E_1,E_1+\tilde \eta]\cup [ E_2-\tilde\eta ,E_2]$. Notice that $|\tilde I| = 2\tilde \eta $. For a $\chi \in C^\infty_0 (\RE)$ with $\chi (y) = 1$ on $[-1,1]$ and $\chi (y) = 0$ on $[-2,2]^c$, we can define the almost analytic extension of $f$ on $\CO$ given by
\[ \wt{f} (x+iy) = (f(x) +iyf'(x))\chi(y)\,.\]
Then, we have, see Eq. (B.12) in \cite{ersy-ejp10},
\[%\begin{equation}\label{eq:HS}
 f(\la) =\frac{1}{2\pi} \int_{\RE^2}
\frac{\partial_{\bar z} \wt{f}(x+iy)}{\la-x-iy}dx \,dy = \frac{1}{2\pi} \int_{\RE^2}
\frac{iyf''(x)\chi(y)+i(f(x)+iyf'(x))\chi'(y)}{\la-x-iy}dx \,dy \,. \]
% \end{equation}
Since $f$ is real, we find
\[
\int_\RE  f(\la) (\rho(\la) - \rho_{sc}(\la)) d\la =
\frac{1}{2\pi} \Re  \int_{\RE^2}
[iyf''(x)\chi(y)+i(f(x)+iyf'(x))\chi'(y)] \La(x+iy) dx \,dy  \,.
\]
Taking the $q$-th moment, and recalling that $f$ and $\chi$ are compactly supported,  we find, similarly to Eq. (7.33) in \cite{ekyy-rxv12},
\begin{equation} \label{2+2}
\begin{aligned}
 \E \Big|\int_\RE  &f(\la) (\rho(\la) - \rho_{sc}(\la)) d\la\Big|^q
  \\ \leq &  C^q \, \E  \left\{ \left|\int_Idx \int_1^2 dy f(x) \chi'(y) \La(x+iy)   \right|^q  + \left| \int_{\tilde I} dx \int_1^2 dy f'(x) y \chi'(y) \La (x+iy) \right|^q \right.  \\
 & \left. + \left| \int_{\tilde I} dx \int_0^{\tilde \eta}  dy f''(x) y \chi(y) \Im \La(x+iy)  \right|^q
 + \left| \int_{\tilde I} dx \int_{\tilde \eta}^2 dy f''(x) y \chi(y) \Im \La(x+iy) \right|^q \right\} \, .
\end{aligned}
\end{equation}

%To bound the r.h.s. in \eqref{2+2} we shall use the inequality $\left|\int_\Omega F({\bf x})d{\bf x}
%\right|^q \leq |\Omega|^{q-1}\int_\Omega\left| F({\bf x})\right|^q d{\bf x}$, where $|\Omega|$ is the %volume of the integration region $\Omega$.

\medskip

Recalling that $|f|\leq 1$ and $\chi'\leq C$, the first term on the right hand side of \eqref{2+2} is bounded by
\[
 \E \left|\int_Idx \int_1^2 dy f(x) \chi'(y) \La(x+iy)   \right|^q \leq
C^q   \int_Idx \int_1^2 dy  \, \E \, \left|\La(x+iy)\right|^q\,.
\]
{F}rom Theorem \ref{t:main}, part i), we find
\begin{equation}\label{eq:2+2-1}
 \E \left|\int_Idx \int_1^2 dy f(x) \chi'(y) \La(x+iy)   \right|^q \leq
\frac{(Cq)^{cq^2}}{N^q} \, .
\end{equation}

\medskip

The second term on the r.h.s. of Eq. \eqref{2+2} can be bounded similarly. In this case $|f'| \leq C/\tilde \eta$ is large but it is compensated by the volume factor $|\tilde I| \leq 2 \tilde \eta$. We find, again by part i) of Theorem \ref{t:main}, that
\begin{equation}\label{eq:2+2-2}
 \E\left| \int_{\tilde I} dx \int_1^2 dy f'(x) y \chi'(y) \La (x+iy) \right|^q \leq
C^q \tilde \eta^{-1}   \int_{\tilde I} dx \int_1^2 dy \, \E\left| \La (x+iy) \right|^q \leq \frac{(Cq)^{cq^2}}{N^q}\, .
\end{equation}

\medskip

The third term on the r.h.s. of Eq. \eqref{2+2} can be estimated by
\begin{equation}\label{might}
\E\left| \int_{\tilde I} dx \int_0^{\tilde \eta}  dy f''(x) \, y\,  \chi(y) \Im \La(x+iy)  \right|^q \leq
\frac{C^q}{\tilde \eta^2} \int_{\tilde I} dx \int_0^{\tilde \eta}  dy  \, y^q\, \E\left| \Im \La(x+iy)  \right|^q
\end{equation}
because of the assumption $|f''|\leq C/\tilde \eta^2$. Here we cannot apply directly Theorem \ref{t:main}, because we are arbitrarily close to the real axis. We follow instead an argument of Theorem 2.2 in \cite{EYY}. For any $x\in\RE$, we note that the functions $y\to y\Im m(x+iy)$ and $y\to y\Im m_{sc}(x+iy)$ are monotone increasing on $\RE_+$ (they are Stieltjes transforms of positive measures). This implies that
\[
 y |\Im \La(x+iy)| \leq \tilde \eta(\Im m(x+i\tilde \eta) + \Im m_{sc}(x+i\tilde \eta)) \leq \tilde \eta \,
 (2+ |\Im \La (x+i \tilde\eta )|) \,.
\]
{F}rom \eqref{might}, we get therefore
\begin{equation}\label{eq:2+2-3}
\begin{split}
\E \Big| \int_{\tilde I} dx \int_0^{\tilde \eta}  dy f''(x) \, &y\,  \chi(y) \Im \La(x+iy)  \Big|^q \\ \leq \; &\frac{C^q}{\tilde \eta^2} \int_{\tilde I} dx \int_0^{\tilde \eta}  dy  \, \tilde \eta^q \, (1+\E|\Im\La(x+i\tilde \eta)|^q) \leq
C^q \tilde \eta^q \leq \frac{C^q M^q}{N^q}
\end{split}
\end{equation}
because, by Lemma \ref{l:no}, $\E|\Im\La(x+i\tilde \eta)|^q \leq 1$ (since $\tilde \eta$ was chosen so that $q \leq c (N \tilde \eta)^{1/8}$).

\medskip

To estimate the fourth term on the r.h.s. of Eq. \eqref{2+2}, we integrate by parts, first in $x$ and then in $y$. We obtain
\begin{equation}\label{tell}
\begin{split}
\E \, \Big| \int_{\tilde I} dx \int_{\tilde \eta}^2 dy &f''(x) y \chi(y) \Im \La(x+iy) \Big|^q  \\ \leq \; &
 \E \, \left| \int_{\tilde I} dx \int_{\tilde \eta}^2 dy f''(x) y \chi(y) \La(x+iy) \right|^q \\   \leq \; &
 C^q \, \E \, \left\{ \left| \int_{\tilde I} dx  f'(x) \tilde \eta   \La(x+i\tilde \eta) \right|^q + \left| \int_{\tilde I} dx \int_{\tilde \eta}^2 dy f'(x) \chi(y) \La(x+iy) \right|^q  \right. \\
& \left. \hspace{1cm} + \left| \int_{\tilde I} dx \int_{1}^2 dy f'(x) y \chi'(y) \La(x+iy) \right|^q   \right\}
\end{split}
\end{equation}
where we used the observation that $\chi (\tilde \eta) =1$ and $\chi'(x)=0$ for $x \not \in [1;2]$ in the first and, respectively, in the third term. By part i) of Theorem \ref{t:main} (and using that $|f' (x)| \leq C / \tilde \eta$) the first term on the r.h.s. of (\ref{tell}) can be bounded by
\begin{equation}\label{song1}
 \E\left| \int_{\tilde I} dx  f'(x) \tilde \eta   \La(x+i\tilde \eta) \right|^q \leq
\frac{C^q}{\tilde \eta}  \int_{\tilde I} dx \,  \tilde \eta^q \, \E\left|  \La(x+i\tilde \eta) \right|^q \leq
\frac{(Cq)^{cq^2}}{N^q} \, .
\end{equation}
The third term on the r.h.s. of (\ref{tell}), on the other hand,
can be controlled by
\begin{equation}\label{song2}
\E\left| \int_{\tilde I} dx \int_{1}^2 dy f'(x) y \chi'(y) \La(x+iy) \right|^q   \leq
\frac{C^q}{\tilde \eta} \int_{\tilde I} dx \int_{1}^2 dy  \, \E\left| \La(x+iy) \right|^q \leq \frac{(Cq)^{cq^2}}{N^q}\, .
\end{equation}

As for the second term on the r.h.s. of Eq. \eqref{tell}, we find  that
 \[
\begin{aligned}
\E \, \Big| \int_{\tilde I} dx \int_{\tilde \eta}^2 dy &f'(x) \chi(y) \La(x+iy) \Big|^q \\
\leq  \; & \int_{\tilde I} dx_1 \cdots\int_{\tilde I} dx_q \int_{\tilde \eta}^2 dy_1\cdots \int_{\tilde \eta}^2 dy_q  \, \prod_{j=1}^q |f'(x_j)| |\chi(y_j)| \,  \E  \prod_{j=1}^q |\La(x_j+iy_j)| \\
\leq \; & \prod_{j=1}^q \int_{\tilde I} dx_j  \int_{\tilde \eta}^2 dy_j \,  |f'(x_j)| |\chi(y_j)| \,  \left(\E |\La(x_j+iy_j)|^q \right)^{1/q}
\end{aligned}
 \]
applying H\"older's inequality to the expectation. With the usual bounds on $|f'|$ and $|\chi|$, we obtain
 \begin{equation}
\begin{aligned}\label{song3}
 \E \, \Big| \int_{\tilde I} dx \int_{\tilde \eta}^2 dy &f'(x) \chi(y) \La(x+iy) \Big|^q \\
\leq &  \frac{(Cq)^{cq^2}}{\tilde \eta^q} \left[ \int_{\tilde I} dx \int_{\tilde \eta}^2 dy \, \frac{1}{Ny} \right]^q
\leq  (Cq)^{c q^2} \frac{(\log \tilde \eta)^q}{N^q} \leq (Cq)^{cq^2}  \frac{(\log N)^q}{ N^q} \,.
\end{aligned}
 \end{equation}
 Inserting \eqref{song1}, \eqref{song2} and \eqref{song3} in \eqref{tell} we find
 \[
\begin{aligned}
  \E\left| \int_{\tilde I} dx \int_{\tilde \eta}^2 dy f''(x) y \chi(y) \Im \La(x+iy) \right|^q  \leq
(Cq)^{cq^2} \frac{(\log N)^q}{ N^q}\,.
\end{aligned}
\]
Inserting last equation, together with (\ref{eq:2+2-1}), (\ref{eq:2+2-2}) and (\ref{eq:2+2-3}), into (\ref{2+2}), we get
\[  \E \Big|\int_\RE  f(x) (\rho(x) - \rho_{sc}(x)) dx\Big|^q \leq
(Cq)^{cq^2} \frac{(\log N)^q}{ N^q}\,. \]
Combined with (\ref{7.13}), this implies (\ref{eq:Ebul}).
\end{proof}

\section{Rigidity of the semicircle law}
\label{sec:rigidity}

The goal of this section is to prove Theorem \ref{t:rigidity}, giving precise bounds on the fluctuation of eigenvalues of a Wigner matrix with respect to the locations predicted by the semicircle law. We will need, here, bounds on the location of the largest and smallest eigenvalues. As a first rough bound, we will use Lemma 7.2 in \cite{EYY1}, which states that, for any $x > 3$ and for an appropriate $\varepsilon > 0$,
\begin{equation}\label{eq:lm72} \E \, n (-x) \leq e^{-N^\varepsilon \log \, x}, \quad \text{and } \quad \E \,  n (x) \geq 1-e^{-N^\varepsilon \log \, x} \, .  \end{equation}

\medskip

To obtain better estimates for the extremal eigenvalues, we make use of the following bounds for the Stieltjes transform of the semicircle law. For $E \in \RE$, we set $\kappa = ||E|-2|$. For any fixed $E_0 >2$ and $\eta_0>0$, there exist constants $c,C > 0$ such that (see Eq.~(4.3) in \cite{ekyy-rxv12})
\[%\begin{equation}\label{msc1}
c\sqrt{\kappa+\eta}\leq |1-m_{sc}^2(E+i\eta)| \leq C\sqrt{\kappa +\eta} \qquad \forall \, |E|\leq E_0\,,\; 0< \eta \leq \eta_0
%\end{equation}
\]and
\begin{equation}\label{msc1-2}
c\frac{\eta}{\sqrt{\kappa +\eta}}\leq  \Im m_{sc}(E+i\eta) \leq C \frac{\eta}{\sqrt{\kappa +\eta}} \qquad \forall\, 2\leq |E|\leq E_0 \,,\; 0<\eta\leq \eta_0 \, .
\end{equation}

\medskip

{F}rom Theorem \ref{t:main} part ii), we have
\[ \E|\text{Im } \La|^{2q} \leq \frac{(Cq)^{cq^2}}{(N\eta)^{2q}},\]
for all $q \leq c (N\eta)^{1/8}$. Inserting this estimate and the bound (\ref{msc1-2}) into the r.h.s. of (\ref{e:R}), we conclude that, for $2\leq |E|\leq E_0$ and $\eta \leq \eta_0$ with $N\eta \geq M$ large enough, we have
\[
\begin{aligned}
\E|R|^q \leq & (Cq)^{cq^2}\left(  \max \left\{\frac{1}{(N\eta)^{4q}},\frac{(\Im m_{sc})^{2q} +\E|\Im\La|^{2q}}{(N\eta)^{2q}}\right\}   + \frac{1}{N^{2q}}\right)^{\frac12} \\
\leq &  (Cq)^{cq^2} \left( \frac{1}{(N\eta)^{4q}} + \frac{1}{(N\sqrt{\kappa+\eta})^{2q}}  + \frac{1}{N^{2q}}\right)^{\frac12} \\
\leq &  (Cq)^{cq^2} \left( \frac{1}{(N\eta)^{2q}} + \frac{1}{(N\sqrt{\kappa+\eta})^{q}}\right)\,
\end{aligned}
\]
for all $q \leq c (N\eta)^{1/8}$. Plugging back into (\ref{e:Lambda1}), we obtain a stronger bound for the imaginary part of $\Lambda$, given by
\begin{equation}\label{stronger}
\E|\Im\La|^q \leq  (Cq)^{cq^2} \, \left( \frac{1}{(N(\kappa+\eta))^{q}} + \frac{1}{((N\eta)^{2}\sqrt{\kappa+\eta})^{q}} \right)\,,
\end{equation}
and valid for $2\leq |E|\leq E_0$, $\eta\leq \eta_0$ with $N\eta \geq M$ large enough, $1\leq q \leq c(N\eta)^{1/8}$ (a similar bound is given in Eq. (2.20) in \cite{ekyy-rxv12}). %Moreover we also note that, due to the second term at the r.h.s., the bound \eqref{stronger} is %stronger than $\E|\La|^q \leq \frac{C_q^q}{(N\eta)^q} $ only for $\eta \sqrt{\kappa+\eta}\geq
%\frac{1}{N}$.

\medskip

The next lemma (which is similar to Theorem 7.3 in \cite{ekyy-rxv12}) estimates the probability of having eigenvalues outside the interval $[-2,2]$, making use of the bound (\ref{stronger}).
%%%%%%
%LEMMA
%%%%%%
\begin{lemma} \label{l:extremalevs}
Assume (\ref{e:gd}). Then  there exist $C,c, \varepsilon >0$ such that
\begin{equation}\label{extremalevs}
\P\left(\max_{\al}|\la_\al|\geq 2 +  \frac{K}{N^{2/3}} \right)\leq \frac{(Cq)^{cq^2}}{K^q}
\end{equation}
for all $K > 0$ and $q \in \NA$ with $q \leq N^{\varepsilon}$.
\end{lemma}
%PROOF
\begin{proof}
We prove that
\[%\begin{equation}\label{eq:cl-la1}
\P\left(\la_1\leq - 2 -  \frac{K}{N^{2/3}} \right)\leq  \frac{(Cq)^{cq^2}}{K^q} \, ,
%\end{equation}
\] where $\la_1$ is the smallest eigenvalue of $H$. The proof of the analogous bound for the largest eigenvalue $\la_N$ is similar and we omit it.

\medskip

We can assume that $K$ is large enough, since otherwise (\ref{extremalevs}) is trivial. We can also assume that $K \leq N$, since otherwise we can apply (\ref{eq:lm72}) to estimate
\[
%\begin{equation}\label{eq:Ko} 
\P \left(\lambda_1 \leq -2- \frac{K}{N^{{2/3}}}\right) \leq  \P \left(n\left(-\frac{K}{N^{{2/3}}}\right)\geq  \frac{1}{N}\right) \leq N \,\E \, n\left(-\frac{K}{N^{{2/3}}}\right) \leq e^{- N^\ve \log K } \leq
\frac{1}{K^q} 
%\end{equation}
\]
for all $q\leq N^\ve$.

\medskip

We fix $E_0 > 3$. {F}rom (\ref{eq:lm72}) we find that, for any small $\ve<0$,   $K\leq N$, $q\leq N^\ve$
\begin{equation}\label{eq:minE0} \P (\la_1 \leq -E_0) \leq \P \left(n(-E_0)\geq  \frac{1}{N}\right) \leq N \E n(-E_0)\leq e^{-N^{2\ve} \log E_0} \leq  e^{-N^{\ve} \log K} \leq \frac{1}{K^q}\,.\end{equation}

\medskip

We still have to bound the probability that $-E_0 < \la_1 \leq -2-K/N^{2/3}$. We define
\begin{equation}\label{eq:etajbds}
 \kappa_j = \frac{(K+j)}{N^{2/3}}\;;\qquad \eta_j=\frac{(K+j)^\frac18}{N^{2/3}}\,.
\end{equation}
 Moreover we define the intervals $I_j = \left[-2-\kappa_{j+1}, -2-\kappa_j\right]$ for $j=0, \dots , j_{max},$ where $j_{max}$ is the smallest integer  with $2+\kappa_{j+1}\geq E_0$ (clearly $j_{max}\leq E_0 N^{2/3}$). We denote by $x_j = -2 - \kappa_j$ the endpoints of the intervals $I_j$.

\medskip

We observe that
\begin{equation}\label{8.11}
 \P\left(-E_0\leq \la_1\leq - 2 -  \frac{K}{N^\frac{2}{3}} \right) \leq \sum_{j=0}^{j_{max}} \P \left(\la_1 \in I_j\right)\,.
\end{equation}
We note that $|I_j|=N^{-2/3}$, so that  if $\la_1\in I_j$ one has $|\la_1 - x_j|\leq |I_j| \leq \eta_j$.  Hence, setting $z_j=x_j+i\eta_j$, the event $\lambda_1 \in I_j$ implies that
\[%\begin{equation}\label{times0}
\Im m(z_j) = \frac1N \sum_\al \frac{\eta}{(\la_\al - x_j)^2 + \eta_j^2} \geq \frac{1}{2N\eta_j}\,.
%\end{equation}
\]
On the other hand, since  $2\leq|x_j|\leq 2 E_0$, the bound \eqref{msc1-2} gives
\[%\begin{equation}\label{knives1}
\Im m_{sc}(z_j) \leq \frac{C\eta_j}{\sqrt \kappa_j}\,.
%\end{equation}
\]
Since, by (\ref{eq:etajbds}),
\[ \frac{C\eta_j}{\sqrt{\kappa_j}} \leq \frac{1}{4N \eta_j} \]
for $K > K_0$ large enough (depending only on the value of $C$), we conclude that, if $\lambda_1 \in I_j$,
\begin{equation}\label{eq:Im-dif} \Im m (z_j) - \Im m_{sc} (z_j) \geq \frac{1}{2N \eta_j} - \frac{C\eta_j}{\sqrt \kappa_j} \geq \frac{1}{4N\eta_j}\,. \end{equation}
Next we observe that, again by the definition of $\kappa_j$ and $\eta_j$, and for $K>K_0$ large enough,
\[ \frac{1}{N\eta_j} \geq \frac{C(K+j)^\frac12}{N \kappa_j} \geq \frac{ C(K+j)^\frac12 }{N (\kappa_j + \eta_j)} \]
and
\[ \frac{1}{N\eta_j} \geq  \frac{C(K+j)^\frac12}{(N\eta_j)^2 \sqrt{\kappa_j}}  \geq \frac{C(K+j)^\frac12}{(N\eta_j)^2 \sqrt{\kappa_j + \eta_j}} \, . \]
Hence, (\ref{eq:Im-dif}) implies that
\[  \Im m (z_j) - \Im m_{sc} (z_j) \geq C (K+j)^\frac12 \left( \frac{1}{N (\kappa_j + \eta_j)} + \frac{1}{(N\eta_j)^2 \sqrt{\kappa_j + \eta_j}} \right)\,. \]
With Eq. \eqref{8.11} and by the bound \eqref{stronger} we conclude that
\[
\begin{aligned}
\P\left(-E_0\leq \la_1\leq - 2 -  \frac{K}{N^\frac{2}{3}} \right) \leq & \sum_{j=0}^{j_{max}} \P \left(\la_1 \in I_j\right)  \\
\leq &
 \sum_{j=0}^{j_{max}}  \P \left(| \Im\Lambda (z_j)| \geq C (K+j)^\frac12 \left( \frac{1}{N (\kappa_j + \eta_j)} + \frac{1}{(N\eta_j)^2 \sqrt{\kappa_j + \eta_j}} \right) \right) \\
 \leq & \sum_{j=0}^{j_{max}}  \frac{(Cq)^{cq^2}}{(K+j)^{q/2}} \leq \frac{(Cq)^{cq^2}}{K^{q/2-2}}\sum_{j=0}^{j_{max}} \frac1{j^2} \leq \frac{(Cq)^{cq^2}}{K^{q/2-2}}\,.
\end{aligned} \]
Changing $q/2-2 \to q$, this concludes, together with \eqref{eq:minE0}, the proof of the lemma.
\end{proof}

We are now ready to prove Theorem \ref{t:rigidity}; we proceed here similarly as in the proof of Theorem~7.6 in \cite{ekyy-rxv12}. %PROOF
\begin{proof}[Proof of Theorem \ref{t:rigidity}]
We prove the theorem only for $\al\leq N/2$, i.e.,  $\hat \al = \al$, the case $\al>N/2$ is similar.
We will make use of the inequalities
\begin{equation}\label{going}
c(2+x)^{\frac32} \leq n_{sc}(x) \leq C(2+x)^{\frac32}, \quad \text{ and } \quad
c n_{sc}(x)^{\frac13} \leq \rho_{sc}(x) \leq C n_{sc}(x)^{\frac13}
\end{equation}
which hold true for every $x \in [-2,1]$ (see Eq. (7.31) in \cite{ekyy-rxv12}).
In particular, the second inequality implies that  \begin{equation}\label{eq:rho13} c (\al/N)^{\frac13}\leq\rho_{sc}(\ga_\al)\leq C(\al/N)^{\frac13}, \end{equation} for any $\alpha \leq N/2$.

\medskip

We have
\[ \begin{split} \P \left(|\lambda_\alpha - \gamma_\alpha| \geq \frac{K \log N}{N} \left( \frac{N}{\alpha} \right)^{1/3} \right) \leq \; &\P \left( |\lambda_\alpha- \gamma_\alpha| \geq  \frac{K \log N}{N} \left( \frac{N}{\alpha} \right)^{1/3} \; \text{and} \; \lambda_\alpha \geq \gamma_\alpha \right) \\ &+ \P  \left( |\lambda_\alpha- \gamma_\alpha| \geq \frac{K \log N}{N} \left( \frac{N}{\alpha} \right)^{1/3} \; \text{and} \; \lambda_\alpha < \gamma_\alpha \right) \\ = \; &\text{A} + \text{B}\,. \end{split} \]
We consider first the term $\text{A}$.  We set
\[
%\begin{equation}\label{eq:elldef} 
\ell = \frac{K\log N}{N} \left(\frac{N}{\alpha} \right)^{1/3}\,. 
%\end{equation}
\]
{F}rom $|\lambda_\alpha - \gamma_\alpha| \geq \ell$ and $\lambda_\alpha \geq  \gamma_\alpha$ we find $\lambda_\alpha \geq  \gamma_\alpha + \ell$. This implies that $n \left( \gamma_\alpha + \ell \right) \leq n_{sc} (\gamma_\alpha)$. Therefore, by the mean value theorem, there exists $x^* \in [\gamma_\alpha ; \gamma_\alpha + \ell ]$ such that
\begin{equation}\label{eq:n-n} \begin{split} n_{sc} \left(\gamma_\alpha + \ell \right) - n \left( \gamma_\alpha + \ell  \right) &= n_{sc} (\gamma_\alpha) + \rho_{sc} (x^*) \ell -  n \left( \gamma_\alpha + \ell \right) \geq \rho_{sc} (x^*) \frac{K \log N}{N} \left( \frac{N}{\alpha} \right)^{1/3}\,. \end{split} \end{equation}
If $\gamma_\alpha + \ell > 1$, then $\lambda_\alpha \geq \gamma_\alpha + \ell$ implies $\lambda_\alpha > 1$. On the other hand, if $\gamma_\alpha + \ell \leq 1$, then $x^* \in [\gamma_\alpha, 1]$ and (\ref{eq:rho13}) implies that
\[ \rho_{sc} (x^*) \geq  \min (\rho_{sc} (\gamma_\alpha) , \rho_{sc} (1)) \geq c \left( \frac{\alpha}{N} \right)^{1/3} \,. \]
{F}rom (\ref{eq:n-n}), we conclude that
\[ \text{A} \leq \P (\lambda_\alpha > 1) + \P \left( \left| n \left(\gamma_\alpha + \ell \right) - n_{sc} \left(\gamma_\alpha + \ell \right) \right| \geq  c \frac{K \log N}{N} \right)\,. \]
The event $\lambda_\alpha > 1$ implies that $n(1) < 1/2$ and therefore that $|n(1) - n_{sc} (1)| > n_{sc} (1) - 1/2 > c$ for some $c > 0$. Theorem \ref{t:counting} implies therefore that
\[\begin{split}  \text{A} &\leq \P (|n (1) - n_{sc} (1)| > c) + \P  \left( \left| n \left(\gamma_\alpha + \ell \right) - n_{sc} \left(\gamma_\alpha + \ell \right) \right| \geq c \frac{K \log N}{N} \right) \\ &\leq (Cq)^{cq^2} \left[  \left( \frac{\log N}{N} \right)^q + \frac{1}{ K^q} \right] \,.\end{split} \]
If, say, $K \leq 10 N/(\log N)$, this implies that
\begin{equation}\label{eq:A} \text{A} \leq \frac{(Cq)^{cq^2}}{K^q}\, . \end{equation}
If, on the other hand, $K > 10N / (\log N)$, then (\ref{eq:A}) follows from (\ref{eq:lm72}) since
\[\begin{split}
 \P \bigg(|\lambda_\alpha - &\gamma_\alpha| \geq  \frac{K \log N}{N} \left( \frac{N}{\alpha} \right)^{1/3}  \text{ and } \lambda_\alpha > \gamma_\alpha \bigg)  \\ &\leq \P (\lambda_N > 8) \leq \P (1-n (8) > 1/N) \leq N \, \left(1 - \E \, n (8) \right) \leq N e^{-C N^\varepsilon} \leq \frac{(Cq)^{cq}}{N^q} \end{split} \]
 for all $q \in \NA$.

\medskip

Next, we estimate the term $\text{B}$. We distinguish two cases, $\alpha \leq \tilde{c} K \log N$ and $\alpha > \tilde{c} K \log N$, for some sufficiently small constant $\tilde{c} > 0$.

\medskip

{\it Case 1.} We assume here that $\alpha \leq \tilde{c} K \log N$. {F}rom (\ref{going}), we get $\gamma_\alpha \leq -2 + C (\alpha/N)^{1/3}$
Hence, $|\lambda_\alpha - \gamma_\alpha| \geq \ell$ and $\lambda_\alpha < \gamma_\alpha$ imply that
\[ \begin{split} \lambda_\alpha &< \gamma_\alpha - \ell \\ &\leq -2 + C \left( \frac{\alpha}{N} \right)^{1/3} - \left( \frac{K \log N}{N} \right) \left(\frac{N}{\alpha} \right)^{1/3} \\ & \leq -2 + \frac{1}{N^{2/3}} \left(C \alpha^{2/3} - \frac{K \log N}{\alpha^{1/3}} \right) \leq -2 - \frac{1}{2} \left( \frac{K \log N}{N} \right)^{2/3} \end{split} \]
if the constant $\tilde{c}$ is small enough. {F}rom Lemma \ref{l:extremalevs}, we conclude that
\begin{equation}\label{eq:B1} \P (|\lambda_\alpha - \gamma_\alpha| >\ell, \; \lambda_\alpha < \gamma_\alpha \; \text{ and }
\alpha \leq \tilde{c} K \log N) \leq \frac{(Cq)^{cq^2}}{K^q} \,
\end{equation}
for all $q \leq N^\varepsilon$.
\medskip

{\it Case 2.} Assume now that $\alpha > \tilde{c} K \log N$. {F}rom (\ref{going}), we have $\gamma_\alpha\geq -2 + c \,  (\alpha/N)^{2/3}$, for some small constant $c > 0$. We define
\[ y = -2 + \frac{c}{2} \left( \frac{\alpha}{N} \right)^{2/3} \]
and consider the cases $\gamma_\alpha - \ell > y$ and $\gamma_\alpha - \ell \leq y$ separately. Let us first assume $\gamma_\alpha-\ell > y$. Then $\lambda_\alpha < \gamma_\alpha - \ell$ implies that $n (\gamma_\alpha -\ell) \geq n_{sc} (\gamma_\alpha)$. Hence, from the mean value theorem, we find $x^* \in [\gamma_\alpha - \ell ; \gamma_\alpha] \subset [y ; \gamma_\alpha]$ such that
\[ n (\gamma_\alpha - \ell) - n_{sc} (\gamma_\alpha - \ell) = n(\gamma_\alpha - \ell) - n_{sc}(\gamma _\alpha) + \rho_{sc} (x^*) \ell \geq \rho_{sc} (x^*) \frac{K \log N}{N} \left( \frac{N}{\alpha} \right)^{1/3}\,. \]
Since $\rho_{sc}$ is increasing on $(-\infty;0]$, we have
\[ \rho_{sc}(x^*) \geq \rho_{sc} (y) \geq c \sqrt{2+y} \geq c \left( \frac{\alpha}{N} \right)^{1/3}\,. \]
{F}rom Theorem \ref{t:counting}, we conclude that
\begin{equation}\label{eq:B2} \begin{split} \P \bigg( |\lambda_\alpha - \gamma_\alpha| > \ell , \; \lambda_\alpha &< \gamma_\alpha , \; \alpha > \tilde{c} K \log N \; \text{and } \gamma_\alpha - \ell > y \bigg) \\ & \leq \P \left( n(\gamma_\alpha - \ell) - n_{sc}(\gamma_\alpha - \ell) \geq c \frac{K \log N}{N} \right) \leq \frac{(Cq)^{cq^2}}{K^q} \end{split} \end{equation}
for any $q \in \NA$. Finally, we consider the case $\gamma_\alpha - \ell \leq y$. Then $\lambda_\alpha < \gamma_\alpha - \ell$ also implies $\lambda_\alpha < y$ and therefore $n(y) -n_{sc} (\gamma_\alpha) \geq 0$. Hence, from the mean value theorem, there exists $x^* \in [y, \gamma_\alpha]$ with
\begin{equation}\label{eq:nnsc-last} n(y) - n_{sc} (y) = n (y) - n_{sc} (\gamma_\alpha) + \rho_{sc}(x^*) (\gamma_\alpha - y) \geq c \rho_{sc} (x^*) \left( \frac{\alpha}{N} \right)^{2/3} \end{equation}
by the very definition of $y$. Since $x^* > y$, we find again
\[ \rho_{sc} (x^*) \geq c \left( \frac{\alpha}{N} \right)^{1/3} \, . \]
With (\ref{eq:nnsc-last}), we conclude from Theorem \ref{t:counting} that
\begin{equation}\label{eq:B3} \begin{split} \P \bigg( | \lambda_\alpha - &\gamma_\alpha | > \ell , \; \lambda_\alpha < \gamma_\alpha , \; \alpha > \tilde{c} K \log N \text{ and } \gamma_\alpha - \ell \leq y \bigg) \\ & \leq \P \left( n (y) - n_{sc} (y) > c \left( \frac{\alpha}{N} \right) \right) \leq \P \left(  n (y) - n_{sc} (y) > c  \frac{K \log N}{N} \right) \leq \frac{(Cq)^{cq^2}}{K^q} \, . \end{split} \end{equation}
Combining (\ref{eq:B1}), (\ref{eq:B2}) and (\ref{eq:B3}), we obtain that $\text{B} \leq (Cq)^{cq^2} K^{-q}$.
\end{proof}

\appendix

\section{Large deviation estimates for quadratic forms}

In the following proposition we recall an inequality for the fluctuations of quadratic forms, this is a well known result due to Hanson and Wright. For the proof we refer to \cite[Prop. 4.5]{ESY3}, see also \cite[App. B]{EYY1} and \cite{hw71}.
\begin{proposition}
\label{p:HW}
For $j=1,\dots ,N$ let $x_{j} = \Re  x_j + i \Im  x_j$, where $\{ \Re x_{j} , \Im  x_j \}_{j=1}^N$ is a sequence of $2N$ real iid random variables, whose common distribution $\nu$ has subgaussian decay. Let $A = (a_{ij})$ be a $N\times N$ complex  matrix. Then there exist constants
$c,C > 0$ such that, for any $\de >0$
\[
\P\left(\left|\sum_{i,j=1}^Na_{ij}\left(x_i\bar x_j - \E x_i \bar x_j\right) \right|\geq \de \sqrt{\Tr A^* A} \right) \leq C e^{-c\min\{\de,\de^2\}}\,.
\]
\end{proposition}

\medskip

\subsection*{Acknowledgements.} The work of B.S. has been partially supported by ERC Starting Grant MAQD-240518.
C.C. acknowledges the support of the FIR 2013 project  ``Condensed Matter in Mathematical Physics''  (code RBFR13WAET) in the final stages of the development of the paper. A.M. acknowledges the support of the Leverhulme Trust Early Career Fellowship (ECF 2013-613). We would like to thank the referee for pointing out a mistake in a previous version of this paper.

\end{document}